\documentclass[twocolumn]{revtex4-2}
\usepackage{amsmath}
\usepackage{amssymb}
\usepackage{amsthm}
\usepackage{bm}
\usepackage[shortlabels]{enumitem}
\usepackage{tabularx}
\usepackage{graphicx}
\usepackage{tikz-cd}
\usepackage{xcolor}
\usepackage{hyperref}
\usepackage[normalem]{ulem}
\usepackage[textsize=scriptsize,backgroundcolor=white]{todonotes}

\DeclareMathOperator{\tr}{tr}

\DeclareMathOperator{\spn}{span}
\DeclareMathOperator{\ran}{ran}

\DeclareMathOperator{\Real}{Re}

\newcommand{\alg}{\mathfrak A}

\newcommand{\krn}{\tilde k}
\newcommand{\rkha}{\mathfrak A}

\makeatletter
\newcommand{\dtilde}[1]{% 
\begingroup%
    \let\macc@kerna\z@%
    \let\macc@kernb\z@%
    \let\macc@nucleus\@empty%
    \tilde{\raisebox{.35ex}{\vphantom{\ensuremath{#1}}}\smash{\tilde{#1}}}%
    \endgroup%
}
\makeatother

\newtheorem{thm}{Theorem}
\newtheorem{lem}[thm]{Lemma}
\newtheorem{prop}[thm]{Proposition}
\newtheorem{cor}[thm]{Corollary}
\theoremstyle{remark}
\newtheorem*{rk*}{Remark}

\newcommand{\crossmark}{$\bm\times$}
\newcommand{\bra}[1]{\langle #1 \rvert}
\newcommand{\ket}[1]{\lvert #1 \rangle}

\begin{document}

\title{Embedding classical dynamics in a quantum computer}
\author{Dimitrios Giannakis}
\affiliation{Department of Mathematics, Dartmouth College, Hanover, NH 03755, USA}
\affiliation{Department of Mathematics, Courant Institute of Mathematical Sciences, New York University, New York, NY 10012, USA}
\author{Abbas Ourmazd}
\affiliation{Department of Physics, University of Wisconsin--Milwaukee, Milwaukee, WI 53211, USA}

\author{Philipp Pfeffer}
\affiliation{Institut f\"ur Thermo- und Fluiddynamik, Technische Universit\"at Ilmenau, D-98684 Ilmenau, Germany}

\author{J\"org Schumacher}
\affiliation{Institut f\"ur Thermo- und Fluiddynamik, Technische Universit\"at Ilmenau, D-98684 Ilmenau, Germany}
\affiliation{Tandon School of Engineering, New York University, New York, NY 11201, USA}

\author{Joanna Slawinska}
\affiliation{Department of Computer Science, University of Helsinki, FI-00014 Helsinki, Finland}
\affiliation{Pusan National University, Busan, South Korea}
\affiliation{Center for Climate Physics, Institute for Basic Science (IBS), Busan, South Korea}

\date{\today}

\begin{abstract}
    We develop a framework for simulating measure-preserving, ergodic dynamical systems on a quantum computer. Our approach provides a new operator-theoretic representation of classical dynamics by combining ergodic theory with quantum information science. The resulting \emph{quantum embedding of classical dynamics (QECD)} enables efficient simulation of spaces of classical observables with exponentially large dimension using a quadratic number of quantum gates. The QECD framework is based on a quantum feature map that we introduce for representing classical states by density operators on a reproducing kernel Hilbert space, $\mathcal H $. Furthermore, an embedding of classical observables into self-adjoint operators on $\mathcal H$ is established, such that quantum mechanical expectation values are consistent with pointwise function evaluation. In this scheme, quantum states and observables evolve unitarily under the lifted action of Koopman evolution operators of the classical system. Moreover, by virtue of the reproducing property of $\mathcal H$, the quantum system is pointwise-consistent with the underlying classical dynamics. To achieve an exponential quantum computational advantage, we project the state of the quantum system onto a finite-rank density operator on a $2^n$-dimensional tensor product Hilbert space associated with $n$ qubits. By employing discrete Fourier-Walsh transforms of spectral functions, the evolution operator of the finite-dimensional quantum system is factorized into tensor product form, enabling implementation through an $n$-channel quantum circuit of size $O(n)$ and no interchannel communication. Furthermore, the circuit features a state preparation stage, also of size $O(n)$, and a quantum Fourier transform stage of size $O(n^2)$, which makes predictions of observables possible by measurement in the standard computational basis. We prove theoretical convergence results for these predictions in the large-qubit limit, $n\to\infty$. In light of these properties, QECD provides a consistent, exponentially scalable, stochastic simulator of the evolution of classical observables, realized through projective quantum measurement. We demonstrate the consistency of the scheme in prototypical dynamical systems involving periodic and quasiperiodic oscillators on tori. These examples include simulated quantum circuit experiments in Qiskit Aer, as well as actual experiments on the IBM Quantum System One.      
\end{abstract}

\maketitle

%========================================================================================
\section{Introduction}

Ever since a seminal paper of Feynman in 1982 \cite{Feynman82}, the problem of identifying physical systems that can faithfully and efficiently simulate large classes of other systems (performing, in Feynman's words, \emph{universal computation}) has received considerable attention. Under the operating principle that nature is fundamentally quantum mechanical, and with the realization that simulating quantum systems by classical systems is exponentially hard, much effort has been focused on the design of universal simulators of quantum systems. Such efforts are based on the axioms of quantum mechanics, with gates connected in quantum circuits performing unitary (and thus reversible) transformations of quantum states  \cite{Lloyd96,BerryEtAl07,Barnett09,Nielsen10,Preskill18,Deutsch2020}. 

Over the past decades, several numerically hard problems have been identified, for which quantum algorithms are significantly faster than their classical counterparts. A prominent example is the Grover search algorithm, which results in a quadratic speedup over classical search \cite{Grover01}. In a few cases, such as random sampling, quantum computers have solved problems that would be effectively unsolvable with present-day classical supercomputing resources, thus opening the way to quantum supremacy \cite{AruteEtAl19}. See also Ref.~\cite{ZhouEtAl20} for a discussion of the result in Ref.~\cite{AruteEtAl19}. 

Yet, at least at the level of effective theories, a great variety of phenomena are well described by classical dynamical systems, generally formulated as systems of ordinary or partial differential equations. Since simulating a quantum system by a classical system can be exponentially hard, it is natural to ask whether simulation of a \emph{classical} system by a quantum system is an exponentially ``easy'' problem, enabling a substantial increase in the complexity and range of computationally amenable classical phenomena.  

The possibility to simulate classical dynamical systems on a quantum computer has attracted growing attention, on par with research on fundamental new quantum algorithms and their practical implementation \cite{Meyer02}. Already 20 years ago, for example, Benenti et al.~\cite{BenentiEtAl01} studied the sawtooth map generating rich and complex dynamics. The implementation of an Euler method to solve systems of coupled nonlinear ordinary differential equations (ODEs) was addressed by Leyton and Osborne~\cite{LeytonOsborne08}. A framework for sequential data assimilation (filtering) of partially observed classical systems based on the Dirac--von Neumann formalism of quantum dynamics and measurement was proposed in Ref.~\cite{Giannakis19b}. The simulation of classical Hamiltonian systems using a Koopman--von Neumann approach was studied by Joseph~\cite{Joseph20}. This quantum computational framework was shown to be exponentially faster than a classical simulation when the Hamiltonian is represented by a sparse matrix. More recently, the potential of quantum computing for fluid dynamics, in particular turbulence, was explored in Refs.~\cite{BharadwajSreenivasan20,Gaitan20}. This includes, for example, transport simulators for fluid flows in which the formal analogy between the lattice Boltzmann method and Dirac equation is used \cite{MezzacapoEtAl15}. Lubasch et al.~\cite{LubaschEtAl20} took a different path inspired by the success of quantum computing in solving optimization problems, modeling the one-dimensional Burgers equation by a variational quantum computing method, made possible by its correspondence with the nonlinear Schr\"odinger equation. Quantum systems have also been employed in the modeling of classical stochastic processes, where they have shown a superior memory compression \cite{ElliotGu18,ElliotEtAl20}.  

Here, we present a procedure for simulating a classical, measure-preserving, ergodic dynamical system by means of a finite-dimensional quantum mechanical system amenable to quantum computation. Combining operator-theoretic techniques for classical dynamical systems with the theory of quantum dynamics and measurement, our framework leads to exponentially scalable quantum algorithms, enabling the simulation of classical systems with otherwise intractably high-dimensional spaces of observables. Our work thus opens a novel route to the full realization of quantum advantage in the computation of classical dynamical systems. 

Another noteworthy aspect of our approach is that it interfaces between classical \cite{BerryEtAl15,GiannakisEtAl15,Kawahara16,Giannakis19,BerryEtAl20,DasGiannakis20,KlusEtAl20,DasEtAl21} and quantum \cite{SchuldEtAl15,BiamonteEtAl17,CilibertoEtAl18,SchuldKilloran19,HavlicekEtAl19,BlankEtAl20,Schuld21} machine learning techniques  based on kernel methods. Connections with other data-driven, operator-theoretic techniques for classical dynamics \cite{DellnitzJunge99,DellnitzEtAl00,Mezic05,RowleyEtAl09,WilliamsEtAl15,KlusEtAl16,BruntonEtAl17} are also prevalent. Building on our previous work on quantum mechanical approaches data assimilation \cite{Giannakis19b}, the framework presented here offers a mathematically rigorous route to representing complex, high-dimensional classical dynamics on a quantum computer. The primary contributions of this work are as follows. 

%-----------------------------------------------------------
\begin{figure}[t]
    \centering
    \includegraphics[width=\linewidth]{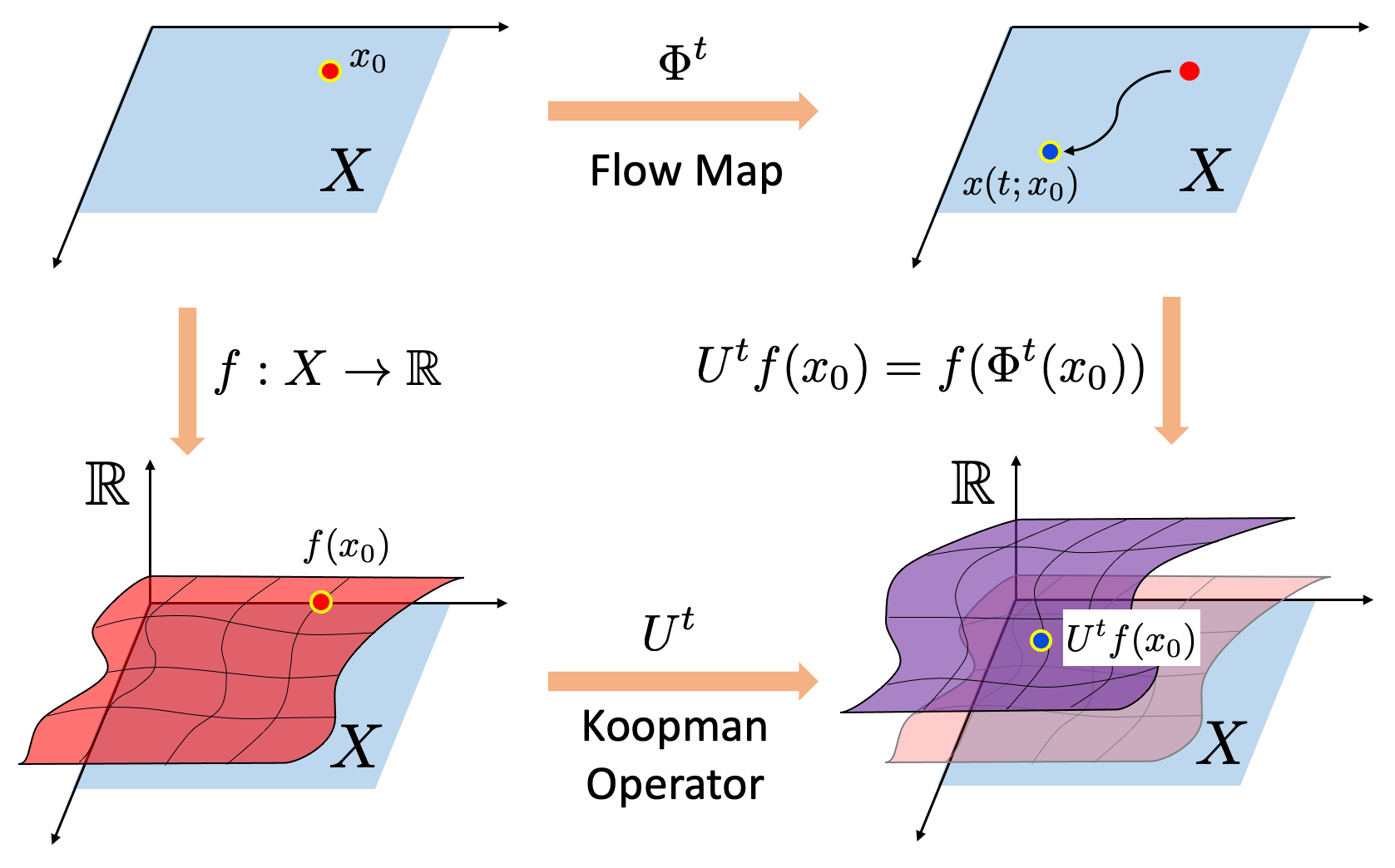}
    \caption{Schematic of the relationship between the flow map $\Phi^t$ that advances the (nonlinear) dynamical system in a  state space $X$ in time and the linear Koopman operator $U^t$ that advances observables $f$ on $X$ in an infinite-dimensional Hilbert space.}
    \label{Koopman}
\end{figure}
%-----------------------------------------------------------
\begin{enumerate}[wide] 
    \item We present a generic pipeline that casts classical dynamical systems in terms amenable to quantum computation. This approach consists of four steps. (a)~A dynamically consistent embedding of the classical state space $X$ into the state space of an infinite-dimensional quantum system with a diagonalizable Hamiltonian. (b)~Eigenspace projection of the infinite-dimensional quantum system onto a finite-dimensional system, whose dynamics are representable by composition of basic commuting unitary transformations, realizable via quantum gates. (c)~A preparation process, encoding the classical initial state in $X$ to a quantum computational state. (d)~A quantum measurement process in the standard basis of the quantum computer to yield predictions for observables. These four steps result in simulations of a $2^n$-dimensional space of classical observables using $n$ qubits and a circuit of size (i.e., number of quantum gates) $O(n^2)$ and depth $O(n)$. We call this framework for encoding a classical dynamical system in terms of a quantum computational system \emph{quantum embedding of classical dynamics (QECD)}. 

    \item We develop the principal mathematical tools employed in this construction using Koopman and transfer operator techniques \cite{Baladi00,EisnerEtAl15} and the theory of reproducing kernel Hilbert spaces (RKHSs) \cite{FerreiraMenegatto13,PaulsenRaghupathi16} and Banach function algebras on locally compact abelian groups \cite{FeichtingerEtAl07,KuznetsovaMolitorBraun12,DasGiannakis20b}. The connection between the dynamical system and the Koopman operator is illustrated in Fig.~\ref{Koopman}. Using RKHSs as the foundation to build quantum mechanical models (as opposed to the $L^2$ spaces employed in Ref.~\cite{Giannakis19}) leads to \emph{pointwise} consistency with the underlying classical dynamical system; that is, consistency for every classical initial condition, rather than in the sense of expectations over initial conditions. This result should be of independent interest in the broader context of representations of classical dynamics in terms of quantum systems, which has received significant attention \cite{Mauro02,Klein18,BondarEtAl19,Morgan20,Giannakis21b}. 
        
    \item In the particular setting of quantum computation, we establish theoretical convergence results for the finite-dimensional systems generated by the compiler, including asymptotic convergence rates in the large-qubit limit, $n\to\infty$. The time evolution of the quantum computational systems leverages discrete Fourier-Walsh techniques \cite{WelchEtAl14} to efficiently represent the Koopman operator using a circuit of size $O(n)$ and depth $O(1)$. The state preparation step, which is a major challenge in quantum computing \cite{Benedetti2019,Markovic2020}, is also carried with a circuit of size $O(n)$ and depth $O(1)$. In particular, we take advantage of the fact that every quantum state associated with a classical initial state in $X$ can be reached to any desired accuracy by efficient unitary transformations applied to a uniform-superposition state constructed using Hadamard gates. Meanwhile, the measurement process employs the quantum Fourier transform (QFT) to perform efficient approximate diagonalization of observables with a circuit of size $O(n^2)$ and depth $O(n)$ \cite{Coppersmith94,MooreNilsson01}.

    \item We demonstrate the QECD framework in simple, analytically solvable examples of classical dynamics, so that all steps of the procedure are fully reproducible. Specifically, we use QECD to simulate the evolution of observables of periodic and quasiperiodic dynamical systems in a one- and two-dimensional phase space, respectively. We employ the gate-based, universal quantum computing toolkit Qiskit Aer~\cite{Qiskit,Qiskit20}, using up to $n=8$ qubits. Results from simulated quantum circuit experiments (see Figs.~\ref{dim1} and~\ref{dim2}) are found to be in good agreement with the true classical dynamics. In addition, we perform experiments for the periodic system on an actual quantum computer, the IBM Quantum System One, demonstrating the ability of QECD to simulate a classical system on a noisy intermediate-scale quantum (NISQ) device. 
\end{enumerate}

We note that the two-dimensional quasiperiodic dynamics in our examples can be straightforwardly extended to higher dimensions, where the dynamics becomes increasingly indistinguishable from a chaotic system. For quasiperiodic dynamics, no interchannel communication is necessary. Circuits of higher complexity that create inter-qubit entanglement may need to be explored for treatment of chaotic dynamics.

The outline of the paper is as follows. First, in Sec.~\ref{secOverview} we give a high-level description of the methodological framework underlying the quantum embedding.  In Sec.~\ref{secClassical}, we introduce the class of dynamical systems under study, along with the corresponding RKHSs of classical observables. This is followed in Secs.~\ref{secQuantumMechanical}--\ref{secPreparation} by a detailed description of the construction of the QECD for this class of systems. In Secs.~\ref{secExamples} and~\ref{secIBM}, we present our results from simulated and actual quantum computation experiments, respectively. Our primary conclusions are summarized in Sec.~\ref{secConclusions}. The paper contains appendices on RKHS-based quantum mechanical representations of classical systems (Appendix~\ref{appQuantumRep}), Fourier-Walsh factorization of the Koopman generator (Appendix~\ref{appWalsh}), and QFT-based approximate diagonalization of observables (Appendix~\ref{appQFT}). In addition, we provide an overview of elements of Koopman operator theory related to this work and associated numerical techniques as supplementary material (SM) \cite{Suppl}.

%========================================================================================
\section{\label{secOverview} A route to quantum embedding of classical dynamics}
\label{route}

We begin by describing the main components of the QECD framework for representing classical dynamics on a quantum computer. Figure~\ref{Scheme} schematically summarizes the successive levels used in the procedure, passing through classical, classical statistical, infinite-dimensional quantum mechanical, finite-dimensional quantum mechanical (referred to as matrix mechanical), and quantum computational levels. This diagram juxtaposes the steps for states  and observables side-by-side for easy comparison. In the following subsections, we discuss the individual horizontal and vertical connections (which are maps) on each of the five levels of this diagram.

%-----------------------------------------------------------
\begin{figure*}[t]
    \includegraphics[scale=0.37]{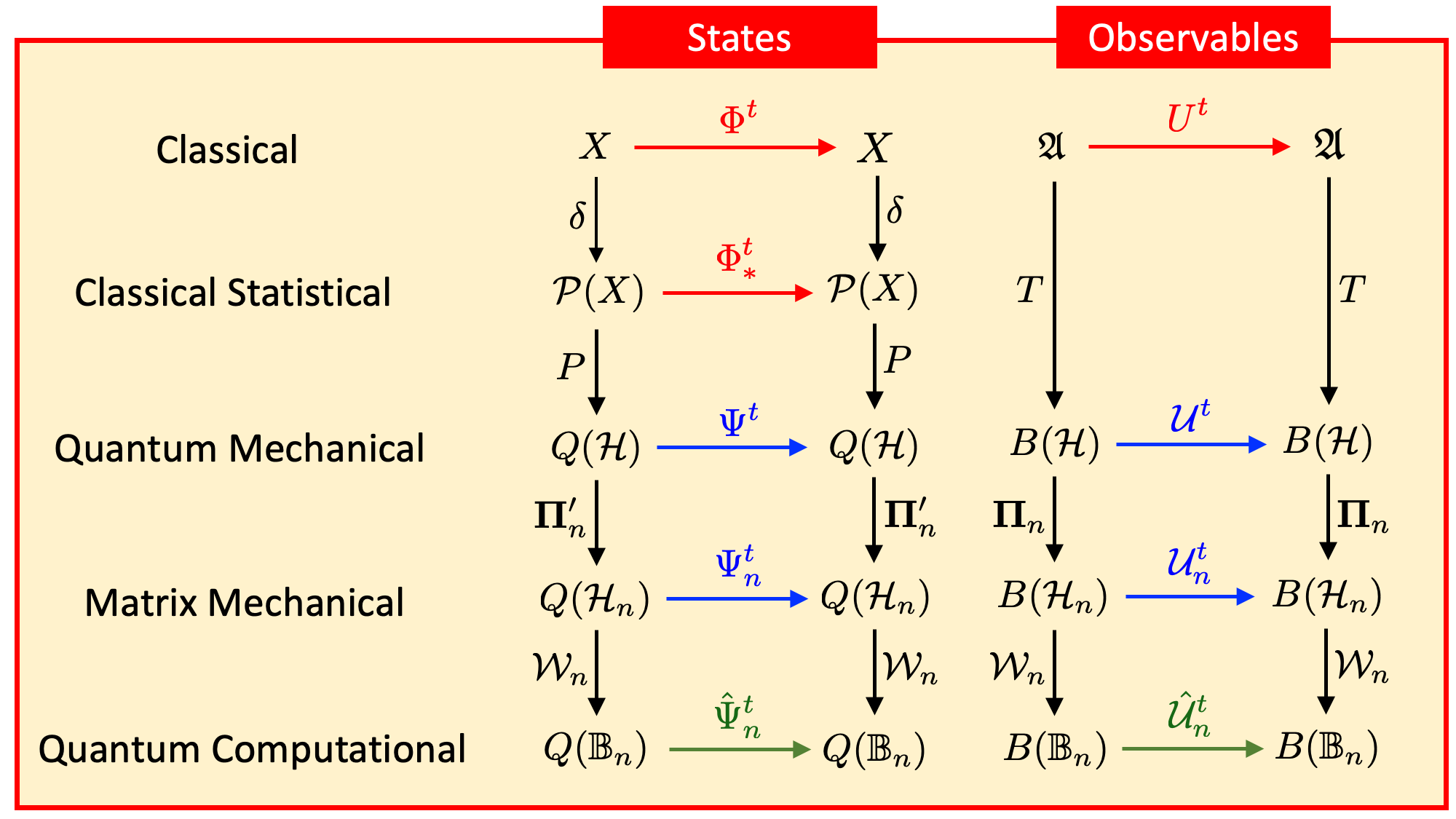}    
    \caption{Schematic representation of the QECD framework applied to states and observables of a classical dynamical system in five successive levels, leading to an $n$-qubit quantum computational system. These are the classical, classical statistical, quantum mechanical, matrix mechanical, and quantum computational levels. The horizontal arrows from top to bottom in the left- and right-hand columns represent the time evolution maps of states and observables, respectively. These are the flow map $\Phi^t$ on the classical state space $X$, the transfer operator $\Phi^t_{\ast}$ on the space of probability measures $\mathcal{P}(X)$, and the Koopman operator $U^t$ on the algebra of classical observables $\alg \subseteq C(X)$. They are followed by the unitary evolution map $\Psi^t$ and the Heisenberg operator $\mathcal{U}^t$ on the space of density operators $Q(\mathcal H)$ and bounded linear operators $B(\mathcal H)$, respectively, on the reproducing kernel Hilbert space $\mathcal H$. The maps at the matrix mechanical level, $\Psi^t_n$ and $\mathcal U^t_n$, are finite-rank projections of $\Psi^t$ and $\mathcal U^t$, respectively, acting on operators on $2^n$-dimensional subspaces $\mathcal H_n$ of $\mathcal H$.  The corresponding maps $\hat\Psi^t_n$ and $\hat{\mathcal U}^t_n$, respectively, at the quantum computational level act on operators on the $2^n$-dimensional tensor product Hilbert space $\mathbb B_n$, which forms the basis of an $n$-qubit quantum computer. The vertical arrows correspond to maps that translate states (left-hand column) and observables (right-hand column) to the next representation level. Under the combined action of these maps, a classical state $x\in X$ is mapped to an $n$-qubit density matrix $\hat\rho_{x,n}\in Q(\mathbb B_n)$, and a classical observable $f\in \alg$ is mapped to a self-adjoint operator $\hat S_n\in B(\mathbb B_n)$. A loop of arrows represents a commutative diagram.}
    \label{Scheme}
\end{figure*}
%-----------------------------------------------------------
\subsection{\label{secOverviewClassical}Classical and classical statistical levels}

Consider a classical dynamical system on a compact metric space $X$, described by a dynamical flow map 
\begin{equation}
\Phi^t : X \to X  \quad\mbox{with}\quad t \in \mathbb R,
\end{equation}
as indicated by a horizontal arrow in the left-hand column of Fig.~\ref{Scheme}. The classical state space $X$ is embedded into the space of Borel probability measures $ \mathcal P( X )$ (i.e., the classical statistical space) by means of the map $\delta$ sending $x \in X$ to the Dirac measure $\delta_x \in \mathcal P(X)$ supported at $x$. The dynamics acts naturally on the classical statistical space by the pushforward map on measures, 
\begin{equation} 
    \label{eqPushforwardMeas}
    \Phi^t_* : \mathcal P( X ) \to \mathcal P( X ) \quad\mbox{with}\quad \Phi^t_*(\nu) = \nu \circ \Phi^{-t},  
\end{equation}
also known as the transfer or Perron-Frobenius operator \cite{Baladi00,EisnerEtAl15}. The map $\delta$ has the equivariance property $ \Phi^t_* \circ \delta = \delta \circ \Phi^t $, represented by the top loop in the the left-hand column in Fig.~\ref{Scheme}.

Associated with the dynamical system are spaces of classical observables, which we take here to be spaces of complex-valued functions on $X$. A natural example is the space of continuous functions, denoted as $C(X)$, which also forms an (abelian) algebra with respect to the pointwise product of functions. The Koopman operator \cite{Koopman31,KoopmanVonNeumann32}, $U^t$, acts on observables in $C(X)$ by composition with the flow map, i.e., 
\begin{equation}
U^t: C(X) \to C(X) \quad\mbox{with}\quad  U^t f = f \circ \Phi^t;
\label{Koopmandef}
\end{equation}
see also Fig.~\ref{Koopman}. The horizontal arrow in the first line of the right-hand column in Fig.~\ref{Scheme} represents the action of the Koopman operator on a subalgebra $\rkha \subseteq C(X) $ that will be described in Sec.~\ref{secOverviewQuantum} below.

In this context, a simulator of the system can be described as a procedure which takes as an input an observable $ f \in C(X)$ and an initial condition $ x \in X $, and produces as an output a function $\hat f^{(t)}(x)$ approximating the evolution $ f( \Phi^t(x) ) $ of the observable under the dynamics. For instance, if $ \Phi^t$ is the flow generated by a system of ODEs $ \dot x = \vec V (x)$ on $ \mathcal X = \mathbb R^m $, and $X \subset \mathcal X$ is an invariant subset of this flow (e.g., an attractor), a standard simulation approach is to construct a finite-difference approximation $ \hat \Phi^t : \mathcal X \to \mathcal X $ of the dynamical flow based on a timestep $\Delta t$ (using interpolation to generate a continuous-time trajectory), and obtain $ \hat f^{(t)}(x) = f(\hat \Phi^t(x))$ by evaluating  the observable of interest $f$ on the approximate trajectory. The scheme then converges in a limit of $ \Delta t \to 0 $ by standard results in ODE theory and numerical analysis for observables $f$ of sufficient regularity. 

From an observable-centric standpoint, a simulator of the system corresponds to a linear operator $\hat U^t$ approximating the Koopman operator $U^t$, giving $\hat f^{(t)}(x) = \hat U^t f(x)$. For instance, the ODE-based approximation just mentioned can be described in this way for $\hat U^t f = f \circ \hat \Phi^t$, but note that not \emph{every} approximation of $U^t$ has to be of the form of a composition operator by a flow. Indeed, ``lifting'' the task of simulation from states to (classical) observables opens the possibility of using new approximation techniques, which in some cases can resolve computational bottlenecks, e.g., due to high dimensionality ($m$) of the ambient state space $\mathcal X$ \cite{WangEtAl20}. Invariably, every practical simulator $\hat U^t$ is restricted to act on a space of observables of finite dimension, $N$ (e.g., a subspace of $C(X)$ or $L^2$). In general, the computation cost of acting with $\hat U^t$ on elements of this space scales as $N^2$, but can be reduced to $O(N)$ if $\hat U^t $ is efficiently represented by a diagonal matrix. The evaluation cost of observables, which corresponds to summation of an $N$-term basis expansion such as a Fourier series, is typically $O(N)$. 

In what follows, rather than employing an approximation $\hat U^t$ acting on classical observables, our goal is to simulate the action of $ U^t $ using a \emph{quantum mechanical} system. As we will see, this can be achieved at a logarithmic cost of elementary quantum operations (gates); specifically, QECD allows simulation of spaces of classical observables of dimension $N=2^n$ using $O(n^2)$ gates.

\subsection{\label{secOverviewQuantum}Quantum computational representation}

The QECD framework effecting the representation of the classical system by a quantum mechanical system employs the following key spaces: 
\begin{enumerate}
    \item The classical  state space $X$.
    \item A Banach $^*$-algebra $\alg \subseteq C(X)$ of classical observables.
    \item An infinite-dimensional RKHS $\mathcal H \subset \rkha $.
    \item A finite-dimensional Hilbert space $\mathbb B_n$ associated with the quantum computer.
\end{enumerate}
The Hilbert spaces $\mathcal H$ and $\mathbb B_n$ have corresponding (non-abelian) algebras of bounded linear operators, $B(\mathcal H)$ and $B(\mathbb B_n)$, respectively, acting as quantum mechanical observables. Moreover, states on these algebras are represented by density operators, i.e., trace-class, positive operators of unit trace, acting on the respective Hilbert space. We denote the spaces of density operators on $\mathcal H$ and $\mathbb B_n$ by  $Q(\mathcal H)$ and $Q(\mathbb B_n)$, respectively. Below, $n$ represents the number of qubits, thus the dimension of $\mathbb B_n $ is $2^n$. 

The spaces of classical states and observables $X$ and $\alg$ are mapped into the spaces of quantum states and observables $Q(\mathbb B_n)$ and $B(\mathbb B_n)$, respectively; see Fig.~\ref{Scheme}. The following maps on states (left-hand column) and observables (right-hand column) transform the classical system into a quantum-mechanical one on $\mathbb B_n $: 
\begin{itemize}[wide]
    \item We construct a map $ \hat{\mathcal F}_n  : X \to Q( \mathbb B_n )$ from classical states (points) in $X$ to quantum states on $\mathbb B_n$. By analogy with the RKHS-valued feature maps in machine learning \cite{ScholkopfEtAl98}, $\hat{\mathcal F}_n$ will be referred to as a \emph{quantum feature map}. To arrive at $\hat{\mathcal F}_n $, the classical statistical space $ \mathcal P(X)$ is first embedded into the quantum mechanical state space $Q(\mathcal H)$ associated with $\mathcal H$ through a map $P : \mathcal P(X) \to Q(\mathcal H)$ (see~\eqref{eqP} below). The composite map $ \mathcal F := P \circ \delta $ thus describes a one-to-one quantum feature map from $X$ into $ Q(\mathcal H)$. Next, the infinite-dimensional space $Q(\mathcal H)$ is projected onto a finite-dimensional quantum state space $Q(\mathcal H_n)$ associated with a $2^n$-dimensional subspace $\mathcal H_n \subset \mathcal H$ by means of a map $\bm \Pi'_n : Q(\mathcal H) \to Q(\mathcal H_n)$. We refer to this level of description as {\em matrix mechanical} since all quantum states and observables are finite-rank operators, represented by  $2^n \times 2^n $ matrices. To arrive at the {\em quantum computational} state space, we finally apply a unitary $ \mathcal W_n : Q(\mathcal H_n) \to Q(\mathbb B_n)$, so that the full quantum feature map from $X$ to $Q(\mathbb B_n)$ takes the form $ \hat{\mathcal F}_n = \mathcal W_n \circ \bm \Pi'_n \circ P \circ \delta $. 

    \item We construct a linear map $ \hat T_n : \alg \to B( \mathbb B_n ) $ from classical observables in $\alg$ to quantum mechanical observables in $ B(\mathbb B_n)$. This map takes the form $\hat T_n= \mathcal W_n \circ \bm \Pi_n \circ T$, where $\bm \Pi_n : B(\mathcal H) \to B(\mathcal H_n)$ is a projection, so that $\hat T_n$ yields a quantum computational representation of classical observables passing through intermediate quantum mechanical and matrix mechanical representations. Here, $T : \alg \to B(\mathcal H) $ is one-to-one on real-valued functions in $\alg$, and $Tf$ is self-adjoint whenever $ f $ is real.
\end{itemize}

Next, we describe the maps governing the temporal evolution of states and observables, represented by horizontal arrows in Fig.~\ref{Scheme}:
\begin{itemize}[wide]
    \item At the quantum mechanical level, states in $Q(\mathcal H)$ evolve under the operator $ \Psi^t $ (horizontal arrow in the left-hand column) induced by a unitary Koopman operator $U^t=e^{tV}$ on $\mathcal H$. This evolution is generated by a skew-adjoint generator $ V: D(V) \to \mathcal H $, defined on a dense subspace $D(V) \subset \mathcal H $ and possessing a countable spectrum of eigenfrequencies. 
     \item The generator $V$ is mapped to a self-adjoint Hamiltonian $H_n: \mathbb B_n\to \mathbb B_n$  given by $ H_n = \mathcal W_n \bm \Pi_n V /  i$. This Hamiltonian is decomposable as a sum $  H_n = \sum_j G_j $ of mutually-commuting operators $G_j \in B(\mathbb B_n)$, each of which is of pure tensor product form, $G_j = \bigotimes_{i=1}^n G_{ij}$. The latter property enables quantum parallelism in the unitary evolution $ \hat \Psi^t_n : Q(\mathbb B_n) \to Q(\mathbb B_n)$ at the quantum computational level generated by $H_n$ (see horizontal arrow at the bottom of the left-hand column of Fig.~\ref{Scheme}). One of our main results is that $\hat \Psi^t$ can be implemented via a quantum circuit of size $O(n)$ and no interchannel communication (see Figs.~\ref{dim1} and~\ref{dim2}).   
    \item The horizontal arrow at the quantum mechanical level represents the action of the Heisenberg evolution operator $ \mathcal U^t : B(\mathcal H) \to B(\mathcal H)$. Under the assumption that the RKHS $\mathcal H$ is invariant under the Koopman operator, $ \mathcal U^t$ acts on $B(\mathcal H)$ by conjugation with $U^t$, i.e., $ \mathcal U^t A = U^t A U^{t*}$.
    \item The corresponding Heisenberg evolution operator at the quantum computational level, $ \hat{\mathcal U}^t_n : B(\mathbb B_n) \to B(\mathbb B_n)$, acting on quantum mechanical observables on the Hilbert space $ \mathbb B_n$, is represented by the horizontal arrow at the bottom of the right-hand column. This operator is obtained by projection of $\mathcal U^t$, viz.\ $\hat{\mathcal U}_n^t = \mathcal W_n \bm\Pi_n\mathcal U^t $. 
\end{itemize}

Given a classical initial condition $x \in X$, the quantum computational system constructed by QECD makes probabilistic predictions $ \hat f^{(t)}_n(x) $ of $ f( \Phi^t(x))$ through quantum mechanical measurement of the projection-valued measure (PVM) \cite{Barnett09,Peres02} associated with the quantum register on the quantum state $ \hat \rho^{(t)}_{x,n} := \hat \Psi^t_n(\hat \rho_{x,n})$, where $\hat \rho_{x,n} = \hat{\mathcal F}_n(x)$. The state $\hat \rho_{x,n}$ is prepared by means of a circuit of size $O(n)$, which is applied to the standard initial state vector of the quantum computer. Furthermore, the measurement step is effected by performing a rotation $ \hat \rho^{(t)}_{x,n} \mapsto \tilde \rho^{(t)}_{x,n} $ by a QFT, which is implementable via a circuit of size $O(n^2)$. An ensemble of such measurements then approximates the quantum mechanical expectation value 
\begin{equation}
    \label{eqQuantumEvolution}
    \langle \hat T_n f \rangle_{\hat \rho^{(t)}_{x,n}} := f^{(t)}_n(x). 
\end{equation}
The function $x \mapsto f^{(t)}_n(x)$ converges in turn uniformly to the true classical evolution, i.e., $U^t f(x)$, in the large-qubit limit, $n\to \infty$. We will return to these points in a more detailed discussion in Secs.~\ref{secMatrixMechanical}--\ref{secPreparation}.

In summary, the key distinguishing aspects of QECD are as follows:
\begin{enumerate}[wide]
    \item \emph{Dynamical consistency.} The predictions made by the quantum quantum computational system via~\eqref{eqQuantumEvolution} converge to the true classical evolution as the number of qubits $n$ increases. In particular, since $ \dim \mathbb B_n = 2^n$, the convergence is exponentially fast in $n$.  
    \item \emph{Quantum efficiency.} The full circuit implementation of the scheme, including state preparation, dynamical evolution, and measurement, requires a circuit of size $O(n^2)$ and depth $O(n)$. Since, as just mentioned, the dimension of $\mathbb B_n$ increases exponentially with $n$, the quantum computational system constructed by QECD has an exponential advantage over classical simulators of the underlying $2^n$-dimensional subspace of classical observables. 
    \item \emph{State preparation.} The quantum computational state $\hat \rho_{x,n}$ corresponding to classical state $x$ is prepared by passing the standard initial state vector of the quantum computer through a circuit of size $O(n)$ and depth $O(1)$. This overcomes the expensive (potentially exponential) state preparation problem affecting many quantum computational algorithms. 

    \item \emph{Measurement process.} The process of querying the system to obtain predictions is a standard projective measurement of the quantum register. Importantly, no quantum state tomography or auxiliary classical computation is needed to retrieve the relevant information.
\end{enumerate}

In the ensuing sections, we lay out the properties of the classical system under study (Sec.~\ref{secClassical}), and describe the conversion to the quantum computational system using QECD (Secs.~\ref{secQuantumMechanical}--\ref{secPreparation}). 

%========================================================================================
\section{\label{secClassical}Classical dynamics and observables}

\subsection{\label{secClassicalDynamics}Dynamical system}

We focus on the class of continuous, measure-preserving, ergodic flows with a pure point spectrum generated by finitely many eigenfrequencies and continuous corresponding eigenfunctions. Every such system is topologically conjugate (for our purposes, equivalent) to an ergodic rotation on a $d$-dimensional torus, so we will set $ X = \mathbb T^d$ without loss of generality. Using the notation $ x = ( \theta^1, \ldots, \theta^d ) $ to represent a point $ x \in \mathbb T^d$, where $ \theta^j \in [ 0, 2 \pi )$ are canonical angle coordinates, the dynamics is described by the flow map
\begin{equation}
    \label{eqTorusRotation}
    \Phi^t( x ) = ( \theta^1 + \alpha_1 t, \ldots, \theta^d + \alpha_d t ) \mod 2 \pi,
\end{equation}
where $ \alpha_1, \ldots, \alpha_d $ are positive,  rationally independent (incommensurate) frequency parameters. This dynamical system is also known as a linear flow on the $d$-torus, but note that $\mathbb T^d$ is not a linear space. In dimension $ d > 1 $, the orbits $ \Phi^t(x)$ of the dynamics do not close by incommensurability of the $ \alpha_j $, each forming a dense subset of the torus (i.e., a given orbit passes by any point in $\mathbb T^d$ at an arbitrarily small distance). The case $ d = 2 $ is shown for two choices of $\alpha_j$ in Fig.~\ref{Torus}, illustrating the difference between ergodic and non-ergodic dynamics. In dimension $d=1$, the flow map corresponds to a harmonic oscillator on the circle, $ \mathbb T^1 = S^1$, where each orbit is periodic and samples the whole space.

It is important to note that if the dynamical system is not presented in the form of a torus rotation, then standard constructions from ergodic theory may be used to transform it into the form in~\eqref{eqTorusRotation}. These constructions are based entirely on spectral objects (i.e., eigenfunctions and eigenfrequencies) associated with the Koopman operator of the system. See Sec.~I~D in the SM for further details \cite{Suppl}. The same constructions allow one to treat the case where $X$ is a periodic or quasiperiodic attractor of a  dynamical flow $\Phi^t : \mathcal X \to \mathcal X $ on a higher-dimensional space $\mathcal X \supseteq X$. By virtue of these facts, the quantum mechanical framework described in this paper can readily handle simulations of observables of general measure-preserving, ergodic flows with pure point spectrum. Relevant examples include ODE models on $\mathcal X = \mathbb R^m$ with quasiperiodic attractors \cite{GrebogiEtAl85}, as well as PDE models where $\mathcal X$ is an infinite-dimensional function space. The latter class includes many pattern-forming physical systems such as thermal convection flows \cite{EckeEtAl91}, plasmas \cite{WeixingEtAl93}, and reaction-diffusion systems \cite{KalogirouEtAl15} in moderate-forcing regimes.

%-----------------------------------------------------------
\begin{figure}[t]
    \centering
    \includegraphics[width=\linewidth]{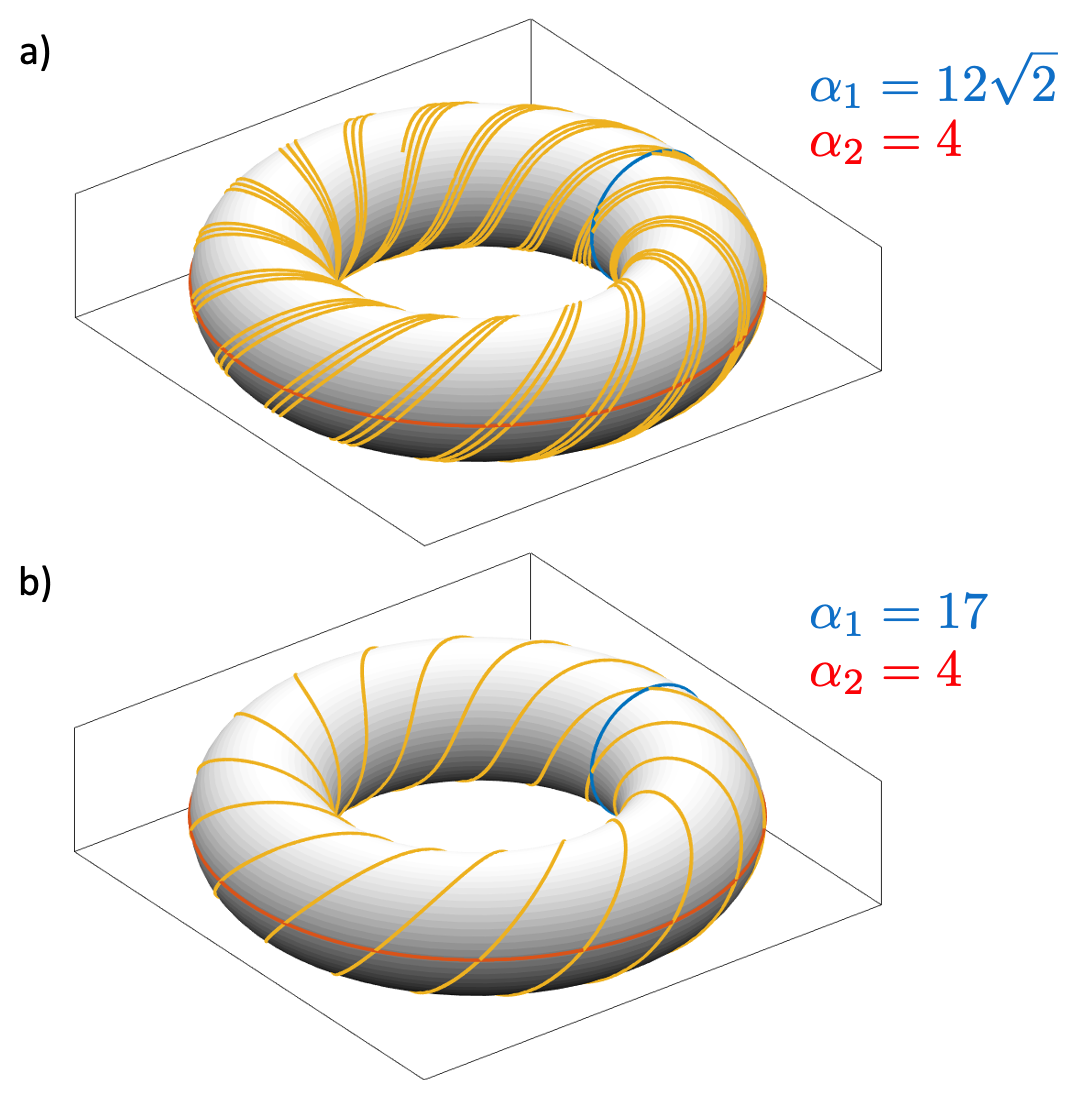}
    \caption{Ergodic (a) and non-ergodic (b) linear flows on the two-dimensional torus $\mathbb T^2$. In (a) the ratio of the frequency parameters $\alpha_1/\alpha_2$ is irrational, and the trajectory starts to fill the torus surface. In (b) the  ratio of the frequencies is rational, and the trajectory is closed. The corresponding frequency parameters $\alpha_1$ and $\alpha_2$ are given to the right of each figure.}
    \label{Torus}
\end{figure}
%-----------------------------------------------------------

At any dimension $ d $, the flow in~\eqref{eqTorusRotation} is measure-preserving and ergodic for a probability measure $\mu$ given by the normalized Haar measure. The dynamics of classical observables $ f : X \to \mathbb C$ is governed by the Koopman operator $U^t$, which is a \emph{linear} operator, acting by composition with the dynamical flow in accordance with~\eqref{Koopmandef} \cite{Baladi00,BudisicEtAl12,EisnerEtAl15}. The Koopman operator acts as an isometry on the Banach space of continuous functions on $X$, i.e., $ \lVert U^t f \rVert_{C(X)} = \lVert f \rVert_{C(X)}$, where $ \lVert f \rVert_{C(X)} = \max_{x\in X} \lvert f(x) \rvert$ is the uniform norm. In addition, $U^t$ lifts to a unitary operator on the Hilbert space $L^2(\mu)$ associated with the invariant measure. That is, using  $ \langle f, g \rangle_{L^2(\mu)} = \int_X f^* g \, d\mu$ to denote the $L^2(\mu)$ inner product, we have $ \langle U^t f, U^t g \rangle_{L^2(\mu)} = \langle f, g \rangle_{L^2(\mu)} $ for all $ f, g \in L^2(\mu)$, which implies, in conjunction with the invertibility of $\Phi^t$, that 
\begin{displaymath}
    U^{t*} = {U^t}^{-1}. 
\end{displaymath}
Here, $U^{t*}$ denotes the operator adjoint, which is also frequently denoted as $(U^{t})^\dag$. The collection $ \{ U^t : L^2(\mu) \to L^2(\mu) \}_{t \in \mathbb R}$ then forms a strongly continuous unitary group under composition of operators \footnote{This implies (i) $U^s \circ U^t = U^{s+t}$, (ii) $(U^t)^{-1} = U^{-t}$, and (iii) $ \lim_{t\to 0} U^t f = f$ for all $s,t \in \mathbb R $ and $ f \in L^2(\mu)$.}. See again Fig.~\ref{Koopman}. 

By Stone's theorem on one-parameter unitary evolution groups \cite{Stone32}, the Koopman group on $L^2(\mu)$ has a skew-adjoint infinitesimal generator, i.e., an operator $V : D(V) \to L^2(\mu)$ defined on a dense subspace $D(V) \subset L^2(\mu)$ satisfying 
\begin{equation}
V^* = - V \quad\mbox{and}\quad V f = \lim_{t\to0} \frac{U^t f - f}{t}, 
\label{eqGenerator}
\end{equation}
for all $ f \in D(V)$. The generator gives the Koopman operator at any time $t $ by exponentiation, 
\begin{equation}
U^t = e^{ tV}.
\label{eqExpKoopman}
\end{equation}
Modulo multiplication by $ 1/i $ to render it self-adjoint, it plays an analogous role to a quantum mechanical Hamiltonian generating the unitary Heisenberg evolution operators.        

As already noted, the torus rotation in~\eqref{eqTorusRotation} is a canonical representative of a class of continuous-time continuous dynamical systems on topological spaces with quasiperiodic dynamics generated by finitely many basic frequencies. This means that every such system can be transformed into an ergodic torus rotation of a suitable dimension by a homeomorphism (continuous, invertible map with continuous inverse). By specializing to this class of systems (as opposed to a more general measure-preserving, ergodic flow), we gain two important properties: 
\begin{enumerate}[wide]
    \item The dynamics has no mixing (chaotic) component. This implies that the spectrum of the Koopman operator for this system acting on $L^2(\mu) $, or a suitable RKHS as in what follows,  is of ``pure point'' type, obviating complications arising from the presence of continuous spectrum as would be the case under mixing dynamics. 
    \item The state space $X$ is a smooth, closed manifold with the structure of a connected, abelian Lie group. The abelian group structure, in particular, renders this system amenable to analysis with Fourier analytic tools. 
\end{enumerate}

Below, we use a $d$-dimensional vector $ j = ( j_1, \ldots, j_d ) \in \mathbb Z^d$ to represent a generic multi-index, and
    \begin{equation}
        \phi_{j}(x) = \prod_{m=1}^d\varphi_{j_m}(\theta^m) \quad\mbox{with}\quad \varphi_{l}(\theta) = e^{i l \theta },
        \label{FourierF}
    \end{equation}
to represent the Fourier functions on $\mathbb T^d$. In Sec.~\ref{secConclusions}, we will discuss possible avenues for extending the framework presented here to other classes of dynamical systems, such as mixing dynamical systems with continuous spectra of the Koopman operators.

\subsection{\label{secRKHA}Algebra of observables}

According to the scheme described in Sec.~\ref{secOverviewQuantum}, we perform quantum conversion of an (abelian) algebra $\alg$ of classical observables, i.e., a space of complex-valued functions on $X$ which is closed under the pointwise product of functions. We construct $\alg$ such that it is a subalgebra of $C(X)$ with additional (here, $C^\infty$) regularity and RKHS structure. This structure is induced by a smooth, positive-definite kernel function $ \krn : X \times X \to \mathbb R$, which has the following properties for every point $x \in X$ and function $ f \in \rkha$: 
\begin{enumerate}
    \item The kernel section $\krn_x := \krn (x, \cdot )$ lies in $\rkha$. 
    \item Pointwise evaluation, $ x \mapsto f(x)$, is continuous, and satisfies
    \begin{equation}
        \label{eqReproducing}
        f(x) = \langle \krn_x, f \rangle_\rkha, 
    \end{equation}
    where $\langle\cdot,\cdot\rangle_\rkha$ is the inner product of $\rkha$.  
\end{enumerate}
Equation~\eqref{eqReproducing} is known as the \emph{reproducing property}, and underlies the many useful properties of RKHSs for tasks such as function approximation and learning. Note, in particular, that $L^2$ spaces, which are more commonly employed in Koopman operator theory and numerical techniques (see Sec.~\ref{secClassicalDynamics}), do not have a property analogous to~\eqref{eqReproducing}. In fact, pointwise evaluation is not even defined for the $L^2(\mu)$ Hilbert space on $\mathbb T^d$.  See Refs.~\cite{CuckerSmale01,FerreiraMenegatto13,PaulsenRaghupathi16} for detailed expositions on RKHS theory.    

Our construction of $\rkha$ follows Ref.~\cite{DasGiannakis20b}. We begin by setting parameters $ p \in (0,1) $ and $ \tau > 0$, and defining the map $ \lvert \cdot \rvert_p : \mathbb Z^d \to \mathbb R_+$,
\begin{displaymath}
    \lvert j \rvert_p := \lvert j_1 \rvert^p + \ldots + \lvert j_d \rvert^p,
\end{displaymath}
and the functions $ \psi_{j} \in C(X)$,  
\begin{displaymath}
    \psi_{j} := e^{-\tau \lvert j \rvert_p  / 2} \phi_{j} \quad \text{with} \quad j \in \mathbb Z^d. 
\end{displaymath}
We then define a kernel $ \krn : X \times X \to \mathbb R_+$ via the series 
\begin{equation}
    \label{eqKTorus}
    \krn(x,x') = \sum_{j \in \mathbb Z^d} \psi^*_{j}(x) \psi_{j}(x'),  
\end{equation}
where the sum over $ j $ converges uniformly on $ X \times X$ to a smooth function. Intuitively, $ \tau $ can be thought of as a locality parameter for the kernel, meaning that as $ \tau $ decreases $\krn (x,x')$ becomes increasingly concentrated near $x = x'$,  approaching a $\delta$-function as $ \tau \to 0$. 

An important property of the kernel that holds for any $ \tau > 0 $ is that it is translation-invariant on the abelian group $ X= \mathbb T^d$. That is, using additive notation to represent the binary group operation on $ X$, we have 
\begin{equation}
    \krn(x+y,x'+y) = \krn(x,x'), \quad \forall x,x',y \in X.
    \label{eqKTrans}
\end{equation}
In particular, setting $ y = \Phi^t(e)$, where $e$ is the identity element of $X$, and noticing that the dynamical flow from~\eqref{eqTorusRotation} satisfies $ \Phi^t(x) = x + \Phi^t(e)$, we deduce the dynamical invariance property
\begin{displaymath}
    \krn(\Phi^t(x),\Phi^t(x') ) = \krn(x,x'), \quad \forall x,x' \in X, \quad \forall t \in \mathbb R.
\end{displaymath}

In Ref.~\cite{DasGiannakis20b} it was shown that for every $ p > 0 $ and $ \tau > 0 $, the kernel $\krn$ in~\eqref{eqKTorus} is a strictly positive-definite kernel on $X$, so it induces an RKHS, $\rkha $, which is a dense subspace of $C(X)$. One can verify that the collection $ \{ \psi_{j} : j \in \mathbb Z^d \} $ forms an orthonormal basis of $\rkha$, consisting of scaled Fourier functions, so every observable $f \in \rkha$ admits the expansion
\begin{displaymath}
    f = \sum_{j \in \mathbb Z^d} \tilde f_{j} \psi_{j} = \sum_{j \in \mathbb Z^d} \tilde f_{j} e^{-\tau \lvert j \rvert_p / 2 } \phi_{j},
\end{displaymath}
where the sum over $ j $ converges in $\rkha$ norm. The above manifests the fact that $\rkha$ contains continuous functions with Fourier coefficients decaying faster than any polynomial, implying in turn that every element of $\rkha$ is a smooth function in $ C^\infty(X)$. 

It can also be shown that the RKHS induced by $\krn$  acquires an important special property which is not shared by generic RKHSs---namely, it becomes an abelian, unital, Banach $^*$-algebra under pointwise multiplication of functions. We list the defining properties for completeness in Appendix~\ref{BanachStar}. In Ref.~\cite{DasGiannakis20b}, the space $\rkha$ was referred to as a {\em reproducing kernel Hilbert algebra (RKHA)} as it enjoys the properties of both RKHSs and Banach algebras. In particular, a distinguishing aspect of  $\rkha$ is that it simultaneously has Hilbert space structure (as $L^2(\mu)$) and Banach $^*$-algebra structure (as $C(X)$), while also allowing pointwise evaluation by continuous functionals, (i.e., the reproducing property in~\eqref{eqReproducing}). The RKHAs associated with the family of kernels in~\eqref{eqKTorus} are examples of harmonic Hilbert spaces on locally compact abelian groups \cite{FeichtingerEtAl07}, and are also closely related (by Fourier transforms) to weighted convolution algebras \cite{KuznetsovaMolitorBraun12} on the dual group $\mathbb Z^d$ of $X = \mathbb T^d$.  

Table~\ref{tableSpaces} summarizes the properties of $\rkha$ and other function spaces on $ X$ employed in this work. In what follows, we shall let $ \rkha_\text{sa}$ denote the set of self-adjoint elements of $ \rkha$, i.e., the elements $ f \in \rkha$ satisfying $ f^* = f $. Since the $^*$ operation of $\rkha $ corresponds to complex conjugation of functions, it follows that $\rkha_\text{sa}$ contains the real-valued functions in $\rkha$. Note that if $ f = \sum_{j\in\mathbb Z^d} \tilde f_j \psi_j $ is an element of $\rkha_\text{sa}$, then its expansion coefficients in the $\psi_j $ basis satisfy $ \tilde f^*_j = \tilde f_{-j}$.    

Next, we state a product formula for the orthonormal basis functions $ \psi_{j}$, which follows directly from their definition, viz.
\begin{align}
\label{eqProd}
        \psi_j \psi_l  &= c_{jl} \psi_{j+l} \quad \mbox{with} \nonumber\\
        c_{jl} &= \exp\left(-\tau \frac{\lvert j \rvert_p + \lvert l \rvert_p - \lvert j +l \rvert_p}{2}\right).
\end{align}
In the above, we interpret the coefficients $c_{jl}$ as ``structure constants'' of the RKHA $\rkha$. Figure~\ref{figStructConst} displays representative matrices formed by the $c_{jl} $ in dimension $ d = 1 $ and 2 for $ p = 1/4 $ and $ \tau = 1/ 4$. 

\begin{figure}
    \centering
    \includegraphics[width=.9\linewidth]{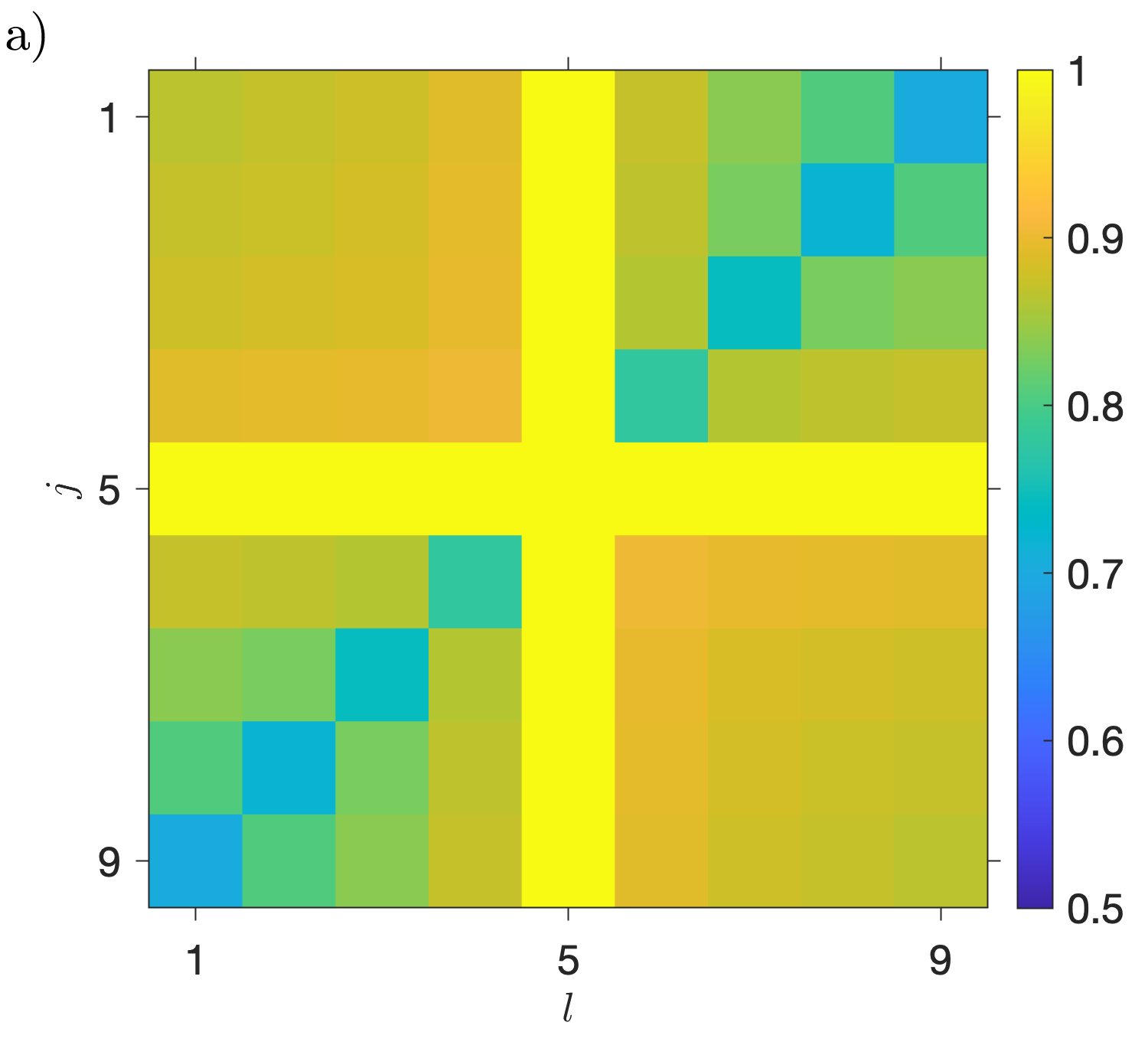}
    \includegraphics[width=.9\linewidth]{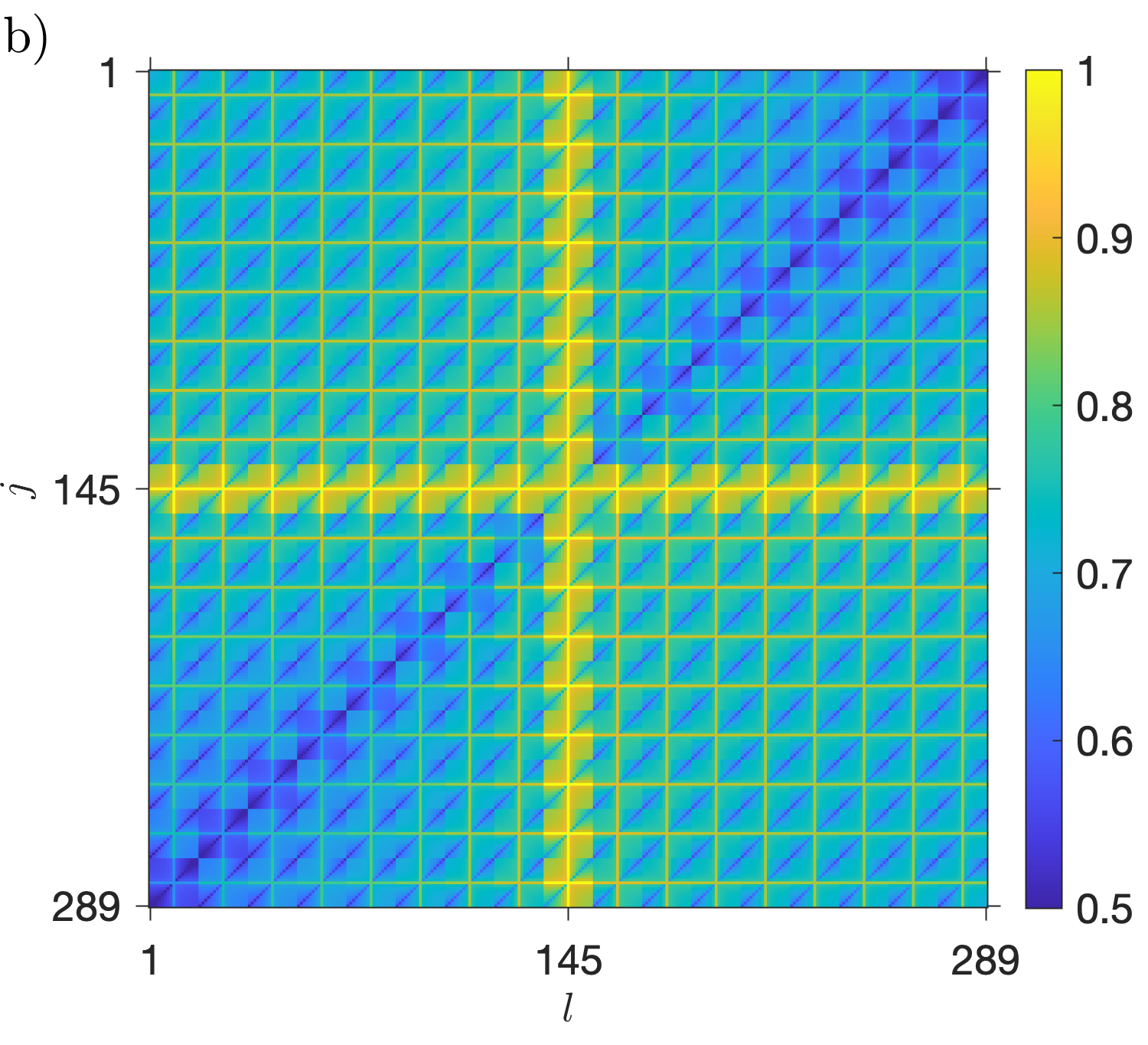}
    \caption{\label{figStructConst}Structure constant matrices $c_{jl}$ for reproducing kernel Hilbert algebras on (a) the circle with $d=1$ and (b) the 2- torus with $d=2$. In both cases, we use the parameter values $ p = 1/4 $ and $ \tau = 1/ 4$ as given in~\eqref{eqProd}. In (a), we consider indices in the range $-2^{n-1} \leq j,l \leq 2^{n-1}$ with $ n = 3$. In (b), the multi-indices $ j = (j_1,j_2)$ and $ l = (l_1,l_2)$ satisfy $ -2^{n/2 -1} \leq j_i,l_i \leq 2^{n/2 -1}$ with $n = 8 $. In both (a) and (b), we map $ j $ and $l $ into standard matrix indices  $1, 2, \ldots, ( 2^{n/d} + 1 )^d$ (which results in $(2^4+1)^2=289$ for (b)) by lexicographical ordering. The matrix in (b) is thus equal to the Kronecker product of the matrix in (a) with itself.}
\end{figure}

In the special case $ d = 1$, we will let $\rkha^{(1)} $ be the RKHA on the circle $S^1 \equiv \mathbb T^1$ constructed as above. We denote the reproducing kernel of $\rkha^{(1)} $ by $\krn^{(1)}$, and let $ \psi_j^{(1)}$, $ j \in \mathbb Z$, be the corresponding orthonormal basis functions with $ \psi_j^{(1)}(\theta) = e^{-\lvert j \rvert^p \tau / 2 } \varphi_j(\theta)$. It then follows that $\rkha$ admits the tensor product factorization 
\begin{equation}
    \label{eqRKHAProd}
    \rkha = \bigotimes_{i=1}^d \rkha^{(1)},
\end{equation}
and the reproducing kernel and orthonormal basis functions of $\rkha$ similarly factorize as 
\begin{align*}
    \krn(x,x') &= \prod_{i=1}^d\krn^{(1)}( \theta^i, \theta^{i\prime} ), \\
    \psi_j(x)&=\prod_{i=1}^d \psi^{(1)}_{j_i}(\theta^i),
\end{align*}
where $ j = (j_1, \ldots, j_d) $, and $\theta^i$, $\theta^{i\prime}$ are canonical angle coordinates of the points $x = (\theta^1, \ldots, \theta^d)$, $x'= (\theta^{1\prime}, \ldots, \theta^{d\prime})$, respectively (see also~\eqref{FourierF}). 
%-----------------------------------------------------------
\begin{table}
    \centering
    \begin{tabularx}{\linewidth}{@{\extracolsep{\fill}}lccccc}
        \hline\hline
        & $L^2(\mu)$ & $L^\infty(\mu)$ & $C(X)$ & $C^\infty(X)$ & $\rkha$\\
        \hline
        Completeness & \checkmark & \checkmark & \checkmark & \crossmark & \checkmark \\
        Hilbert space structure & \checkmark & \crossmark & \crossmark & \crossmark & \checkmark \\ 
        Pointwise evaluation & \crossmark & \crossmark & \checkmark & \checkmark & \checkmark \\
        $^*$-algebra structure & \crossmark & \checkmark & \checkmark & \checkmark & \checkmark \\
        $C^\infty$ regularity & \crossmark & \crossmark & \crossmark & \checkmark & \checkmark \\
        \hline\hline
    \end{tabularx}
    \caption{Properties of representative spaces of classical observables on the compact abelian group $X = \mathbb T^d$. The space $\rkha$ is an RKHA, which, in addition to being an RKHS it has Banach $^*$-algebra structure.}
    \label{tableSpaces}
\end{table}
%-----------------------------------------------------------

\subsection{\label{secRKHAEvolution}Evolution of RKHA observables}

From an operator-theoretic perspective, simulating the dynamical evolution of a continuous classical observable $f \in C(X)$  can be understood as approximating the Koopman operator $U^t $ on $C(X)$; for, if $U^t$ were known one could use it to compute $U^tf(x) = f(\Phi^t(x))$ for every observable $ f \in C(X)$, time $ t \in \mathbb R$, and initial condition $x \in X $ (cf.\ Sec.~\ref{secOverview}). Yet, despite its theoretical appeal, consistently approximating the Koopman operator on $C(X)$ is challenging in practice, as this space lacks the Hilbert space structure underpinning commonly employed operations used in numerical techniques, such as orthogonal projections (see Table~\ref{tableSpaces}). For a measure-preserving, ergodic dynamical system such as the torus rotation in~\eqref{eqTorusRotation}, a natural alternative is to consider the unitary Koopman operator on the $ L^2(\mu) $ Hilbert space associated with the invariant measure $\mu$. While this choice addresses the absence of orthogonal projections on $C(X)$, $L^2(\mu)$ lacks the notion of pointwise evaluation of functions, so one must correspondingly abandon the notion of pointwise forecasting in this space.   

In light of the above considerations, RKHSs emerge as attractive candidates of spaces of classical observables in which to perform simulation, as they allow pointwise evaluation through the reproducing property in~\eqref{eqReproducing} while having a Hilbert space structure. Unfortunately, an obstruction to using RKHSs in dynamical systems forecasting is that a general RKHS $ \mathcal H$ on $X$ need not be preserved under the dynamics, even if the reproducing kernel $k$ is continuous. That is, in general, if $f: X \to \mathbb C$ lies in an RKHS, the composition $ f \circ \Phi^t$ need not lie in the same space, and thus the Koopman operator is not well-defined as an operator mapping the RKHS into itself \cite{DasGiannakis20}.  Intuitively, this is because membership of a function $f$ in an RKHS generally imposes stringent requirements in its regularity, as we discussed for example in Sec.~\ref{secRKHS} with the rapid decay of Fourier coefficients, which need not be preserved by the dynamical flow.  

An exception to this obstruction occurs when the reproducing kernel is translation-invariant, which holds true for the class of kernels introduced in Sec.~\ref{secRKHA} (see~\eqref{eqKTrans}). In fact, it can be shown \cite{Giannakis21b} that  the RKHA $\rkha$ associated with the kernel $\tilde k$ in~\eqref{eqKTorus},is invariant under the Koopman operator $U^t$ for all $ t \in \mathbb R$, and  $U^t : \rkha \to \rkha $ is unitary and strongly continuous. Analogously to the $L^2(\mu)$ case, the evolution group $ \{ U^t : \rkha \to \rkha \}_{t \in \mathbb R} $ is uniquely characterized through its skew-adjoint generator $ V : D(V) \to \rkha$, defined on a dense subspace $D(V) \subset \rkha $, and acting on observables as displayed in~\eqref{eqGenerator}. 

For the torus rotation in~\eqref{eqTorusRotation}, $V$ is diagonalizable in the $ \{ \psi_j \} $ basis of $\rkha$. That is, for $ j = (j_1, \ldots, j_d) \in \mathbb Z^d$, we have
\begin{displaymath}
    V \psi_j = i \omega_j \psi_j,
\end{displaymath}
where $ \omega_j $ is a real eigenfrequency given by
\begin{equation}
    \label{eqOmega}
    \omega_j = j_1 \alpha_1 + \ldots + j_d \alpha_d. 
\end{equation}
Moreover, $ V $ admits a decomposition  into mutually commuting, skew-adjoint generators $ V_1, \ldots, V_d $  satisfying
\begin{equation}
    V_l \psi_j = i j_l \alpha_l \psi_j \quad \mbox{with}\quad l=1, \dots, d. 
    \label{eqVJ}
\end{equation}
In particular, since $ \{ \psi_j \} $ is an orthonormal basis, \eqref{eqVJ} completely characterizes $V_l$, and we have 
\begin{equation}
    \label{eqVDecomp}
    \begin{gathered}
        V = V_1 + \ldots + V_d, \\
        [ V_j, V_l ] = 0, \quad [ V_j, V ] = 0.
    \end{gathered}
\end{equation}
It should be noted that the Koopman generator on $L^2(\mu)$ admits a similar decomposition to~\eqref{eqVDecomp}; see e.g., Ref.~\cite{Giannakis19} for further details. Analogously to the $L^2(\mu)$ case, the Koopman operator on $\rkha$ can be recovered at any $ t \in \mathbb R$ from the generator by exponentiation as given in~\eqref{eqExpKoopman}. 

%========================================================================================
\section{\label{secQuantumMechanical}Embedding into an infinite-dimensional quantum system}

The initial stages of the QECD procedure outlined in Sec.~\ref{secOverview} involve embedding classical states and observables into states and observables of quantum system associated with an infinite-dimensional RKHS $\mathcal H$, arriving at the quantum mechanical level depicted in Fig.~\ref{Scheme}. In this section, we describe the construction of this quantum system and associated embeddings of classical states and observables. First, in Sec.~\ref{secRKHS} we build $\mathcal H$ as a subspace of the RKHA $\rkha$ from Sec.~\ref{secRKHA}. Then, in Secs.~\ref{secQuantumFeature} and~\ref{secQuantumRep} we establish representation maps $ Q : X \to Q(\mathcal H)$ and $ T : \rkha \to B(\mathcal H) $ from classical states and observables into quantum mechanical states and observables, respectively, on $\mathcal H$. Note that the quantum mechanical embedding of states $Q$ passes through an intermediate classical statistical level associated with probability measures on the classical state space (second row in the left-hand column of Fig.~\ref{Scheme}). In Secs.~\ref{secConsistency} and~\ref{secQuantumDynamical}, we establish the classical-quantum consistency and associated dynamical properties of our embeddings.

\subsection{\label{secRKHS}Reproducing kernel Hilbert space}

We choose  $\mathcal H$ as an infinite-dimensional subspace of the RKHA $\rkha$ containing zero-mean functions. For that, we introduce the (infinite) index set
\begin{equation}
    \label{eqJ}
    J = \{ ( j_1, \ldots, j_d ) \in \mathbb Z^d : j_i \neq 0 \},
\end{equation}
and define $\mathcal H$ as the corresponding infinite-dimensional closed subspace
\begin{displaymath}
    \mathcal H = \overline{\spn \{ \psi_j : j \in J \} }.
\end{displaymath}
The space $ \mathcal H $ is then an RKHS with the reproducing kernel
\begin{equation}
    \label{eqK}
    k(x,x') = \sum_{ j \in  J } \psi_j^*(x) \psi_j(x').
\end{equation}
In particular, for every $ f \in \mathcal H$, which is necessarily an element of $\rkha$, the reproducing property in~\eqref{eqReproducing} reads
\begin{displaymath}
    f(x) = \langle k_x, f \rangle_{\mathcal H} = \langle \krn_x, f \rangle_{\rkha},
\end{displaymath}
where $k_x := k(x,\cdot)$ is the section of the kernel $k$ at $x \in X$, and $\langle \cdot, \cdot \rangle_{\mathcal H}$ denotes the inner product of $\mathcal H$. 

By excluding zero indices from the index set $J$,  every element $f$ of $\mathcal H$ has zero mean, $ \int_X f \, d\mu = 0$, as noted above. The reason for adopting this particular definition for $\mathcal H$, instead of, e.g., working with the entire space $\rkha$, is that later on it will facilitate construction of $2^n$-dimensional subspaces $\mathcal H_n \subset \mathcal H$ suitable for quantum computation (see Sec.~\ref{secMatrixMechanical}). In what follows, $ \Pi : \rkha \to \rkha$ will denote the orthogonal projection with $\ran  \Pi = \mathcal H$. Moreover, we set 
\begin{align*}
    \kappa &= k(x,x) = \sum_{j \in J} e^{-\tau \lvert j \rvert_p},\\ 
    \tilde \kappa &= \krn(x,x) = \sum_{j \in \mathbb Z^d} e^{-\tau \lvert j \rvert_p},
\end{align*}
where these definitions are independent of the point $x \in X$ by~\eqref{eqKTrans}. We also note that, by construction, $ \mathcal H$ is a Koopman-invariant subspace of $\rkha$, so we may define unitary Koopman operators $U^t : \mathcal H \to \mathcal H$ by restriction of $U^t : \rkha \to \rkha $ from Sec.~\ref{secRKHAEvolution}.

\subsection{\label{secQuantumFeature}Representation of states with a quantum feature map}

For our purposes, a key property that the RKHS structure of $\mathcal H$ endows is the \emph{feature map}, which is the continuous map $ F : X \to \mathcal H$ mapping classical state $x \in X $ to the RKHS function 
\begin{equation}
    \label{eqFeature}
    F(x) = k_x. 
\end{equation}
It can be shown that for the choice of kernel in~\eqref{eqK}, $F$ is an injective map, and the functions $ \{ F(x) \in \mathcal H : x \in X \} $ are linearly independent. It is then natural to think of the normalized feature vectors 
\begin{equation}
    \xi_x := \frac{k_x}{ \lVert k_x \rVert_{\mathcal H} } = \frac{k_x}{\kappa}
\label{feature_vec}    
\end{equation}
as ``wavefunctions'' corresponding to classical states $x \in X$. 

We can generalize this idea by associating every such wavefunction $\xi_x$ with the pure quantum state $ \rho_x = \langle \xi_x, \cdot \rangle_{\mathcal H} \xi_x$. The mapping $\mathcal F: X \to Q(\mathcal H) $ with  
\begin{equation}
    \label{eqQ}
    \mathcal F(x) = \rho_x
\end{equation}
then describes an embedding of the classical state space $X$ into quantum mechanical states in $ Q(\mathcal H)$, which we refer to as a \emph{quantum feature map}. Note that there is no loss of information in representing $x \in X $ by $ \rho_x \in Q(\mathcal H)$. Moreover, $\mathcal F$ can be understood as a composition $\mathcal F = P \circ \delta $, where $\delta  : X \to \mathcal P(X)$ maps classical state $ x \in X  $ to the Dirac  probability measure $ \delta_x \in \mathcal P(X)$, and $ P  : \mathcal P(X) \to Q(\mathcal H)$ is a map from classical probability measures on $ X $ to quantum states on $\mathcal H$, such that 
\begin{equation}
    \label{eqP}
    P(p) = \int_X \rho_x \, dp(x).
\end{equation}
The map $P$ describes an embedding of the state space $X$ into the space of probability measures $\mathcal P(X)$, i.e., the classical statistical level in the left-hand column of Fig.~\ref{Scheme}. See Ref.~\cite{DasGiannakis20b} for further details on the properties of this map. 

By virtue of it being an RKHS, we can also define classical and quantum feature maps for the RKHA $\rkha$. Specifically, we set $ \tilde F : X \to \rkha$ and $ \tilde{\mathcal F} : X \to Q(\rkha) $, where
\begin{equation}
    \label{eqQRKHA}
    \tilde F(x) = \langle \krn_x, \cdot \rangle_{\rkha}, \quad \tilde{\mathcal F}(x) = \langle \tilde \xi_x, \cdot \rangle_{\rkha} \tilde \xi_x,
\end{equation}
and $ \tilde \xi_x  = \krn_x / \lVert \krn_x \rVert_{\rkha}$. The feature maps $\tilde F$ and $\tilde{\mathcal F}$ have analogous properties to $F$ and $\mathcal F$, respectively, which we do not discuss here in the interest of brevity.

\subsection{\label{secQuantumRep}Representation of observables}

The quantum mechanical representation of classical observables in $\rkha$  is considerably facilitated by the Banach algebra structure of that space. In Sec.~\ref{secRKHARep}, we leverage that structure to build representation maps from functions in $\rkha$ to bounded linear operators in $B(\rkha)$. Then, in Sec.~\ref{secRKHSRep}, we consider associated representations mapping into bounded linear operators on the RKHS $\mathcal H$ (which is a strict subspace of $\rkha$), arriving at the map $ T : \rkha \to B(\mathcal H)$ depicted in the right-hand column of Fig.~\ref{Scheme}. Additional details on the construction are provided in Appendix~\ref{appQuantumRep}. 

\subsubsection{\label{secRKHARep}Representation on the RKHA $\rkha$}
We begin by noting that the joint continuity of the multiplication operation of Banach algebras (see~\eqref{eqBanachAlg}) implies that for every $ f \in \rkha$ the multiplication operator $ A_f : g \mapsto f g$ is well-defined as a bounded operator in $B(\rkha)$. This leads to the regular representation $ \pi : \rkha \to B(\rkha)$, which is the algebra homomorphism of $\rkha$ into $B(\rkha)$, mapping classical observables in $\rkha$ to their corresponding multiplication operator,
\begin{equation}
    \label{eqPi}
    \pi f := A_f.
\end{equation}
This mapping is a homomorphism since
\begin{displaymath}
    \pi(fg) = A_{fg} = A_f A_g, \quad \forall f, g \in \rkha,
\end{displaymath}
and it is injective (i.e., faithful as a representation) since $(\pi(f-f')) 1_X = f - f' \neq 0 $ whenever $ f \neq f' $. However, $\pi$ is not a $^*$-representation; i.e., $\pi(f^*)$ is not necessarily equal to $A_f^*$. In particular, $A_f$ need not be a self-adjoint operator in $B(\rkha)$ if $ f $ is a self-adjoint element in $\rkha_\text{sa}$.   
To construct a map from $\rkha$ into the self-adjoint operators in $ B(\rkha)$, we define $\tilde T: \rkha \to B(\rkha) $ with 
\begin{equation}
    \label{eqT}
    \tilde T f = \frac{\pi f  + (\pi f)^* }{2}.
\end{equation}
By construction, $ \tilde T f $ is self-adjoint for all $ f \in \rkha $, and it can also be shown (see Appendix~\ref{appInjectivity}) that $\tilde T$ is injective on $ \rkha_\text{sa}$. That is, $ \tilde T$ provides a one-to-one mapping between real-valued functions in $\rkha$ and self-adjoint operators in $ B(\rkha)$. 
%-----------------------------------------------------
\begin{figure}
    \centering
    \includegraphics[width=.93\linewidth]{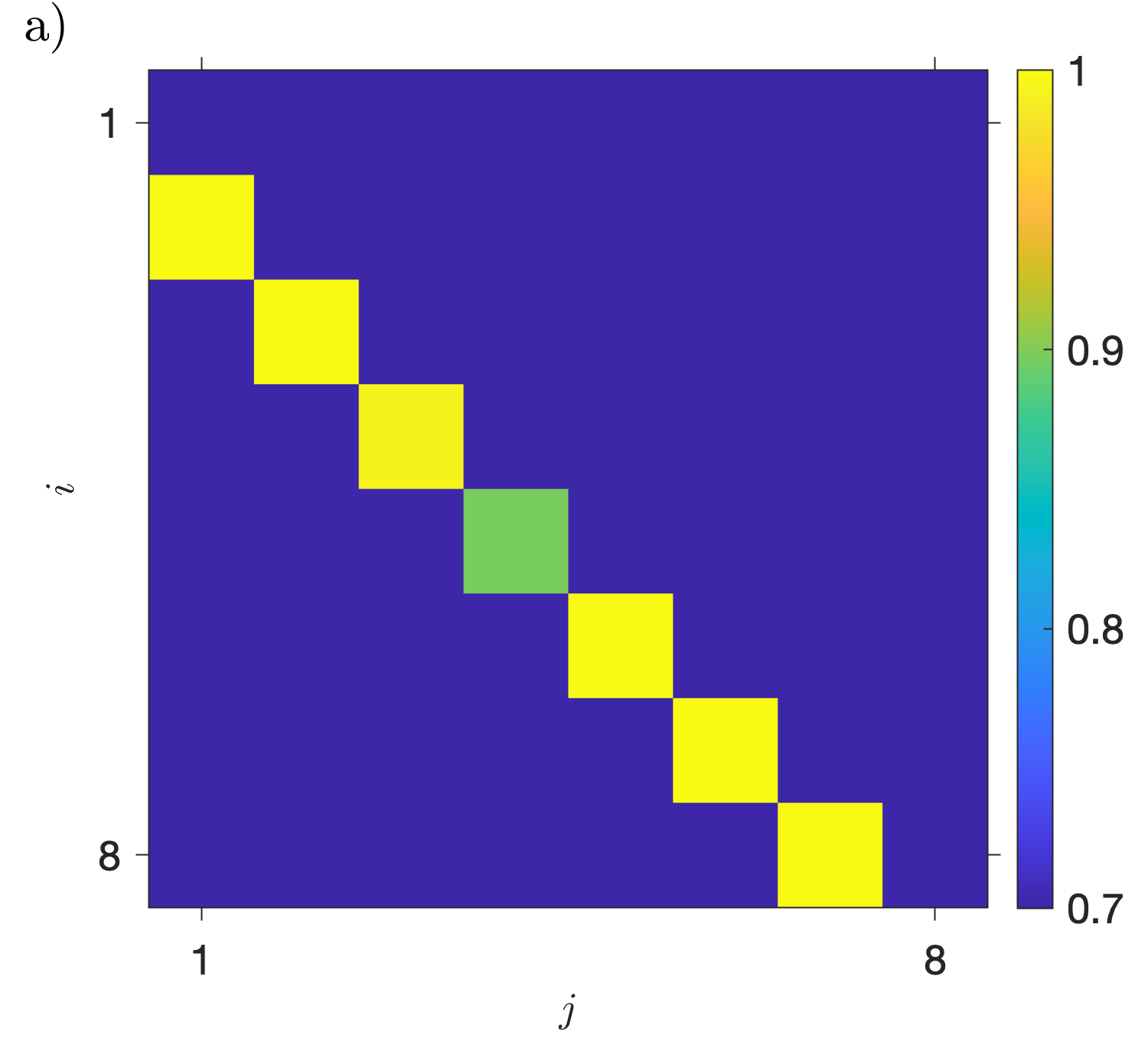}
    \includegraphics[width=.93\linewidth]{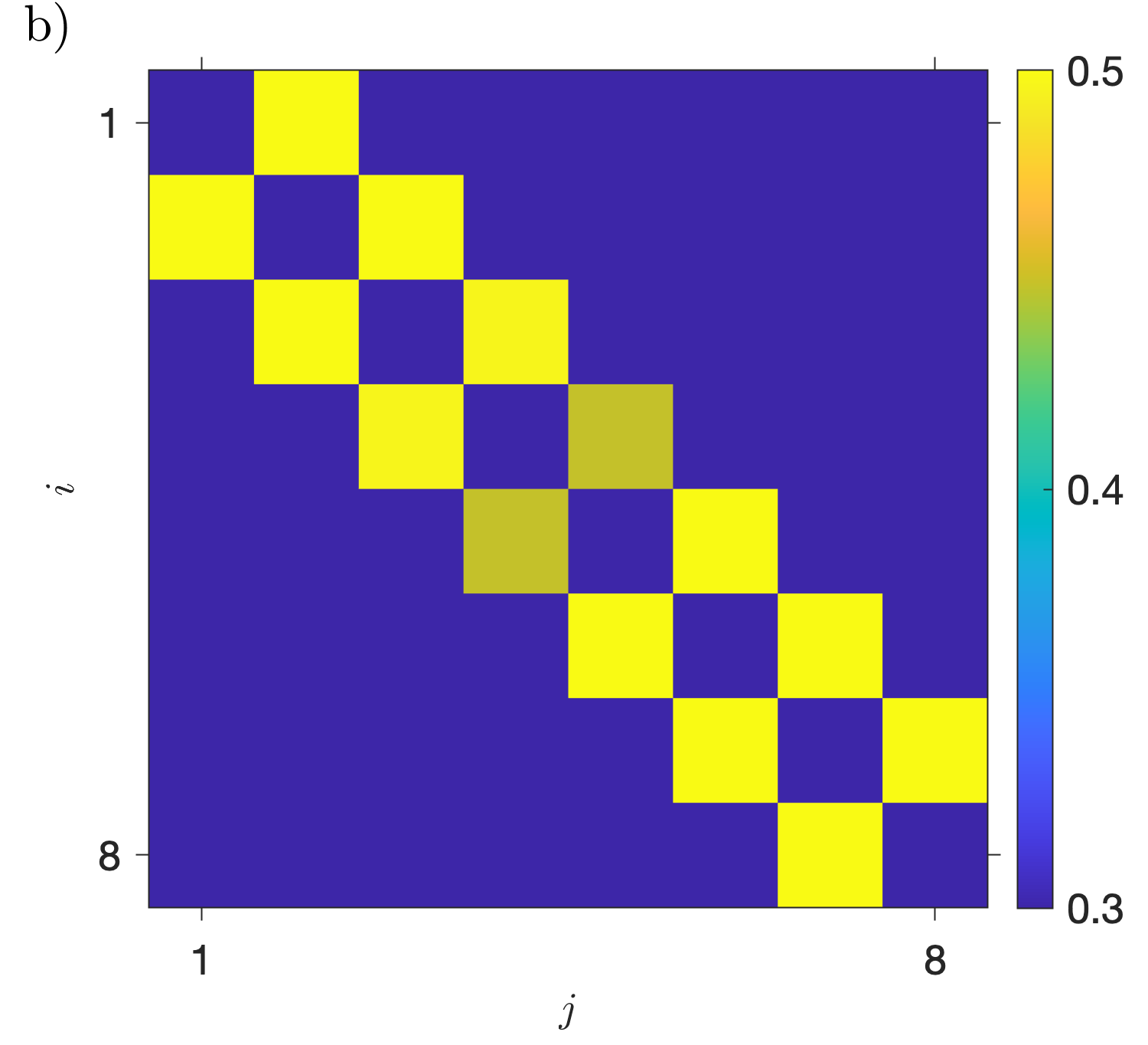}
    \caption{\label{figMultOp} Matrix elements  $(A_{\psi_1})_{ij}$ (a) and $(S_{\psi_1})_{ij}$ (b) of the multiplication operator $A_{\psi_1}$ (Eq.~\eqref{eqMIJ}) and the self-adjoint operator $S_{\psi_1}$ (Eq.~\eqref{eqSIJ}) representing the basis function $ \psi_1$ for dimension $d=1$. As in Fig.~\ref{figStructConst}(a), we consider the reproducing kernel Hilbert algebra $\rkha$ on the circle from  with $p=1/4$ and $\tau=1/4$, and map the indices $i$ and $j$ into standard matrix indices $1,2,\ldots, 2^n+1$ with $n=3$. The matrix in (a) has nonzero elements only in the first lower diagonal, $ i - j = 1$. The matrix in (b) is a symmetric bidiagonal matrix with elements in the first upper and lower diagonals, $ i - j = \pm 1$.}
\end{figure}
%-----------------------------------------------------

It follows from the  product formula in~\eqref{eqProd} that if $f \in \rkha$ has the expansion $ f = \sum_{j \in \mathbb Z^d} \tilde f_j \psi_j$, where $ \tilde f_j = \langle \psi_j, f \rangle_{\rkha}$, then the corresponding multiplication operator $A_f = \pi f$ has the matrix elements
\begin{equation*}
    (A_{f})_{ij} := \langle \psi_i, A_f \psi_j \rangle_{\rkha} = \langle \psi_i, f \psi_j \rangle_{\rkha},
\end{equation*}
and thus
\begin{equation}
    \label{eqMIJ}
    (A_{f})_{ij} = c_{j,i-j} \tilde f_{i-j}. 
\end{equation}
Correspondingly, the matrix elements of the self-adjoint operator $ S_f := \tilde T f$ are given by 
\begin{equation*}
    (S_{f})_{ij} := \langle \psi_i, S_f \psi_j \rangle_{\rkha} = \frac{(A_{f})_{ij} + (A_{f})^*_{ji}}{2}.
\end{equation*}
If, in addition, $ f $ lies in $\rkha_\text{sa}$, then we have $ \tilde f_{j-i}^* = \tilde f_{i-j} $ and the formula above reduces to
\begin{equation}
    \label{eqSIJ}
    (S_{f})_{ij} = \frac{c_{j,i-j} + c_{i,j-i}}{2} \tilde f_{i-j}. 
\end{equation}

Here, of particular interest are the multiplication and self-adjoint operators representing the basis elements of $\rkha$, i.e., $A_{\psi_l}$ and $S_{\psi_l}$, respectively, for $ l \in \mathbb Z^d$. Since $ \tilde f_{i-j} = \delta_{l,i-j} $ for $ f = \psi_l $, it follows from~\eqref{eqMIJ} that after a suitable lexicographical ordering of multi-indices (as in Fig.~\ref{figStructConst}), $(A_f)_{ij}$ forms a banded matrix with nonzero elements only in the diagonal corresponding to multi-index $k$. Figure~\ref{figMultOp}(a) illustrates the nonzero matrix elements of $A_{\psi_1} $ in the one-dimensional case, $d = 1$. Similarly, the self-adjoint operator $S_{\psi_1}$ is a bi-diagonal matrix with nonzero entries in the diagonals corresponding to $ \pm 1 $, as shown in Fig.~\ref{figMultOp}(b). 

We deduce from these observations that if $ f $ is a bandlimited observable (i.e., expressible as a finite linear combination of Fourier functions $ \phi_j $), $A_f$ is represented by a banded matrix, whose $l$-th diagonal comprises of the  structure constants $ c_{lj} $ multiplied by $\tilde f_l $. The matrix representing $S_f $ is also banded whenever $f$ is bandlimited. If, in addition, $f$ is real, the $l$-th diagonal of $ (S_f)_{ij} $ is given by the multiple of $ ( c_{lj} + c_{li} ) / 2  $ with $ \tilde f_{l}$.

\subsubsection{\label{secRKHSRep}Representation on the RKHS $\mathcal H$}

We now take up the task of defining analogs of the maps $ \pi : \rkha \to B(\rkha)$ and $ T : \rkha \to B(\rkha)$ from Sec.~\ref{secRKHSRep}, mapping elements of $\rkha$ to bounded operators on the RKHS $\mathcal H \subset \rkha$ (i.e., the Hilbert space underlying the infinite-dimensional system at the quantum mechanical level). To that end, let $\bm \Pi$ be the projection map from $B(\rkha)$ to $B(\mathcal H)$, defined as  
\begin{equation}
    \label{eqPiProj}
    \bm \Pi A := \Pi A \Pi, 
\end{equation}
where $\Pi$ is the orthogonal projection from $\rkha $ to $\mathcal H$ introduced in Sec.~\ref{secRKHS}. One can explicitly verify that the map $ \bm \Pi \circ \pi : \rkha \to B(\mathcal H)$ is injective, so there is no loss of information in representing $f \in \rkha$ by $ \bm \Pi(\pi f) \in B(\mathcal H)$ as opposed to $ \pi f \in B(\rkha)$. For our purposes, however, in addition to injectivity we require that our representation maps provide value-level consistency between classical and quantum measurements (in a sense made precise in Sec.~\ref{secConsistency} below). For that, it becomes necessary to modify the map $ \bm \Pi \circ \pi $, as well as its self-adjoint counterpart $\bm \Pi \circ \tilde T $, to take into account the contractive effect of the projection $ \bm \Pi $. 

In Appendix~\ref{appConsistency}, we construct a self-adjoint, invertible operator $L : \rkha \to \rkha$, which is diagonal in the $ \{ \psi_j \} $ basis, and whose role is to counter-balance that contraction. Specifically, we define $ \varpi : \rkha \to B(\mathcal H)$ and $T : \rkha \to B(\mathcal H) $ with   
\begin{equation}
    \label{eqVarpi}
    \varpi = \bm \Pi \circ \pi \circ L^{-1}, \quad T = \bm \Pi \circ \tilde T \circ L^{-1}.
\end{equation}
Here, $L^{-1}$ inflates the expansion coefficients of functions in the $ \{ \psi_j \} $ basis of $\rkha$, absorbing the contractive action of $\bm \Pi$. Analogously to $\pi$ and $\tilde T$, respectively, $\varpi $ is one-to-one, and $T$ is one-to-one on the real functions in $\rkha_\text{sa}$. Moreover, every operator in the range of $T$ is self-adjoint. The map $T$ provides the representation of classical observables in $\rkha$ by self-adjoint operators in $B(\mathcal H)$ at the quantum mechanical level, depicted by vertical arrows in the right-hand column of Fig.~\ref{Scheme}.

\subsection{\label{secConsistency}Classical--quantum consistency}

We now come to a key property of the regular representation $\pi$ and the associated map $\tilde T$, which is a consequence of the reproducing property and Banach algebra structure of $\rkha$. Namely, $\pi$ and $\tilde T $ provide a consistent correspondence between evaluation of classical observables and quantum mechanical expectation values. To see this, for any quantum state $ \varrho \in Q( \rkha) $ and quantum mechanical observable $ A \in B(\rkha)$, let 
\begin{equation}
\langle A\rangle_{\varrho} := \tr( \varrho A )
\end{equation}
be the standard quantum mechanical expectation functional. Then, it follows from the reproducing property in~\eqref{eqReproducing}, the definition of the quantum feature map $\tilde{\mathcal F}$ in~\eqref{eqQRKHA}, and the definition of the regular representation in~\eqref{eqPi} that for any observable $ f \in \rkha $ and classical state $ x \in X$,
\begin{equation}
    f(x) = \langle \pi f\rangle_{\varrho_x} = \langle \tilde Tf\rangle_{\varrho_x}, 
\label{expec1}    
\end{equation}
where $\varrho_x = \tilde{\mathcal F}(x)$. The last equality in~\eqref{expec1} requires that $f$ is a self-adjoint element in $ \rkha_\text{sa}$; see Ref.~\cite{DasGiannakis20b} for further details. Equation~\eqref{expec1} shows, in particular, that by passing to the quantum mechanical representation we maintain pointwise consistency with the classical measurement processes for special sets of quantum mechanical observables and states. These are the self-adjoint operators $S_f$ and the pure states $\varrho_x$.

To express these relationships in terms of matrix elements, note first that the quantum state $\varrho_x$ satisfies 
\begin{align}
    \nonumber
    (\varrho_{x})_{ij}  :=& \langle \psi_i, \varrho_x \psi_j \rangle_{\rkha} \nonumber\\
                     =& \frac{\langle \psi_i, \krn_x \rangle_{\rkha} \langle \krn_x, \psi_j \rangle_{\rkha} }{\tilde\kappa} \nonumber\\
                     =& \frac{\psi^*_i(x) \psi_j(x)}{\tilde\kappa}.
\label{eqRhoIJ}    
\end{align}
Combining this result with~\eqref{eqMIJ}, we obtain
\begin{align*}
    f(x) =& \tr( \varrho_x (\pi f) ) \nonumber \\
         =& \sum_{i,j \in \mathbb Z^d} (\varrho_{x})_{ij} (A_f)_{ji} \nonumber \\
    =& \sum_{i,j \in \mathbb Z^d} \frac{\psi^*_i(x) \psi_j(x) c_{i,j-i} \tilde f_{j-i}}{\tilde\kappa},
\end{align*}
and this relationship holds irrespective of whether $f$ is self-adjoint or not. If $f $ is a self-adjoint element in $\rkha_{\text{sa}}$, then we can use the matrix elements of the self-adjoint operator $S_f$ from~\eqref{eqSIJ}, in conjunction with the fact that $ \varrho_x$ is also self-adjoint, to arrive at the expression   
\begin{align*}
    f(x) &= \tr( \varrho_x (\tilde Tf) ) \nonumber\\
         &= \sum_{i,j \in \mathbb Z^d} (\varrho_{x})_{ij} (S_{f})_{ji} \nonumber \\
         &= \sum_{i,j \in \mathbb Z^d} \frac{\psi^*_i(x) \psi_j(x) ( c_{i,j-i} + c_{j,i-j} ) \tilde f_{j-i}}{2 \tilde\kappa}.
\end{align*}

Even though  $\mathcal H$ is a strict subspace of the RKHA $\rkha$, it is still possible to consistently recover all predictions made for classical observables, as we describe in Appendix~\ref{appConsistency}. There, we show that the modified versions $\varpi : \rkha \to B(\mathcal H)$ and $ T : \rkha \to B(\mathcal H)$ of $\pi : \rkha \to B(\rkha)$ and $\tilde T : \rkha \to B(\rkha)$, respectively (defined in~\eqref{eqVarpi}), satisfy the analogous consistency relation to~\eqref{expec1}, i.e., 
\begin{equation}
    f(x) = \langle \varpi f\rangle_{\rho_x} = \langle Tf\rangle_{\rho_x}, 
    \label{expecRKHS}
\end{equation}
where $\rho_x = \mathcal F(x)$ is the quantum state on $\mathcal H$ obtained from the feature map in~\eqref{eqQ}. As with~\eqref{expec1}, the first equality in~\eqref{expecRKHS} holds for any $f \in \rkha$ and the second holds for real-valued elements $ f \in \rkha_{\text{sa}}$.

\subsection{\label{secQuantumDynamical}Dynamical evolution}

In this section, we describe the dynamics of quantum states and observables associated with the RKHA $\rkha$ and RKHS $\mathcal H \subset \rkha $, and establish consistency relations between the classical and quantum evolution. 

First, recall that the Koopman operators $U^t$ act on $\rkha$ as a unitary evolution group. As a result, there is an induced action $ \mathcal U^t : B( \rkha ) \to B(\rkha)$ on quantum mechanical observables in $B(\rkha)$, given by 
\begin{equation}
    \label{eqUOp}
    \mathcal U^t A = U^t A U^{t*}.
\end{equation}
This action has the important property of being compatible with the action of the Koopman operator on functions in $ \rkha$ under the regular representation. Specifically, for every $ f \in \rkha$ and $ t \in \mathbb R$, we have
\begin{equation}
    \label{eqUPi}
    \mathcal U^t ( \pi f ) = \pi( U^t f ).  
\end{equation}

The unitary evolution in~\eqref{eqUOp} has a corresponding dual action $ \Psi^t : Q( \rkha) \to Q(\rkha) $ on quantum states, given by
\begin{equation}
    \label{eqPsiOp}
    \Psi^t ( \varrho ) = U^{t*} \varrho U^t \equiv \mathcal U^{-t} \varrho.
\end{equation}
One can verify that this action is compatible with the classical dynamical flow under the feature map $\tilde{\mathcal F} : X \to Q(\rkha) $, viz.
\begin{equation}
    \label{eqPsiF}
    \Psi^t( \tilde{\mathcal F}(x) ) = \tilde{\mathcal F}(\Phi^t(x)).
\end{equation}
Using~\eqref{expec1}, \eqref{eqUOp}, and~\eqref{eqPsiOp}, we arrive at the consistency relationships
\begin{equation}
    U^t f(x) = \langle \mathcal U^t (\pi f) \rangle_{\varrho_x}= \langle \pi f \rangle_{\Psi^t(\varrho_x)},
    \label{eqUtMF}    
\end{equation}
with $\varrho_x = \tilde{\mathcal F}(x)$. This holds for every classical observable $ f \in \rkha$, initial condition $ x \in X$, and evolution time $ t \in \mathbb R$. If, in addition, $ f $ is a self-adjoint element in $ \rkha_{\text{sa}}$, we may compute the evolution $U^t f $ using the self-adjoint operator $ \tilde Tf$, which is accessible via physical measurements. That is, for $ f \in \rkha_\text{sa}$ we have
\begin{equation}
    \label{eqUtSF}
    U^t f(x) = \langle \mathcal U^t (\tilde T f) \rangle_{\varrho_x}= \langle \tilde T_f\rangle_{\Psi^t(\varrho_x)}. 
\end{equation}

In summary, we have constructed a dynamically consistent embedding of the torus rotation from~\eqref{eqTorusRotation} into a quantum mechanical system on the RKHA $ \rkha$. For completeness, we note that the matrix elements of the state $\Psi^t(\rho_x)$ are given by
\begin{align*}
    \langle \psi_i, \Psi^t(\varrho_x) \psi_j \rangle_{\rkha} &= \langle U^t \psi_i, U^t \varrho_x \psi_j \rangle_{\rkha} \\ 
    &= e^{i (\omega_j - \omega_i) t} (\varrho_{x})_{ij}. 
\end{align*}
Using this formula together with the expressions for the matrix elements of  $\tilde T f $ in~\eqref{eqSIJ}, respectively, we arrive at the expression
\begin{displaymath}
    U^tf = \sum_{i,j \in \mathbb Z^d} e^{i(\omega_j-\omega_i)t} \frac{ \psi^*_i(x) \psi_j(x) ( c_{i,j-i} + c_{j,i-j} ) \tilde f_{j-i}}{2\tilde\kappa},
\end{displaymath}
which holds for all self-adjoint elements $ f = \sum_{j \in \mathbb Z^d} \tilde f_j \psi_j \in \rkha_{\text{sa}}$. 

Our discussion was thus far based on the RKHA $\rkha$, as opposed to the RKHS $\mathcal H$. In Appendix~\ref{appRKHSDyn}, we establish that the dynamics of classical states and observables can be represented consistently through their representatives on $\mathcal H$ using the maps $\varpi$ and $T$ in~\eqref{eqVarpi}. Specifically, we show that for any $ f \in \rkha$, 
\begin{displaymath}
    U^t f(x) = \langle \mathcal U^t (\varpi f) \rangle_{\rho_x}= \langle \varpi f \rangle_{\Psi^t(\rho_x)},
\end{displaymath}
while for any real-valued $ f \in \rkha_\text{sa}$,
\begin{displaymath}
    U^t f(x) = \langle \mathcal U^t (T f) \rangle_{\rho_x}= \langle T f \rangle_{\Psi^t(\rho_x)},
\end{displaymath}
where $\rho_x = \mathcal F(x)$. In the above, $\mathcal U^t : B(\mathcal H) \to B(\mathcal H)$ and $\Psi^t : Q(\mathcal H) \to Q(\mathcal H)$ are evolution operators on quantum observables and states on $\mathcal H$, respectively, defined analogously to their counterparts on $\rkha$ using the Koopman operator $U^t : \mathcal H \to \mathcal H$ (see Sec.~\ref{secRKHS}).

%========================================================================================
\section{\label{secMatrixMechanical}Projection to finite dimensions}

While being dynamically consistent with the underlying classical evolution, the quantum system constructed in Sec.~\ref{secQuantumMechanical} is infinite-dimensional, and thus not directly accessible to simulation by a quantum computer. We now describe an approach for projecting the infinite-dimensional quantum system to a finite-dimensional system. In Fig.~\ref{Scheme} we refer to this level of representation as {\em matrix mechanical}, since all linear operators involved have finite rank and are representable by matrices. Our objectives are to construct this projection such that (a) it is refinable, i.e., the original quantum system is recovered in a limit of infinite dimension (number of qubits); and (b) it facilitates the eventual passage to the quantum computational level (to be described in Sec.~\ref{secQuantumComputational}). 

We begin by fixing a positive integer parameter $ n $ (the number of qubits), chosen such that it is a multiple of the dimension $d$ of the classical state space $X$, and defining the index sets
\begin{equation}
    \label{eqIdxJ}
    \begin{gathered}
        J_{n,d} = \{ - 2^{n/d - 1}, \ldots, -1, 1, \ldots, 2^{n/d-1} \}, \\ 
        J_n = \{ (j_1, \ldots, j_d) \in \mathbb Z^d :  j_i \in J_{n,d} \}.
    \end{gathered}
\end{equation}
Note that $J_n$ is a subset of $J$ from~\eqref{eqJ} with $ N \equiv 2^n $ elements. Next, consider the $N$-dimensional subspace of $\mathcal H$ given by
\begin{displaymath}
    \mathcal H_n = \spn \{ \psi_j : j \in J_n \},
\end{displaymath}
and let $ \Pi_n : \mathcal H \to \mathcal H$ be the orthogonal projection mapping into $ \mathcal H_n$. When appropriate, we will interpret $\Pi_n$ as a map into its range, i.e., $ \Pi_n : \mathcal H \to \mathcal H_n$, without change of notation. The subspace $ \mathcal H_n $ has the structure of an RKHS of dimension $ 2^n$, associated with the spectrally truncated reproducing kernel
\begin{displaymath}
    k_n(x,x') = \sum_{ j \in J_n } \psi_j^*(x) \psi_j(x').
\end{displaymath}
Moreover, $ \mathcal H_d, \mathcal H_{2d}, \mathcal H_{3d}, \ldots $ is a nested family of subspaces, increasing towards  $\mathcal H$. 

By virtue of being spanned by eigenfunctions of the generator $V$, $\mathcal H_n$ is invariant under the Koopman operator, i.e., $ U^t \mathcal H_n = \mathcal H_n$ for all $ t \in \mathbb R$. Moreover, the projection $\Pi_n $ commutes with both $V $ and $U^t$, 
\begin{displaymath}
    [ V, \Pi_n ] = 0, \quad [ U^t, \Pi_n ] = 0.
\end{displaymath}
These invariance properties allow us to define a projected generator
\begin{equation}
    \label{eqVN}
    V_n \equiv  \Pi_n V \Pi_n, 
\end{equation}
and an associated Koopman operator 
\begin{displaymath}
    U^t_n := e^{t V_n} \equiv  \Pi_n U^t \Pi_n,
\end{displaymath}
such that the following diagram commutes for all $ t \in \mathbb R$:
\begin{equation}
    \label{eqComU}
    \begin{tikzcd}
        \mathcal H \arrow[r, "U^t"] \arrow[d, "\Pi_n", swap] & \mathcal H \arrow[d, "\Pi_n"] \\
        \mathcal H_n \arrow[r, "U^t_n"] & \mathcal H_n
    \end{tikzcd}.
\end{equation}
Similarly, we define a finite-rank Heisenberg operator 
\begin{displaymath}
    \mathcal U^t_n := \mathcal U^t \bm \Pi_n,
\end{displaymath}
where $\bm \Pi_n : B(\mathcal H) \to \mathcal B(\mathcal H)$ is the projection on $B(\mathcal H)$ defined as $ \bm \Pi_n A = \Pi_n A \Pi_n$. This leads to an analogous commutative diagram to that in~\eqref{eqComU} viz., 
\begin{displaymath}
    \begin{tikzcd}
        B( \mathcal H ) \arrow[r, "\mathcal U^t"] \arrow[d, "\bm \Pi_n", swap] & B( \mathcal H ) \arrow[d, "\bm \Pi_n"] \\
        B( \mathcal H_n ) \arrow[r, "\mathcal U^t_n"] & B( \mathcal H_n )
    \end{tikzcd}.
\end{displaymath}

Next, we introduce a spectrally truncated feature map $ F_n : X \to  \mathcal H_n$, defined analogously to~\eqref{eqFeature} as 
\begin{displaymath}
    F_n(x) = k_{x,n} := k_n(x,\cdot ),
\end{displaymath}
as well as a corresponding quantum feature map $ \mathcal F_n : X \to Q(\mathcal H_n) $, such that $ \mathcal F_n(x) = \rho_{x,n} $ is given by
\begin{equation}
    \label{eqRhoN}
    \begin{gathered}
        \rho_{x,n} = \langle \xi_{x,n}, \cdot \rangle_{\mathcal H_n} \xi_{x,n} \quad\mbox{with} \\
        \xi_{x,n} = \frac{k_{x,n}}{\sqrt{\kappa_n}}, \quad \kappa_n = k_n(x,x) = \sum_{j\in J_n} e^{-\tau \lvert j \rvert_p }.
    \end{gathered}
\end{equation}
In the sequel, we will use the states $ \rho_{x,n}$ as approximations of the states $\rho_x = \mathcal F(x)$. These approximations have the following properties.
\begin{enumerate}
    \item The dynamical evolution of $\rho_{x,n}$ is governed by a \emph{finite-rank} operator $ \Psi^t_n : Q(\mathcal H_n) \to Q(\mathcal H_n)$, where
        \begin{displaymath}
            \Psi^t_n(\rho_{x,n}) =  U^{t*}_n \rho_{x,n} U^t_n.
        \end{displaymath}
    \item As $n \to \infty $ (i.e., in the infinite qubit limit), $\rho_{x,n}$ converges to $ \rho_x$, in the sense that for any quantum mechanical observable $A \in B(\mathcal H)$,
        \begin{equation}
            \label{eqERhoXN}
            \langle A_n\rangle_{\rho_{x,n}} \xrightarrow{n\to\infty} \langle A \rangle_{\rho_x},
        \end{equation}
        where $ A_n = \bm \Pi_n A$, and the convergence is uniform with respect to $ x \in X$. 
\end{enumerate}

In light of the above, we employ the following approximations to the quantum mechanical representation of the evolution of classical observables from Sec.~\ref{secQuantumDynamical} (see also Appendix~\ref{appRKHSDyn}), 
\begin{equation}
    \label{eqUtApprox}
    \begin{aligned}
        \check f_n^{(t)}(x) &:= \langle \bm \Pi_n (\varpi f)\rangle_{\Psi^t(\rho_{x,n})}, \\
        f_n^{(t)}(x) &:= \langle \bm \Pi_n (Tf)\rangle_{\Psi^t(\rho_{x,n})}. 
    \end{aligned}
\end{equation}
By~\eqref{eqERhoXN}, for every function $ f \in \rkha$ and evolution time $ t \in \mathbb R$,  $ \check f_n^{(t)}(x)$ converges as $n\to \infty$ to $U^t f(x)$, uniformly with respect to $x \in X $, whereas $ f_n^{(t)}(x)$ converges to $U^tf(x)$ if $f$ is self-adjoint (real-valued). 

%========================================================================================
\section{\label{secQuantumComputational}Representation on a quantum computer}

We are now ready to perform the final step in the QECD pipeline, namely passage from the matrix mechanical level to the quantum computational level associated with the $n$-qubit Hilbert space  $\mathbb B_n = \mathbb B^{\otimes n}$ (see bottom row in Fig.~\ref{Scheme}). We will do so by applying a unitary map, so that the systems in the matrix mechanical and quantum computational levels are isomorphic as quantum systems. However, the key aspects that the quantum computational system provides are that (a) it can be efficiently implemented as a quantum circuit with a quadratic number of gates in $n$; and (b) information about the evolution of classical observables can be extracted by measurement of the standard projection-valued measure associated with the computational basis. We describe the construction of the unitary map from the matrix mechanical to quantum computational levels and the properties of the resulting quantum system in Secs.~\ref{secTensorProd} and~\ref{secWalsh}, respectively.

\subsection{\label{secTensorProd}Quantum computational system on the tensor product Hilbert space}

Being expressible in terms of finite-rank quantum states, observables, and evolution operators, the approximation framework described in Sec.~\ref{secMatrixMechanical} can be encoded in a quantum computing system operating on a finite-dimensional Hilbert space. In particular, letting $ \mathbb B \simeq \mathbb C^2$ denote the 2-dimensional Hilbert space associated with a single qubit, it follows immediately from the fact that $\mathcal H_n$ is a $2^n$-dimensional Hilbert space that there exists a unitary map $ W_n: \mathcal H_n \to \mathbb B_n$, where  
\begin{equation}
    \mathbb B_n := \mathbb B^{\otimes n} \simeq \underbrace{\mathbb C^2 \otimes \dots \otimes \mathbb C^2}_n
\label{tensor}
\end{equation}
is the tensor product Hilbert space associated with $n$ qubits. Under such a unitary, the projected generator $V_n$ from~\eqref{eqVN} maps to a skew-adjoint operator $ \hat V_n := W_n V_n W^*_n $, inducing a self-adjoint Hamiltonian 
\begin{equation}
    \label{eqH}
    H_n := \frac{1}{i}\hat V_n,
\end{equation}
and a corresponding unitary evolution operator $ \hat U^t_n := e^{i H_n t}$ on $ \mathbb B_n$. This leads to the commutative diagram
%-------------------------------------------------
\begin{displaymath}
    \begin{tikzcd}
        \mathcal H_n \arrow[r, "U^t_n"] \arrow[d, "W_n",swap] & \mathcal H_n \arrow[d, "W_n"] \\
        \mathbb B_n \arrow[r, "\hat U^t_n"] & \mathbb B_n
    \end{tikzcd},
\end{displaymath}
%-------------------------------------------------
expressing the fact that elements of $\mathcal H_n$ and $\mathbb B_n$ evolve consistently under $U^t_n$ and $\hat U^t_n$, respectively. Note that we work here with the self-adjoint Hamiltonian $H_n$ as opposed to the skew-adjoint generator $ \hat V_n $ for consistency with the usual convention in quantum mechanics.  

In addition, $W_n$ induces a unitary $\mathcal W_n : B(\mathcal H_n) \to B(\mathbb B_n)$, with $ \mathcal W_n A = W_n A W_n^* $, mapping quantum mechanical observables on $\mathcal H_n$ to quantum mechanical observables on $ \mathbb B_n$. The restriction of $\mathcal W_n$ on $ Q(\mathcal H_n) \subset B(\mathcal H_n)$ then induces a continuous, invertible map $ \mathcal W_n : Q(\mathcal H_n) \to Q(\mathbb B_n) $ from quantum states on $\mathcal H_n$ to quantum states on $ \mathbb B_n$ (which we continue to denote using the symbol $\mathcal W_n$). Moreover, we have the evolution maps 
%-------------------------------------------------
\begin{equation}
    \label{eqPsiUN}
    \begin{aligned}
        \hat \Psi_n^t : Q(\mathbb B_n) &\to Q(\mathbb B_n): \hat{\rho}_n \mapsto  \hat \Psi^t_n(\hat \rho_n) = \hat U^{t*}_n \hat \rho_n \hat U^t_n ,\\
        \hat{\mathcal U}^t_n : B(\mathbb B_n) &\to B(\mathbb B_n):  \hat{A}_n \mapsto \hat{\mathcal U}^t_n \hat{A} _n= \hat U^t_n \hat{A}_n \hat U^{t*}_n,
    \end{aligned}
\end{equation} 
%-------------------------------------------------
such that the maps for states and observables between and within the matrix mechanical and quantum computational level in Fig. \ref{Scheme} constitute commutative diagrams.    
In particular, following the vertical arrows in the left- and right-hand columns from the classical level to the quantum computational level gives the maps $ \hat{\mathcal F}_n : X \to Q(\mathbb B_n) $ and $\hat T_n: \alg \to B(\mathbb B_n)$, where 
%-------------------------------------------------
\begin{equation}
    \label{eqFTN}
    \begin{aligned}
        \hat{\mathcal F}_n &= \mathcal W_n \circ \bm \Pi_n' \circ P \circ \delta \\
        \hat T_n &= \mathcal W_n \circ \bm \Pi_n \circ T.  
    \end{aligned}
\end{equation}
%-------------------------------------------------
The maps $\hat{\mathcal F}_n$ and $\hat T_n$ provide the quantum computational representation of classical states and observables, respectively, which are two of the main ingredients of the QECD (see Sec.~\ref{secOverview}). By unitary equivalence, they have analogous convergence properties in the $n\to\infty$ limit as those of their matrix mechanical counterparts $\mathcal F_n$ and $T_n$ described in Sec.~\ref{secMatrixMechanical}.We also note that the evolution operator $\hat U^t_n $ at the quantum computational level can be equivalently obtained as a projection of the Koopman operator $U^t$ on $ \rkha $, i.e., 
\begin{equation}
    \label{eqHatUtN}
    \hat U^t_n = (\mathcal W_n \circ \bm \Pi_n \circ \bm \Pi ) U^t.
\end{equation} {
}

\subsection{\label{secWalsh}Factorizing the Hamiltonian in tensor product form}

In order for the representation of the dynamics on $\mathbb B_n$ to exhibit robust quantum parallelism, i.e., implementation on a quantum circuit of small depth, it is highly beneficial that the Hamiltonian $H_n$ can be decomposed as a sum of commuting operators in pure tensor product form, i.e., 
\begin{equation}
    \label{eqHDecomp}
    H_n = \sum_{j \in J_n} G_j = \sum_{j \in J_n} G_{1j} \otimes \cdots \otimes G_{nj},
\end{equation}
where $ [ G_i, G_j ] = 0 $ and $ G_{lj} : \mathbb B \to \mathbb B$ are mutually-commuting, single-qubit Hamiltonians. With such a decomposition, the unitary operator $ \hat U^t_n = e^{iH_nt}$ generated by $H_n$ factorizes as
\begin{equation}
    \label{eqUDecomp} \hat U^t_n = \exp\left(i \sum_{j \in J_n} G_j t\right) = \prod_{j \in J_n} \exp\left(\bigotimes_{l=1}^n i G_{lj}t\right).
\end{equation}
Thus, $ \hat U^t_n $ can be split into a composition of up to $2^n$ unitaries  $\exp(i G_jt)$ (depending on the number of nonzero terms $G_j $ in the right-hand side of~\eqref{eqHDecomp}), which can be applied in any order by commutativity of the $G_j$. Moreover, each unitary $\exp(i G_jt)$ has a generator of pure tensor product form, and thus can be represented as a quantum circuit with at most $n$ quantum gates for rotations of the individual qubits. 

In fact, as we will now show, using a Walsh operator representation \cite{WelchEtAl14}, for a dynamical system with pure point spectrum the decomposition in~\eqref{eqHDecomp} only has $n$ nonzero terms $G_j $, and for each nonzero term, the tensor product factorization $G_j = \bigotimes_{l=1}^n G_{lj} $ has all but one factors $G_{lj} $ equal to the identity. As a result, 
\begin{equation*}
    \exp\left(\bigotimes_{l=1}^n i G_{lj}t\right) = \bigotimes_{l=1}^n \exp(i G_{lj} t),     
\end{equation*}
and the decomposition in~\eqref{eqUDecomp} reduces to a tensor product of $n$ unitaries, 
\begin{equation}
    \label{eqUDecomp2}
    \hat U^t_n = \bigotimes_{l=1}^n \Xi^t_l \quad \mbox{with} \quad \Xi^t_l = \exp\left(i \sum_{j\in J_n} G_{lj} t\right).
\end{equation}
The key point about~\eqref{eqUDecomp2} is that $\hat U^t_n$ can be implemented via a quantum circuit of $n$ qubit channels with no cross-channel communication.  We will return to this point in Sec.~\ref{secMeas}.

\subsubsection{Walsh-Fourier transform and Walsh operators}

Classical states and observables of the dynamical system have been transformed into pure state density operators and self-adjoint operators on the $2^n$-dimensional Hilbert space $\mathbb B_n$ which is a tensor product of the single qubit quantum state spaces as given in \eqref{tensor}. We will employ the commonly used Dirac bra-ket notation \cite{Nielsen10} to denote vectors in  $\mathbb B_n$. We let $\{|0\rangle, |1\rangle \}$ be the standard orthonormal basis of the single-qubit Hilbert space $\mathbb B \simeq \mathbb C^2$ comprising of eigenvectors of the Pauli $Z$ operator, 
\begin{displaymath}
    Z \ket 0 = \ket 0, \quad Z |1\rangle = -|1\rangle,
\end{displaymath}
with 
\begin{displaymath}
Z= \begin{pmatrix} 
        1 & 0\\0 & -1
    \end{pmatrix}, \quad
    \ket 0 = 
    \begin{pmatrix}
        1\\0
    \end{pmatrix},
    \quad\mbox{and}\quad \ket 1 = 
    \begin{pmatrix} 0\\1 \end{pmatrix}.
\end{displaymath}
Thus, each vector $\ket\psi \in \mathbb B$ can be expanded in this basis as 
\begin{equation}
|\psi\rangle = \alpha |0\rangle +\beta |1\rangle \quad \mbox{with}\quad \alpha, \beta \in \mathbb C.     
\end{equation}

In order to arrive at the decomposition in~\eqref{eqUDecomp2}, we employ the approach developed in Ref.~\cite{WelchEtAl14}, which is based on discrete Walsh-Fourier transforms, and the associated Walsh operators, as follows. First, for any  non-negative integer $ j \in \mathbb N_0$, we let $\beta(j) = ( \beta_1(j),\ldots, \beta_l(j) ) \in \{ 0, 1 \}^l$ be its binary expansion; that is, 
\begin{equation*}
    j = \sum_{i=1}^l \beta_i(j) 2^{i-1}=\beta_1(j) 2^0+\beta_2(j) 2^1+ \ldots +\beta_l(j) 2^l,
\end{equation*}
where $l \in \mathbb N$ is the smallest positive integer such that $j \leq 2^l - 1$. For example, we have $\beta(0)=0$, $\beta(1)=1$, $\beta(2)=(0,1)$, $\beta(3)=(1,1)$, and $\beta(4)=(0,0,1)$. Moreover, for every real number $ u \in [ 0, 1 )$ we let $ \gamma( u ) = ( \gamma_1(u), \gamma_2(u), \ldots ) \in \{ 0, 1 \}^{\mathbb N}$ be its dyadic expansion, i.e., 
\begin{equation*}
    u = \sum_{i=1}^\infty \gamma_i(u) 2^{-i}=\frac{\gamma_1(u)}{2}+\frac{\gamma_2(u)}{4}+\frac{\gamma_3(u)}{8}+\dots.
\end{equation*}
Note that the most significant digit in $\beta(j)$ is the last one, $\beta_l(j)$, whereas the most significant digit in $ \gamma(u)$ is the first one, $ \gamma_1(u)$. 

With this notation, for every $ j \in \mathbb N_0$ we define the \emph{Walsh function} $ w_j : [ 0, 1 ) \to \{ 0, 1 \}$ as
\begin{displaymath}
    w_j(u) = (-1)^{\sum_{i=1}^l \beta_i(j ) \gamma_i(u)}.
\end{displaymath}
Furthermore, for any $ n \in \mathbb N_0 $ and $ j \in \{ 0, \ldots, 2^n -1 \}$, we define the {\em discrete Walsh function of order $n$}, $w^{(n)}_j : \{ 0, \ldots, 2^n -1 \} \to \{0, 1\}$ as 
\begin{displaymath}
    w^{(n)}_j( m) = w_j( m / 2^n ), \quad\mbox{with}\quad m=0, \dots, 2^n-1.
\end{displaymath}
It then follows that 
\begin{align*}
    w^{(n)}_j(m) &= (-1)^{\sum_{i=1}^l \beta_i(j ) \gamma_i(m / 2^n)} \\
    &= (-1)^{\sum_{i=1}^n \beta^{(n)}_i(j) \tilde \beta_{i}^{(n)}( m) }.  
\end{align*}
Here, $ \beta^{(n)}(j) = ( \beta_1, \ldots, \beta_l, 0, \ldots, 0 ) \in \{ 0, 1\}^n$ is the $n$-digit binary expansion of $j$ obtained by padding $\beta(j)$ to the right with zeros, as needed. Moreover, 
\begin{align*}
\tilde \beta^{(n)}(j) &= (\tilde \beta_1^{(n)}(j), \ldots, \tilde \beta_n^{(n)}(j))\\
&=(\gamma_1(j/2^m), \ldots, \gamma_n(j/2^m) ) 
\end{align*}
is the $n$-digit reversed binary representation of $m$. Thus, the exponent in the expression for $w^{(n)}_j( m)$ is given by the inner product between the $n$-digit binary expansion of $j$ with the \emph{bit-reversed} binary expansion of $m$. For example, with $n=2$ and $m=0,1,2,3$, we have  $w_0^{(2)}(m)=\{1,1,1,1\}$, $w_1^{(2)}(m)=\{1,1,-1,-1\}$, $w_2^{(2)}(m)=\{1,-1,1,-1\}$, and $w_3^{(2)}(m)=\{1,-1,-1,1\}$. 

Among the Walsh functions $w_j$, those with $ j =1, 2, 4, \dots, 2^l $ for $ l \in \mathbb N_0 $ are called \emph{Rademacher functions}, $R_l$, and satisfy
\begin{equation}
    \label{eqRad}
    w_{2^l}(u)\equiv R_l(u) = (-1)^{\gamma_l(u)}.  
\end{equation}
That is, $R_l(u)$ depends only on the $(l+1)$-th bit in the dyadic expansion of $u$. Using~\eqref{eqRad}, it follows that for any integer $ m \in \{ 0, \ldots, 2^n - 1 \} $, we have  
\begin{displaymath}
    \frac{m}{2^n} = \sum_{i=0}^{n-1} \frac{ 1 - R_{i+1}(m / 2^n)}{2^{i+2}},
\end{displaymath}
meaning that we can express the $i$-th bit in the dyadic decomposition of $ m / 2^n$ in terms of the $(i-1)$-th Rademacher function,
\begin{equation}
    \label{eqRad2}
    \gamma_i(m/ 2^n) = \frac{1 - R_{i-1}(m/2^n) }{ 2 }.
\end{equation}

It is known that the set $ \{ w_j \}_{j \in \mathbb N_0} $ forms an orthonormal basis of the Hilbert space $L^2([0,1])$ with respect to Lebesgue measure. In the discrete case, we let $L^2_n([0,1])$ be the $N$-dimensional Hilbert space, $N \equiv 2^n$, with respect to the normalized counting measure supported on $\{  0, 1 / N, 2 / N, \ldots, (N-1)/N \} $. Then, the set of discrete Walsh functions of order $n$, $ \{ w^{(n)}_j \}_{j=0}^{N-1} $, is an orthonormal basis of $L^2_n([0,1])$. One obtains
\begin{align*}
    f &= \sum_{j=0}^{N-1} \hat f_j w^{(n)}_j \in L^2_n([0,1]) \quad \mbox{with} \\
    \hat f_j &= \frac{1}{N} \sum_{m=0}^{N-1} \hat w_j^{(n)}(m) f(m/N).
\end{align*}
The map $ \mathsf F_n : L^2_n([0,1]) \to \mathbb C^{N}: f \mapsto ( \hat f_0, \ldots, \hat f_{N-1} ) $ is called the \emph{discrete Walsh-Fourier transform} of the function $f\in L^2_n([0,1])$.  

Next, consider the tensor product basis $ \{ \ket{\bm b} = \ket{b_1} \otimes \cdots \otimes \ket{b_n} \} $ of $\mathbb B_n$ with $\ket{b_i} \in \{|0\rangle, |1\rangle\}$, where the multi-index $ \bm b = (b_1, \ldots, b_n) \in \{ 0, 1 \}^n$ runs over all binary strings of length $n$. Whenever convenient, we will employ the notation $\ket b \equiv \ket{ \bm b } $, where $ \bm b = \tilde \beta^{(n)}(b)$. That is, $b$ is an integer in the range $0, \ldots, 2^n-1$, whose reversed binary representation is equal to $\bm b$, 
\begin{displaymath}
    b=\sum_{i=1}^n \tilde\beta^{(n)}_i(b) 2^{n-i} = \sum_{i=1}^n b_i 2^{n-i}. 
\end{displaymath}
For example, in a system with $n=3$ qubits  $|b\rangle=|6\rangle$ corresponds to $|\bm b\rangle=|110\rangle$, where the least significant bit is the one to the right. Note that $ \{ \ket b \}_{b=0}^{2^n-1}$ is also the standard quantum computational basis for an $n$-qubit problem in the Qiskit framework \cite{Qiskit20,Qiskit} that we will employ in Sec.~\ref{secExamples} \footnote{The qubit ordering in Qiskit is reverse to that in most textbooks on quantum computing.}. 

For every $ \bm b \in \{ 0, 1 \}^n$, we define the associated \emph{Walsh operator} $Z_{\bm b} : \mathbb B_n \to \mathbb B_n$ as
\begin{displaymath}
    Z_{\bm b} = Z^{b_1} \otimes Z^{b_2} \otimes \cdots \otimes Z^{b_n}.
\end{displaymath}
By construction, the $Z_{\bm b}$ form a collection of mutually-commuting, self-adjoint operators, which have pure tensor product form and are diagonal in the  $ \{ \ket{\bm b } \} $ basis of $\mathbb B_n$, i.e., 
\begin{displaymath}
    Z_{\bm b} \ket {\bm c} = \left(\prod_{i=1}^n (-1)^{b_i(c_i-1)} \right) \ket{ \bm c},
\end{displaymath}
where $\ket{\bm c}$ is again a quantum computational basis vector. For example, for $n=2$ qubits, the Walsh operator $Z_{\bm b}$ with $\bm b =\ket{0,1} $, and the basis vector $\ket{\bm c} = \ket{\bm b}$, one obtains 
\begin{displaymath}
    Z_{\bm b} \ket{\bm c}= (I\otimes Z) (\ket 0 \otimes \ket 1) =-\ket{\bm c}.
\end{displaymath}

It follows from a counting argument that the collection $ \{ Z_{\bm b} \}_{\bm b \in \{ 0, 1 \}^n} $ forms a basis of the vector space of operators in $ B( \mathbb B_n)$ which are diagonal in the $ \{ \ket{ \bm b} \} $ basis. In Ref.~\cite{WelchEtAl14}, it was shown that if $ A \in B(\mathbb B_n)$ is such a diagonal operator,
\begin{displaymath}
    A \ket{ \bm b } = a_{\bm b} \ket{ \bm b } \quad\mbox{with} \quad a_{\bm b} \in \mathbb C, 
\end{displaymath}
then it admits the expansion
\begin{equation}
    \label{eqWalshDecomp}
    A = \sum_{j=0}^{N-1} \hat f_j Z_{\beta^{(n)}(j)}, \quad \hat f_j \in \mathbb C,
\end{equation}
where the expansion coefficients $\hat f_j $ are the complex Walsh-Fourier coefficients $ ( \hat f_0, \ldots, \hat f_{N-1} ) = \mathsf F_nf $ of the function $f = ( f_0, \ldots, f_{N - 1} ) \in L^2_n([0,1]) $ with $ f_j = a_{\beta^{(n)}(j)}$. That is, $f_j$ is equal to the eigenvalue $a_{\bm b}$, where $\bm b$ is the $n$-digit binary representation of the integer $j$.

\subsubsection{Walsh representation of the Hamiltonian}

In order to effect the decomposition in~\eqref{eqHDecomp} for the generator-induced Hamiltonian from~\eqref{eqH}, let $ o : J_{n,d} \to \{ 0, \ldots, 2^{n/d } -1 \} $ be the enumeration on the index set $J_{n,d} $ from~\eqref{eqIdxJ} based on the standard order of integers, i.e.,  $o(-2^{n/d-1}) = 0, 1, \ldots,  2^{n/d} - 1 = o(2^{n/d-1})$. For example, $n=2$, $d=1$ gives $j\in J_{2,1}=\{-2,-1,1,2\}$, which is mapped to $o(j)=\{0,1,2,3\}$. The mapping of $J_{2,2}$ (for $n=2$, $d=2$) is displayed in first and third columns of Table~\ref{tableWalsh}.  We define $ W_n : \mathcal H_n \to \mathbb B_n$ as the unique (unitary) linear map such that for $j = ( j_1, \ldots, j_d )$,
\begin{equation}
    \label{eqWn}
    \begin{aligned}
        W_n \psi_j &= \ket{\bm b} \quad\mbox{with}\\
        \bm b &= \bm \eta(j) :=  ( \eta(j_1), \ldots, \eta(j_d)), \\
        \eta(j_i) &= \tilde\beta^{(n/d)}( o( j_i ) ).
   \end{aligned}
\end{equation}
That is, $W_n$ maps the basis element $\psi_j$ of $\mathcal H_n$ with multi-index $ j = ( j_1,\ldots, j_d) \in J_n $ to the tensor product basis element $\ket{\bm b} $, with $\bm b$ given by an invertible binary string encoding of $j$. Here, $\bm b$ is obtained as a concatenation $( \eta(j_1), \ldots, \eta(j_d)) $ of $d$ binary strings of length $n/d$, corresponding to the dyadic decompositions of $o(j_1), \ldots, o(j_d) $, respectively. See again Table~\ref{tableWalsh}, where we list the mapping for a two-dimensional torus with $2 = n/d$ qubits for each torus dimension.

Since $ \hat V_n \psi_j = i \omega_j \psi_j$ with $ \omega_j$ given by~\eqref{eqOmega}, we have 
\begin{equation}
    \label{eqHE}
    H_n \ket{\bm b} = \omega_{\bm \eta^{-1}(\bm b)} \ket{\bm b },
\end{equation}
Thus, in order to decompose $H_n$ into Walsh operators as in~\eqref{eqWalshDecomp}, we need to compute the discrete Walsh transform of the function $ h\in L^2_n([0,1]) $ with
\begin{equation}
    \label{eqHFunc}
    h( m / N ) = \omega_j, \quad j = \bm \eta^{-1}( \tilde \beta^{(n)}(m)). 
\end{equation}
This calculation is detailed in Appendix~\ref{appWalsh}. The eigenvalues $\omega_j$ for the example of a two-dimensional torus with $n=2$ qubits are listed in the fifth column of Table~\ref{tableWalsh}. 

%----------------------------------------------------------------------------
\begin{table}
    \centering
    \begin{tabularx}{\linewidth}{@{\extracolsep{\fill}}cccc}
        \hline\hline 
        $(j_1,j_2)$ & $(\eta(j_1), \eta(j_2))$ & $b$ & $\omega_j$ \\
        \hline
        $(-2,-2)$ & $((0,0),(0,0))$ & 0 & $-2 \alpha_1 -2 \alpha_2$\\
        $(-2,-1)$ & $((0,0),(0,1))$ & 1 & $-2 \alpha_1 -1 \alpha_2$\\
        $(-2,+1)$ & $((0,0),(1,0))$ & 2 & $-2 \alpha_1 +1 \alpha_2$\\
        $(-2,+2)$ & $((0,0),(1,1))$ & 3 & $-2 \alpha_1 +2 \alpha_2$\\

        $(-1,-2)$ & $((0,1),(0,0))$ & 4 & $-1 \alpha_1 -2 \alpha_2$\\
        $(-1,-1)$ & $((0,1),(0,1))$ & 5 & $-1 \alpha_1 -1 \alpha_2$\\
        $(-1,+1)$ & $((0,1),(1,0))$ & 6 & $-1 \alpha_1 +1 \alpha_2$\\
        $(-1,+2)$ & $((0,1),(1,1))$ & 7 & $-1 \alpha_1 +2 \alpha_2$\\

        $(+1,-2)$ & $((1,0),(0,0))$ & 8 & $+1 \alpha_1 -2 \alpha_2$\\
        $(+1,-1)$ & $((1,0),(0,1))$ & 9  & $+1 \alpha_1 -1 \alpha_2$\\
        $(+1,+1)$ & $((1,0),(1,0))$ & 10 & $+1 \alpha_1 +1 \alpha_2$\\
        $(+1,+2)$ & $((1,0),(1,1))$ & 11 & $+1 \alpha_1 +2 \alpha_2$\\

        $(+2,-2)$ & $((1,1),(0,0))$ & 12 & $+2 \alpha_1 -2 \alpha_2$\\
        $(+2,-1)$ & $((1,1),(0,1))$ & 13 & $+2 \alpha_1 -1 \alpha_2$\\
        $(+2,+1)$ & $((1,1),(1,0))$ & 14 & $+2 \alpha_1 +1 \alpha_2$\\
        $(+2,+2)$ & $((1,1),(1,1))$ & 15 & $+2 \alpha_1 +2 \alpha_2$\\
        \hline\hline
    \end{tabularx}
    \caption{\label{tableWalsh} Binary encodings $ \bm \eta( j ) = (\eta(j_1), \eta(j_2))$ and enumeration $ b = (\tilde\beta^{(n)})^{-1}( \bm\eta(j) ) = 0 \dots 2^{n}-1$  of the eigenfrequencies $\omega_j $ with multi-index $ j = (j_1,j_2) $ of a quasiperiodic system on a two-dimensional torus ($d=2$) with basis frequencies $\alpha_1$ and $\alpha_2$. The total number of qubits is $n=4$ with 2 qubits for each torus dimension.}
\end{table}
%----------------------------------------------------------------------------

By virtue of the decomposition in~\eqref{eqFOmega}, the only nonzero coefficients in the Walsh-Fourier transform $ \hat h = ( \hat h_0, \ldots, \hat h_{N-1} ) = \mathsf F_n h $ are the coefficients $ \hat h_j $ with $ j = 2^{l+(i-1)d} $ and $1 \leq l \leq n / d$, $ 1 \leq i \leq d $. Correspondingly, the only nonzero terms $ \hat h_j Z_{\beta^{(n)}(j)} $  in the Walsh operator expansion from~\eqref{eqWalshDecomp} for the Hamiltonian in~\eqref{eqHE} are those for which the binary string $ \bm\eta(j) $ has exactly one bit equal to 1 and the remaining $ n - 1 $ bits equal to 0. In particular, we have
\begin{align}
    \nonumber H_n &= \sum_{i=1}^d \sum_{l=0}^{\frac{n}{d}-1} \hat h_{2^{l+(i-1)n/d}} Z_{\beta(2^{l+(i-1)n/d})} \\
    \nonumber & = \hat h_{1} Z \otimes I \otimes I \otimes I \otimes \cdots \otimes I \\
    \nonumber &+ \hat h_{2} I \otimes Z \otimes I \otimes I \otimes \cdots \otimes I  \\
    \nonumber & + \hat h_{4} I \otimes I \otimes Z \otimes I \otimes \cdots \otimes I + \ldots \\
    \label{eqHDecompQuasiperiodic}& + \hat h_{2^{n-1}} I \otimes \cdots \otimes I \otimes Z.   
\end{align}

Equation~\eqref{eqHDecompQuasiperiodic} verifies the assertion made earlier that the decomposition of $H_n$ in~\eqref{eqHDecomp} can be arranged to have $n$ nonzero terms, each of which factorizes as a tensor product of $n$ operators, with all but one factors equal to the identity. Since 
\begin{displaymath}
    e^{i t I \otimes \cdots \otimes I \otimes Z \otimes I \cdots \otimes I } = I \otimes \cdots \otimes I \otimes e^{it Z} \otimes I \otimes \cdots \otimes I,
\end{displaymath}
we conclude that
\begin{equation}
    \hat U^t_n = e^{i H_n t} = \bigotimes_{k=1}^{2^{n-1}} \exp(i t \hat h_k Z),
\end{equation}
which is consistent with the decomposition in~\eqref{eqUDecomp2}.

\section{\label{secMeas}Projective measurement of observables}

In the classical setting, the process of obtaining the results of a computation is a straightforward readout of the state of the computer. In contrast, in quantum computing, extracting information from the system is a non-trivial process, as it must invariably confront with the intricacies of quantum measurement. In this section, we describe how the QECD performs probabilistic predictions of the evolution of classical observables through projective measurement of quantum computational observables. First, in Sec.~\ref{secIdealizedMeas}, we consider an idealized measurement scenario, where one has access to the spectral measure of the observable of interest. Then, in Sec.~\ref{secApproximateMeas} we develop an approximate measurement procedure based on the QFT, which yields asymptotically consistent results with the idealized measurement, while maintaining an exponential quantum advantage. Additional technical results are provided in Appendix~\ref{appQFT}.    

\subsection{\label{secIdealizedMeas}Idealized quantum measurement}

Our goal is to approximate the classical evolution $U^t f(x)$ through projective measurement of the quantum mechanical observable $\hat S_n := \hat T_n f$ on the quantum state 
\begin{equation}
    \label{eqRhoXNT}
    \hat \rho_{x,n}^{(t)} := \hat \Psi^t_n( \hat \rho_{x,n} ), \quad \hat\rho_{x,n} =  \hat{\mathcal F}_n(x), 
\end{equation}
where the representation maps $\hat T_n$ and $\hat{\mathcal F}_n$ are defined in~\eqref{eqFTN}, and the evolution map $\hat \Psi^t_n $ is defined in~\eqref{eqPsiUN} (see also Fig.~\ref{Scheme}). Since $ \hat S_n $ is a finite-rank, self-adjoint operator, it has a spectral resolution
\begin{equation}
    \hat S_n = \sum_{s \in \sigma(\hat S_n)} s P_s ,
\label{PVM0}    
\end{equation}
where $ \sigma(\hat S_n)$ is the spectrum of $\hat S_n $, i.e., the set of its eigenvalues, and $P_s \in B(\mathbb B_n)$ are the orthogonal projections onto the corresponding eigenspaces. For example, if $s \in \sigma(\hat S_n)$ is an eigenvalue of multiplicity 1 with a corresponding normalized eigenvector $\ket s$, then $ P_s$ is the rank-1 projection given by $ P_s = \ket s \bra s$. The collection $ \{ P_s \} $ defines a {\em projection-valued measure (PVM)} on $ \sigma( \hat S_n) $, i.e., a map $ \mathcal S_n: \Sigma(\hat S_n) \to B(\mathbb B_n) $ given by 
\begin{equation}
    \mathcal S_n( \Upsilon ) = \sum_{ s \in \Upsilon} P_s,
    \label{PVM1}    
\end{equation}
where $ \Sigma(\hat S_n)$ is the collection ($\sigma$-algebra) of all subsets of $ \sigma(S_n)$, and $\Upsilon $ a set in $\Sigma(\hat S_n)$.  A \emph{projective measurement} of $\hat S_n $ on the quantum state $ \hat \rho_{x,n}^{(t)} $ then corresponds to a randomly drawn eigenvalue $ \hat s $ from the spectrum $ \sigma(\hat S_n)$ with probability
\begin{displaymath}
    \mathbb P_{\hat \rho^{(t)}_{x,n}}( \hat s ) = \tr( \hat \rho_{x,n}^{(t)} P_s ).
\end{displaymath}
The random draws $ \hat s $ have expectation 
\begin{displaymath}
    \sum_{s \in \sigma(\hat S_n) } s \mathbb P_{\hat \rho^{(t)}_{x,n}}( s ) = \sum_{s\in \sigma(\hat S_n)} \tr( \hat \rho^{(t)}_{x,n} P_s ) =: f^{(t)}_n(x), 
\end{displaymath}
which is equivalent with~\eqref{eqUtApprox} by unitarity of the transformations from the matrix mechanical to quantum computational level. 

One can compute a Monte Carlo (ensemble) estimate of $ f^{(t)}_n(x) $ by performing a collection $\{ \hat s_1, \ldots, \hat s_K \}$ of measurements of $\hat S_n $ on $K $ independently and identically prepared quantum systems. The number $K$ is oftentimes referred to as the number of {\em shots}. The ensemble mean, 
\begin{equation}
    \label{eqEnsMean}
    \hat f^{(t)}_n(x) := \frac{1}{K} \sum_{k=1}^K \hat s_k 
\end{equation}
converges as $K \to \infty$ to the expectation $ f^{(t)}_n(x) $. The latter, converges in turn to the true value $U^t f(x)$ in the infinite-qubit limit, $ n \to \infty$; that is, we have
\begin{equation}
    \label{eqUConv}
    \lim_{n\to\infty} \lim_{K\to\infty} \hat f^{(t)}_n(x) = U^t f(x).
\end{equation}

\subsection{\label{secApproximateMeas}Approximate quantum measurement using quantum Fourier transforms}

Despite its theoretical consistency, the quantum measurement process described in Sec.~\ref{secIdealizedMeas} is not well-suited for practical quantum computation. The reason is that, in general, a quantum computing platform does not support the measurement of arbitrary PVMs such as $\mathcal S_n$ in~\eqref{PVM0}, and instead only allows measurement of the PVM associated with the quantum register. For an $n$-qubit system, the latter is defined as the PVM  $ \mathcal E_n : \Sigma( \{ 0, 1 \}^n ) \to B(\mathbb B_n)$ (cf.~\eqref{PVM1}),
\begin{displaymath}
    \mathcal E_n(\Upsilon ) = \sum_{\bm b \in \Upsilon}  E_{\bm b} \quad\mbox{with}\quad E_{\bm b} = \ket{\bm b} \bra{\bm b}, 
\end{displaymath}
where $ E_{\bm b}$ is the orthogonal projection along the computational basis vector $\ket{\bm b}$. 

In order to transform a measurement of $\mathcal S_n$  to an equivalent measurement of $\mathcal E_n $, one must apply a unitary transformation $ \hat \rho^{(t)}_{x,n} \mapsto \Lambda_n \hat \rho^{(t)}_{x,n} \Lambda^*_n$ to the quantum state $\hat \rho^{(t)}_{x,n}$, where $\Lambda_n : \mathbb B_n \to \mathbb B_n $ is a unitary map that diagonalizes $\hat S_n$, i.e., $\Lambda^*_n \hat S_n \Lambda_n$ is a diagonal operator in the $\{ \ket{\bm b} \}$ basis of $\mathbb B_n$. Two issues arise with this approach. First, $\Lambda_n$ is generally not known in closed form, and must be determined by solving an (exponentially large) eigenvalue problem for $\hat S_n$. Secondly, even if $\Lambda_n$ were known explicitly, it would likely be difficult to implement efficiently in a quantum circuit as it would generally be represented by a fully occupied matrix. 

To overcome these challenges, instead of working with $\Lambda_n$ directly, we will employ a different unitary map on $\mathbb B_n$  associated with the QFT. As is well known, the QFT on the $n$-qubit space $\mathbb B_n$ has a circuit implementation of size $O(n^2)$ and depth $O(n)$ \cite{Coppersmith94,MooreNilsson01,Nielsen10}. Thus, including it in the QECD pipeline does not result in loss of an exponential advantage in $n$ over classical computation. Crucially for our purposes, moreover, the class of operators $\hat S_n $ induced from multiplication operators $\pi f$ by classical observables on $X$ turns out to be approximately diagonalized by the QFT, with an error that vanishes in a suitable asymptotic limit. 

In more detail, for any $n \in \mathbb N$, let $\mathfrak F_n : \mathbb B_n \to \mathbb B_n$ be the Fourier operator on $\mathbb  B_n$, defined as
\begin{equation}
    \label{eqQFT}
    \mathfrak  F_n \ket m = \frac{1}{\sqrt {2^n}} \sum_{p=0}^{2^n-1} e^{-2\pi i p m /2^n} \ket p,
\end{equation}
where $\ket m$ and $\ket p$ are again two basis vectors of $\mathbb B_n$, parameterized by integers $m $ and $p $, respectively, by conversion of the corresponding binary sequences. Moreover, for $n$ divisible by the state space dimension $d$, let $\mathfrak F_{n,d} : \mathbb B_n \to \mathbb B_n$ be the tensor product operator defined as 
\begin{equation}
    \label{eqQFTProd}
    \mathfrak F_{n,d} = \underbrace{\mathfrak F_{n/d} \otimes \cdots \otimes \mathfrak F_{n/d}}_d,
\end{equation}
and $\bm{\mathfrak F}_{n,d} : B(\mathbb B_n) \to B(\mathbb B_n)$ the induced operator on quantum computational observables, given by
\begin{equation}
    \bm{\mathfrak F}_{n,d} A = \mathfrak F_{n,d} A \mathfrak F_{n,d}^*.
    \label{eqQFTInd}
\end{equation}
In Appendix~\ref{appQFT}, we show that $\tilde S_n := \mathfrak F_{n,d} \hat S_n \mathfrak F_{n,d}^* $ is an approximately diagonal operator in the computational basis $ \{ \ket{ \bm b } \} $. In particular, decomposing $\bm b = ( \bm b^{(1)}, \ldots, \bm b^{(d)} ) $, where $\bm b^{(i)} = ( b^{(i)}_1, \ldots, b^{(i)}_{n/d} ) $ are binary strings of length $n/d$, and defining the points 
\begin{equation}
    \label{eqGrid}
    x_{\bm b} = ( \theta_{\bm b^{(1)}}, \ldots, \theta_{\bm b^{(d)}}) \in \mathbb T^d 
\end{equation}
with the canonical angle coordinates 
\begin{equation*}
    \theta_{\bm b^{(i)}} = \frac{2\pi ({\tilde\beta}^{(n/d)})^{-1}(\bm b^{(i)})}{2^{n/d}}, 
\end{equation*}
we have  
\begin{equation}
    \label{eqApproxDiag}
    \tilde S_n \ket{\bm b} = \tilde s_{\bm b} \ket{\bm b} + \ket{ r_{n\bm b} }, \quad \tilde s_{\bm b} = f( x_{\bm b}). 
\end{equation}
Here, $ \ket{ r_{n\bm b} }$ is a residual that vanishes as $n\to\infty$, and $L: \rkha \to \rkha$ is the self-adjoint, diagonal operator defined in Appendix~\ref{appConsistency} (see also Sec.~\ref{secRKHSRep}). Effectively, the points $x_{\bm b}$ define a uniform grid on the $d$-torus $\mathbb T^d$, indexed by the $n$-digit binary strings $\bm b$. The quantities $\tilde s_{\bm b} $ can thus be interpreted as approximate eigenvalues of $\tilde S_n$, which can be obtained from \emph{classical} measurement of $f$ at the points $x_{\bm b}$, avoiding the need to solve an exponentially large eigenvalue problem for $\hat S_n$.

By virtue of these facts, and since
\begin{displaymath}
    \tr(\hat\rho^{(t)}_{x,n} \hat S_n) = \tr(\tilde\rho^{(t)}_{x,n} \tilde S_n),
\end{displaymath}
with
\begin{equation}
    \label{eqRhoQFT}
    \tilde \rho^{(t)}_{x,n} = \bm{\mathfrak F}_{n,d} \hat \rho^{(t)}_{x,n},
\end{equation}
we can approximate a measurement of $\hat S_n$ on the state $\hat \rho^{(t)}_{x,n}$ by a measurement of the PVM $\mathcal E_n$ on the state $\tilde \rho^{(t)}_{x,n}$. The latter measurement returns a random string $\bm b \in \{ 0,1 \}^n$ with probability 
\begin{displaymath}
    \mathbb P_{\tilde \rho^{(t)}_{x,n} }(\bm b ) = \tr(\tilde \rho^{(t)}_{x,n} E_{\bm b} ) = \bra{\bm b} \tilde \rho^{(t)}_{x,n} \ket{\bm b}, 
\end{displaymath}
inducing a sample $ \tilde s_{\bm b} = f(x_{\bm b}) $. Analogously to~\eqref{eqEnsMean}, we estimate $U^tf(x)$ by forming an ensemble of $K$ independent measurements $\bm b_1, \ldots, \bm b_K$ of $\mathcal E_n$, and computing the ensemble mean by
\begin{equation}
    \hat f^{(t)}_n(x) := \frac{1}{K} \sum_{k=1}^K \tilde s_{\bm b_k}. 
    \label{approxObs}
\end{equation}

Further details on this approximation, such as the proof of asymptotic consistency, can be found in Appendix~\ref{appQFT}. Here, we note that due to errors associated with the QFT-based measurement process, the convergence of $\hat f^{(t)}_n$ to $U^t f$ is not unconditional, but requires taking a sequence of decreasing RKHA parameters $\tau$ (unlike the limit in~\eqref{eqUConv} which holds for any $ \tau > 0 $). It should also be noted that it is possible to simulate multiple classical observables using the same circuit and ensemble of quantum measurements $\{ \bm b_1, \ldots, \bm b_K \} $. That is, to simulate the evolution of a different observable $ g: X \to \mathbb C $, we use the $\bm b_k$ to generate samples $ \dtilde{s}_{\bm b_k} = g( x_{\bm b_k}) $, and estimate $U^t g(x)$ by $\hat g^{(t)}_n(x) := \sum_{k=1}^K \dtilde s_{\bm b_k} / K $, analogously to~\eqref{approxObs}.

\section{\label{secPreparation}State preparation}

Besides measurement of observables, the preparation, or loading, of the quantum state representing the input (initial conditions) to a quantum computer is challenging. In a typical scenario involving an $n$-qubit computation, the register of a quantum computer is initialized with a state vector associated with an unentangled tensor product state,  $ \ket{\bm 0} \equiv \ket 0^{\otimes n}$. The desired initial state must be prepared by applying a unitary transformation (encoder) to $\ket{\bm 0}$, which may in general require a circuit of exponential depth in $n$ when the algorithm is broken down to elementary gate operations \cite{Benedetti2019,Markovic2020}. This poses a potentially significant obstruction to the scalability of quantum computational algorithms. 

In QECD, our task is to prepare the quantum state $\hat \rho_{x,n} = \hat{\mathcal F}_n(x)$ from~\eqref{eqRhoXNT} associated with the classical initial condition $x \in X$. This state is a pure state,
\begin{displaymath}
    \hat \rho_{x,n} = \ket{\hat \xi_{x,n}} \bra{\hat\xi_{x,n}},  
\end{displaymath}
where the state vector $\ket{\hat\xi_{x,n}} = W_n \xi_{x,n}$ is obtained by application of the unitary $W_n : \mathcal H_n \to \mathbb B_n$ from~\eqref{eqWn} on the normalized RKHS feature vector $\xi_{x,n}$ from~\eqref{eqRhoN}. Specifically, we have   
\begin{displaymath}
    \xi_{x,n} = \frac{k_{x,n}}{\sqrt{\kappa_n}} = \frac{1}{\sqrt{\kappa_n}} \sum_{j\in J_{n}} \psi_j^*(x) \psi_j,
\end{displaymath}
and thus
\begin{align}
    \ket{\hat \xi_{x,n}} &= W_n \xi_{x,n} 
    =\sum_{j\in J_{n}} \frac{\psi^*_j(x)}{\sqrt{\kappa_n}}W_n \psi_j \nonumber\\&= \sum_{\bm b \in \{ 0, 1 \}^n} \frac{\psi^*_{\bm\eta^{-1}(\bm b)}(x)}{\sqrt{\kappa_n}} \ket{\bm b}.
    \label{initial}
\end{align}
We now describe how, in the limit of small RKHA parameter $\tau$, this state can be prepared to any degree of accuracy using a circuit of size $O(n)$ and depth $O(1)$. 

First, let $\ket \Omega \in \mathbb B_n$ be the state vector associated with a uniform superposition of the quantum computational basis vectors,    
\begin{displaymath}
    \ket \Omega = \frac{1}{\sqrt{N}} \sum_{\bm b \in \{ 0, 1 \}^n} \ket {\bm b}.
\end{displaymath}
The state vector $\ket \Omega$ can be prepared from $\ket{\bm 0}$ using a circuit of depth 1, associated with an $n$-fold tensor product of Hadamard gates, i.e., 
\begin{equation}
    \ket \Omega = \left( \bigotimes_{i=1}^n \mathsf H \right) \ket{\bm 0}\,,
    \label{Hadamard}
\end{equation}
where $\mathsf H : \mathbb B \to \mathbb B $ is the Hadamard gate, represented by the matrix  
\begin{displaymath}
    \mathsf H = \frac{1}{\sqrt 2} 
    \begin{pmatrix}
        1 & 1 \\
        1 & - 1
    \end{pmatrix}.
\end{displaymath}
We will come back to this point in Sec.~\ref{secIBM} when the algorithm is implemented on an actual quantum computer.  

Next, recall that the basis functions $\psi_j$ of $\mathcal H_n$ have the form $\psi_j = e^{-\tau \lvert j \rvert_p/2} \phi_j$, where the $\phi_j$ are Fourier functions on the abelian group $X=\mathbb T^d$ (see Sec.~\ref{secRKHA}). Since the Fourier functions are characters of the group, they take the value $\phi_j(e) = 1 $ on the identity element $e \in X$ (the point with angle coordinates $\theta = 0 $), and thus 
\begin{displaymath}
    \ket{\hat \xi_{e,n}} = \sum_{\bm b \in \{ 0, 1 \}^n} \frac{e^{-\tau\lvert\bm\eta^{-1}(\bm b)\rvert_p/2}}{\sqrt{\kappa_n}} \ket{\bm b}.
\end{displaymath}
It follows that
\begin{equation}
    \label{eqOmegaApprox}
    \left\lVert \ket{\hat\xi_{e,n}} - \ket{\Omega} \right\rVert^2_{\mathbb B_n} = \sum_{\bm b \in \{ 0, 1 \}^n} \left\lvert \frac{1}{\sqrt N} - \frac{e^{-\tau\lvert\bm\eta^{-1}(\bm b)\rvert_p/2}}{\sqrt{\kappa_n}}\right \rvert^2,
\end{equation}
and noting that $\lim_{\tau\to0}\kappa_n = N$ (see~\eqref{eqRhoN}), we conclude that, for fixed $n$, $\ket{\hat \xi_{e,n}}$ converges to $\ket\Omega$ as $\tau \to 0$. In particular, since $\ket\Omega$ can be efficiently prepared via~\eqref{Hadamard}, we can efficiently approximate $\ket{\hat \xi_{e,n}}$ by $\ket\Omega$ to arbitrarily high precision. 

We now claim that every state vector $\ket{\hat \xi_{x,n}}$ from~\eqref{initial} can be reached efficiently from $\ket{\hat \xi_{e,n}}$ by applying a suitable unitary Koopman operator. Indeed, letting $S^x : \mathcal H \to \mathcal H$ be the shift operator by $ x= (\theta^1,\ldots,\theta^d) \in \mathbb T^d $, i.e.,
\begin{displaymath}
    (S^x f)(y) = f(x+y),
\end{displaymath}
we have that  $S^x = U^t$, where $U^t$ is the Koopman operator for any time $t$ and rotation frequencies $\alpha_1,\dots,\alpha_d$ such that $ x = (\alpha_1 t, \dots, \alpha_d t)$. Thus, if $\hat S^x_n : \mathbb B_n \to \mathbb B_n$ is the unitary operator induced at the quantum computational Hilbert space $\mathbb B_n$ by $S^x$ (cf.~\eqref{eqHatUtN}), 
\begin{displaymath}
    \hat S^x_n = (\mathcal W_n \circ \bm\Pi_n \circ \bm \Pi) S^x,
\end{displaymath}
we can implement $\hat S^x_n$ with a circuit of size $O(n)$ and depth $O(1)$ using an analogous approach to that used for the Koopman operator. In particular, by translation invariance of the kernel $k$ (see~\eqref{eqKTrans}),  we have $ \xi_x = S^{-x} \xi_e $, and thus $ \ket{\hat \xi_{x,n} } = \hat S^{-x}_n \ket{\hat \xi_{e,n}} $. Therefore, the state vector $\ket{\hat \xi_{x,n} }$ can be obtained efficiently by application of that circuit to $\ket{\hat \xi_{e,n}}$. 

Consider now the state vector 
\begin{equation}
    \label{eqDHatXi}
    \ket{\check{\xi}_{x,n}} := \hat S^{-x}_n \ket \Omega.
\end{equation}
We have
\begin{align*}
    \left\lVert \ket{\check{\xi}_{x,n}} - \ket{\hat\xi_{x,n}} \right\rVert_{\mathbb B_n} &= \left\lVert \hat S_n^{-x}\ket{\Omega} - \hat S_n^{-x}\ket{\hat\xi_{e,n}} \right\rVert_{\mathbb B_n} \\
    &= \left \lVert \ket \Omega - \ket{\hat{\xi}_{e,n}} \right\rVert_{\mathbb B_n},
\end{align*}
where we have used the unitarity of $\hat S^{-x}_n$ to obtain the last equality. By~\eqref{eqOmegaApprox}, it follows that as $\tau \to 0 $ at fixed $n$, $\ket{\check{\xi}_{x,n}}$ converges to $ \ket{\hat \xi_{x,n}} $. We therefore conclude that for any error tolerance $\epsilon$ there exists $\tau>0$ such that the desired initial state vector, $ \ket{\hat \xi_{x,n}}$, is approximated by $ \ket{\check{\xi}_{x,n}}$ with an error of at most $\epsilon$ in the norm of $\mathbb B_n$. Moreover, the state vector $ \ket{\check{\xi}_{x,n}}$ can be prepared by passing the initial quantum computational state vector $\ket{\bm 0}$ through a circuit of size $O(n)$ and depth $O(1)$. As with the QFT-based measurement scheme (see Sec.~\ref{secApproximateMeas} and Appendix~\ref{appQFT}), as $n\to\infty$, errors due to approximation of $ \ket{\hat \xi_{x,n}}$ by $ \ket{\check{\xi}_{x,n}}$ can be controlled by taking a decreasing sequence of RKHA parameters $\tau$.

%In summary, the QECD framework described in this paper  has three sources <++>
%========================================================================================
\section{\label{secExamples}Simulated quantum circuit experiments}

In this section, we demonstrate the performance of the QECD framework with simulated quantum circuit experiments implemented in the ideal Qiskit Aer simulator~\cite{Qiskit20,Qiskit}. We consider a periodic example on the circle (Sec.~\ref{secCircle}), as well as a quasiperiodic system on the 2-torus (Sec.~\ref{secTorus}). In both cases, we compare the mean from an ensemble of quantum measurements with the true dynamical evolution of representative classical observables. The numerical results, displayed in Figs.~\ref{dim1} and~\ref{dim2} for the one- and two-dimensional examples, respectively, are in good agreement with the theory developed in Secs.~\ref{secQuantumMechanical}--\ref{secMeas}.

\subsection{\label{secCircle}Circle rotation}
%---------------------------------------------------------------
\begin{figure*}[t]
    \centering
    \includegraphics[width=.95\linewidth]{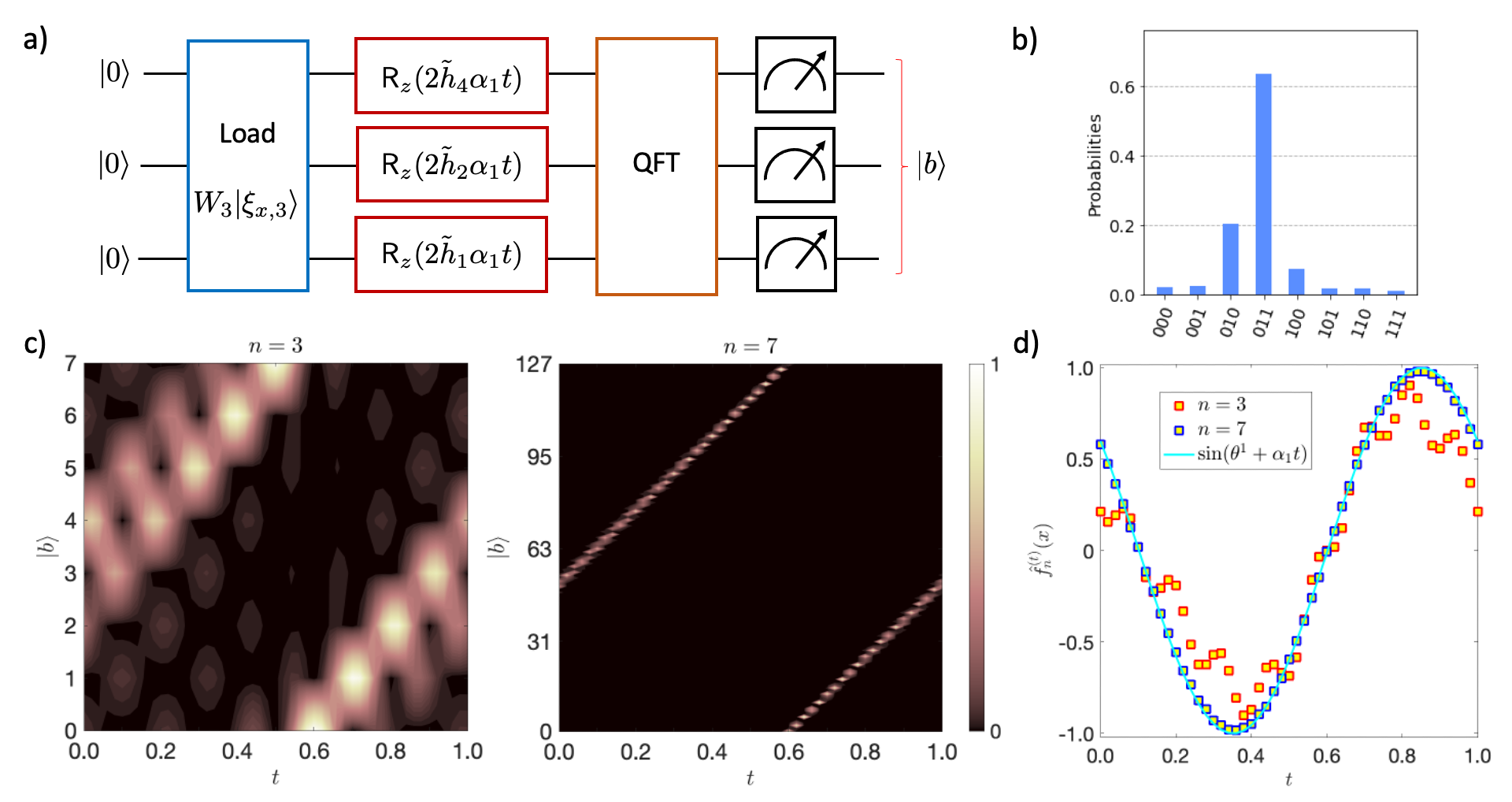}
    \caption{Quantum circuit implementation of the 3- and 7-qubit approximation of a circle rotation with frequency $\alpha_1 = 2\pi$ in the ideal Qiskit Aer environment. (a) Circuit diagram with $n=3$ qubits, comprising (from left to right) of state vector load, Koopman evolution over time $t$ using $\mathsf R_z$ gates, quantum Fourier transform (QFT), and measurement. (b) Empirical distribution of an ensemble of $K=10^6$ projective measurements (shots) of the projection-valued measure (PVM) associated with the computational basis vectors $\ket{\bm{b}} \equiv \ket b$ for $n=3$ and $t=0.94$. (c) Temporal evolution of the empirical probability distributions for $n=3$ and 7. (d) Reconstruction of the classical observable $f^{(t)}(x)=\sin(x(t))=\sin(\theta^1+\alpha_1 t)$ from the ensemble means, $\hat f^{(t)}_n(x)$. The analytical result $f^{(t)}(x)$ is plotted as a cyan solid line. In Panels (b)--(d), the initial condition is $x=\theta^1 = 2.5$ and the reproducing kernel Hilbert algebra (RKHA) parameters are $p=\tau=1/4$. Measurements are performed at a fixed timestep $\Delta t = 0.02$. In Panels~(b) and~(c), the computational basis vectors $\ket b$ are indexed by an integer $b$ in the range $0, \dots, 2^{n-1}$.}
    \label{dim1}
\end{figure*}
%---------------------------------------------------------------
\begin{figure*}
    \centering
    \includegraphics[width=.45\linewidth]{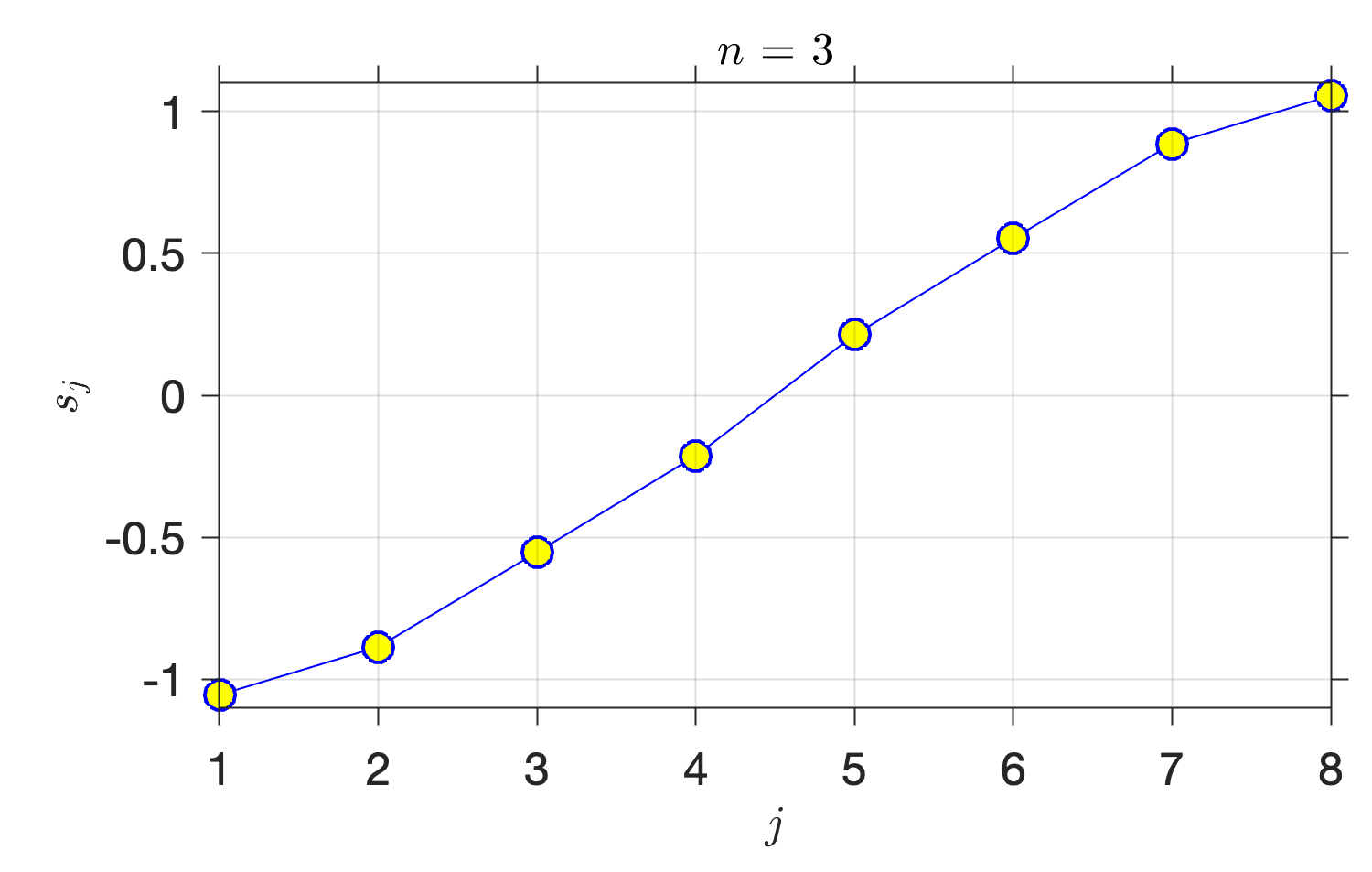} \qquad \includegraphics[width=.45\linewidth]{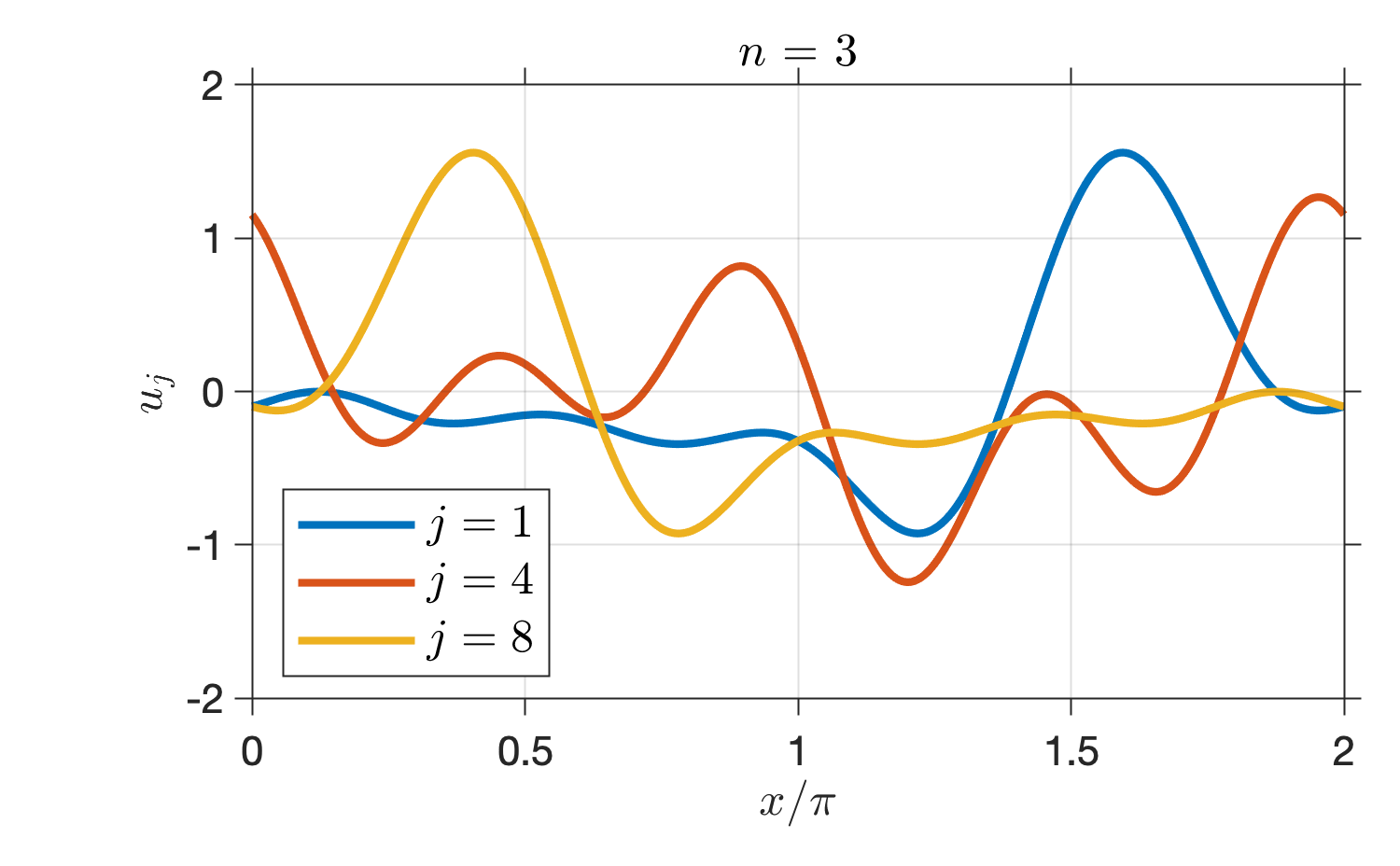} 
    \includegraphics[width=.45\linewidth]{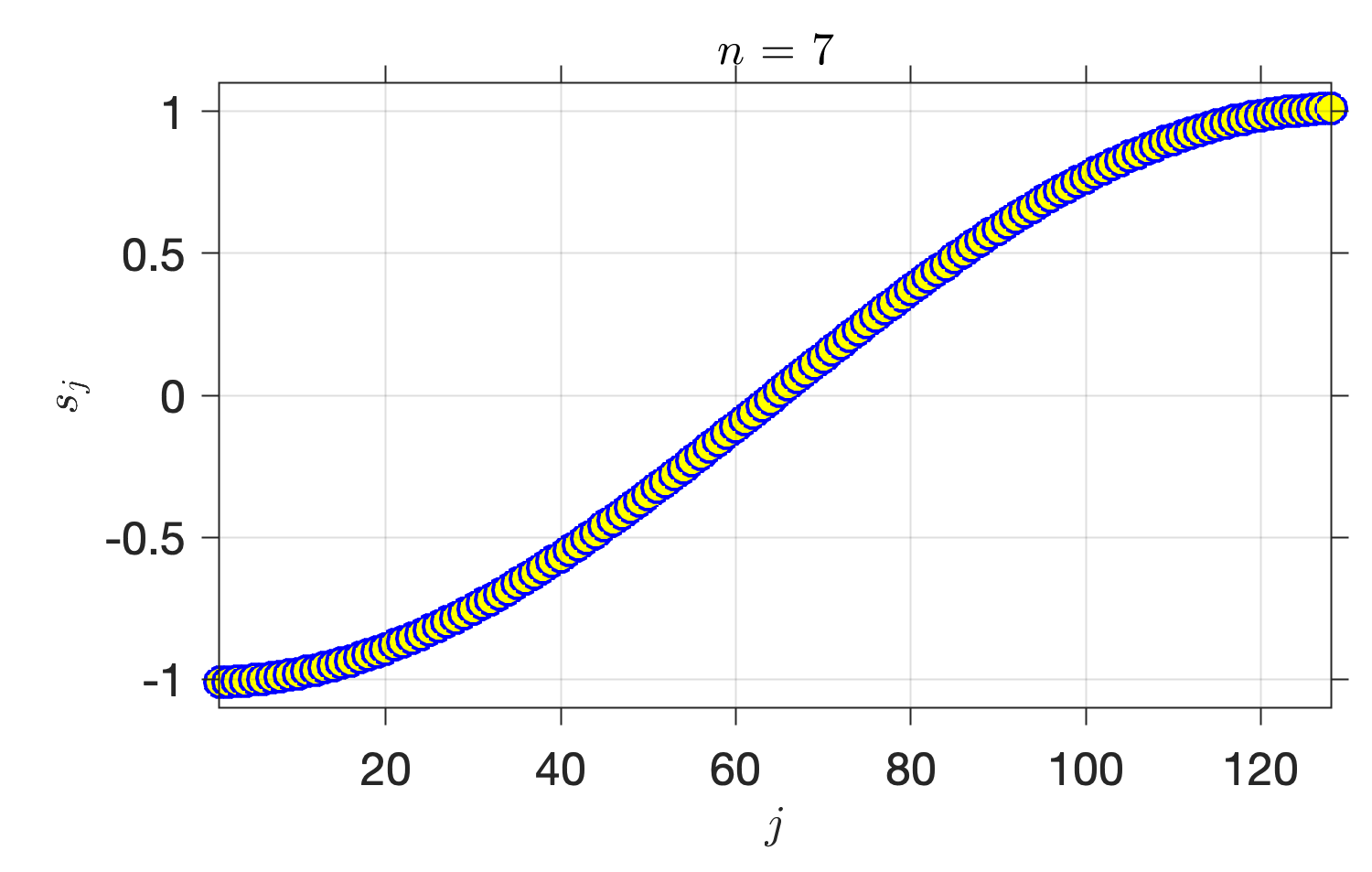} \qquad \includegraphics[width=.45\linewidth]{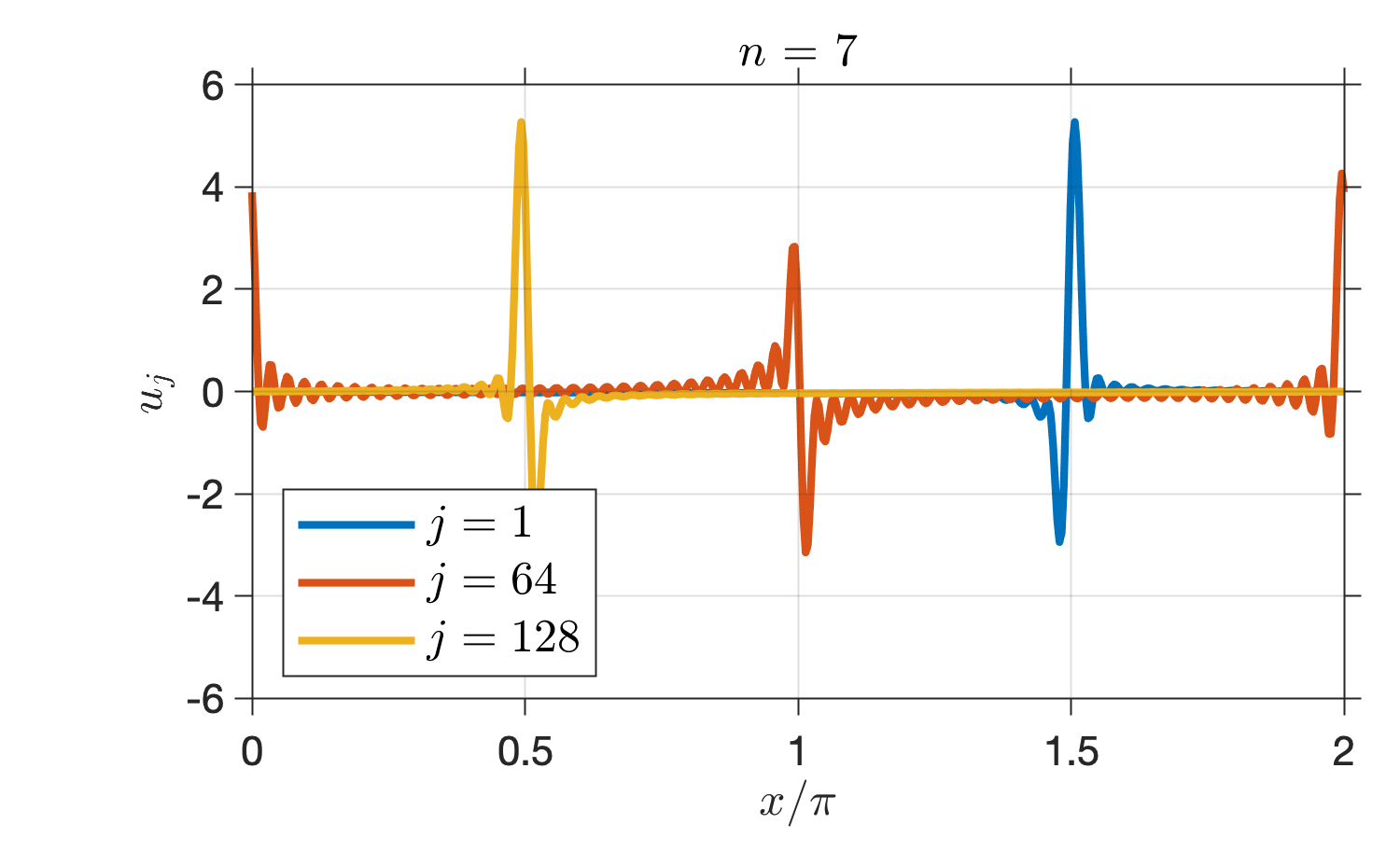} 
    \caption{\label{figSpecEig}Eigenvalues $s_j$ (left-hand column) and representative eigenfunctions $u_j : S^1 \to \mathbb R$ (right-hand column) of the self-adjoint operators $S_n = T_n f$ representing the classical observable $ f(x) = \sin x$ on the circle for the qubit numbers $n=3$ (top row) and 7 (bottom row). The RKHA parameters are $p= \tau = 1/4$ as in Fig.~\ref{figMultOp}. The index $j$ runs from 1 to $2^n$. Notice that as $n$ increases the spectra of $S_n$ provide an increasingly dense sampling of the range of values of $f$ (i.e., the interval $[-1,1]$), and the eigenfunctions $u_j(x)$ become increasingly localized around values of $x$ for which $ f(x) \approx s_j $.}
\end{figure*}

According to~\eqref{eqTorusRotation}, in dimension $d=1$ the orbits of the dynamics are given by 
\begin{displaymath}
    x(t) = \Phi^t(x) = (\theta^1 + \alpha_1 t) \mod 2 \pi,
\end{displaymath}
where $\alpha_1$ is the frequency parameter and $x=\theta^1$ the initial condition. We set $\alpha_1 = 2 \pi $, so the orbits have period $ 2 \pi / \alpha_1 = 1$. We seek to approximate the evolution of a real-valued observable $f : S^1 \to \mathbb R $ on the orbit starting at $x$, which is represented using the Koopman operator as
\begin{displaymath}
    f^{(t)}(x) = U^t f(x) = f(\Phi^t(x)) = f(\theta^1+\alpha_1 t ).  
\end{displaymath}
In this experiment, we consider the bandlimited observable $f(x) = \sin x$. 

The quantum circuit output by the compiler, displayed graphically in Fig.~\ref{dim1}(a), consists of the following four logical stages:
\begin{enumerate}
    \item A load stage, where the initial quantum state $\hat \rho_{x,n} = \hat{\mathcal F}_n(x)$ is prepared using the quantum feature map $\hat{\mathcal F}_n$ in~\eqref{eqFTN}.
    \item A dynamical evolution stage, which evolves $\hat \rho_{x,n}$ to the state $ \hat \rho^{(t)}_{x,n} = \hat\Psi^t_n(\rho_{x,n})$ using the evolution operator $\hat \Psi^t_n$ in~\eqref{eqPsiUN}.  
    \item A QFT stage, rotating $\hat \rho^{(t)}_{x,n}$ to the state $\tilde\rho^{(t)}_{x,n} = \bm{\mathfrak F}_{n,d} \hat\rho^{(t)}_{x,n}$ using the Fourier operator in~\eqref{eqQFTProd}. 
    \item A measurement stage, measuring the quantum-computational PVM $\mathcal E_n$ on the state $\hat \rho^{(t)}_{x,n}$. The quantum mechanical approximation $ \hat f^{(t)}_n(x)$ of $f^{(t)}(x)$ is then obtained as an ensemble mean of $K$ independent shots using~\eqref{approxObs}.
\end{enumerate}
The circuit is parameterized by three parameters, namely the RKHA parameters $p$ and $\tau$ and the number of qubits $n$. We set $ p = \tau = 1/4 $, and consider experiments with $n = 3$ and $n=7$ qubits, corresponding to the quantum computational Hilbert spaces $\mathbb B_3$ and $\mathbb B_7$ of dimension $N=2^3 = 8 $ and $N = 2^7 = 128$, respectively. Another input parameter is the evolution time $t$, which we set to integer multiples of a fixed timestep $\Delta t = 0.02$ for purposes of visualization.

Since all quantum states in the pipeline are pure, in practice we implement the circuit as a sequence of operators on the corresponding state vectors. First, the initial state is given by 
\begin{displaymath}
    \hat \rho_{x,n} = \ket{\hat \xi_{x,n}} \bra{\hat\xi_{x,n}},  
\end{displaymath}
where the state vector $\ket{\hat\xi_{x,n}} = W_n \xi_{x,n}$ is obtained by application of the unitary $W_n : \mathcal H_n \to \mathbb B_n$ from~\eqref{eqWn} on the normalized RKHS feature vector $\xi_{x,n}$ from~\eqref{eqRhoN}. See also~\eqref{initial}. 
We note that in these experiments the state vector $ \ket{\hat\xi_{x,n}}$ is loaded into the quantum register ``exactly'', using an amplitude encoding scheme applied to the initial state vector $\ket{\bm 0}$ (see Fig.~\ref{dim1}(a)), as opposed to the efficient approximate scheme described in Sec.~\ref{secPreparation}. In particular, we loaded $ \ket{\hat\xi_{x,n}}$ using the Qiskit function \texttt{QuantumCircuit.initialize}. We will discuss experiments utilizing the preparation approach of Sec.~\ref{secPreparation} in Sec.~\ref{secIBM}.

The next step is the unitary Koopman evolution, given by
\begin{displaymath}
    \hat\rho^{(t)}_{x,n} = \hat \Psi^t_n(\hat\rho_{x,n}) = \hat U^{t*}_n \ket{\hat \xi_{x,n}} \bra{\hat \xi_{x,n}} \hat U^t_n.
\end{displaymath}
Here, $\hat U^t_n = e^{it H_n}$ is the unitary operator in~\eqref{eqUDecomp2}, which is generated by the Hamiltonian $H_n$ with the Walsh factorization in~\eqref{eqHDecompQuasiperiodic}. We have $\hat\rho^{(t)}_{x,n} = \ket{\hat\xi^{(t)}_{x,n}} \bra{\hat\xi^{(t)}_{x,n}}$ with 
\begin{displaymath}
    \ket{\hat\xi^{(t)}_{x,n}} = \hat U^{t*}_n \ket{\hat \xi_{x,n}} = e^{-iH_n t} \ket{\hat \xi_{x,n}}.    
\end{displaymath}
Therefore, our circuit implements the transformation $\ket{\hat \xi_{x,n}} \mapsto \ket{\hat\xi^{(t)}_{x,n}}$, i.e.,     
\begin{align*}
    \ket{\hat\xi^{(t)}_{x,n}}&= \sum_{b=0}^{2^n-1} \frac{\psi^*_{o^{-1}(b)}(x)}{\sqrt{\kappa_n}} e^{-itH_n}\ket b \\
    &= \sum_{b=0}^{2^n-1} \frac{\psi^*_{o^{-1}(b)}(x)}{\sqrt{\kappa_n}} \left[ \bigotimes_{l=0}^{n-1} \exp(-it\alpha_1 \tilde h_{2^l} Z) \right] \ket b,
\end{align*}
where $\tilde h_{2^l} = \hat h_{2^l}/\alpha_1$, and $\hat h_{2^l}$ are the Walsh-Fourier coefficients in~\eqref{eqHDecompQuasiperiodic}. In more detail, using~\eqref{eqHWalsh} with $d=1$ and $n=3$, we obtain that all coefficients $\tilde h_{2^l}$ are zero except from $\tilde h_1=-5/2$, $\tilde h_2 =-1$, and $\tilde h_4=-1/2$. For $n=7$, the seven non-vanishing coefficients are $\tilde h_1=-65/2$, $\tilde h_2 =-16$, $\tilde h_4=-8$, $\tilde h_8=-4$, $\tilde h_{16}=-2$, $\tilde h_{32}=-1$, and $\tilde h_{64}=-1/2$. The implementation of this second step on the quantum computer is done for each qubit channel separately, as seen in Fig.~\ref{dim1}(a), by a $\mathsf R_z$ rotation gate given by
\begin{equation*}
    \mathsf R_z(\vartheta)=e^{-i\vartheta Z/2}=
    \begin{pmatrix}
        e^{-i\vartheta/2} & 0\\
        0 & e^{i\vartheta/2}
    \end{pmatrix}.
\end{equation*}
Specifically, we have
\begin{align*}
    \exp(-it\alpha \tilde h_{2^l} Z) = \mathsf R_z(2 \alpha t \tilde h_{2^l}).
\end{align*}

The third step is the application of the QFT, which results to 
\begin{displaymath}
    \tilde \rho^{(t)}_{x,n} = \bm{\mathfrak F}_{n,1} \hat \rho^{(t)}_{x,n} = \ket{\tilde \xi^{(t)}_{x,n}} \bra{\tilde \xi^{(t)}_{x,n}}, 
\end{displaymath}
where $ \ket{\tilde \xi^{(t)}_{x,n}} = \mathfrak F_{n,1} \ket{\hat \xi^{(t)}_{x,n}}  $. We again operate at the level of state vectors, effecting the transformation $ \ket{\hat \xi^{(t)}_{x,n}} \mapsto \ket{\tilde \xi^{(t)}_{x,n}}$ using a standard QFT circuit. The subsequent measurement of the PVM $\mathcal E_n$ on the state represented by $\ket{\tilde \xi^{(t)}_{x,n}}$ for $K$ shots leads to an empirical probability distribution over the binary strings $\bm b \in \{ 0, 1 \}^n $ (which index the basis vectors $\ket b \equiv \ket{\bm b}$), depicted in Fig.~\ref{dim1}(b) for a representative evolution time $t$. In Fig.~\ref{dim1}(c), we display the time evolution of this probability distribution for $K=10^6$ shots and $n=3$ and $n=7$ qubits. Notice that as $n$ increases, the probability distribution becomes increasingly concentrated around straight lines that periodically fold wrap around the set $b = 0, \ldots, 2^n$ indexing the $\ket b$ vectors. This is a manifestation of the fact that the time-dependent quantum state $\tilde \rho^{(t)}_{x,n}$ ``tracks'' the underlying classical state $x(t)$. 

Figure~\ref{dim1}(d) displays the true ($f^{(t)}(x)$) and simulated ($\hat f^{(t)}_n(x)$) evolution of the observable $f(x)=\sin x$ over the time interval $ t \in [0,1]$ starting from the initial condition $x = \theta^1 = 2.5$. The simulated evolution $\hat f^{(t)}_n(x)$, which is again obtained using $K= 10^6$ shots, is seen to be in good agreement with the true signal for $n=7$ qubits. The simulation fidelity for $ n = 3$ qubits is clearly degraded, exhibiting higher variance near the extrema $f^{(t)}(x) = \pm 1$ of the true signal, but nevertheless captures an approximately sinusoidal waveform with the correct frequency.  

To gain intuition on the expected fidelity of the quantum computational model as a function of the number of qubits, in Fig.~\ref{figSpecEig} we show the spectra of eigenvalues $s_j$ and representative corresponding eigenfunctions $u_j$ of the self-adjoint operator $S_n := \bm\Pi_n(Tf) $ from~\eqref{eqUtApprox} for $n=3$ and 7 qubits. Recall, in particular, that $S_n$ is an approximation of the multiplication operator by $f$, with its spectrum of eigenvalues $\sigma(S_n)$ providing a discretization of the (continuous) range of values of $f$, i.e., in  this case the interval $[-1,1]$. Moreover, $S_n$ is unitarily equivalent to the quantum computational observable $\hat S_n = \mathcal W_n S_n$, which is in turn approximately unitarily equivalent to the Fourier-transformed observable $\tilde S_n = \bm{\mathfrak F}_{n,1} \hat S_n$ that our circuit approximately measures. In Fig.~\ref{figSpecEig}, it is evident that as $n$ increases, $\sigma(S_n)$ samples the interval $ [ -1, 1 ]$ with increasingly high density, exhibiting a clustering of eigenvalues near the boundary points $\pm 1 $. This concentration of density is consistent with the distribution of $f(x) = \sin x$ induced by a fixed-frequency rotation on the circle. Meanwhile, as $n$ increases, the eigenfunctions exhibit increasingly high localization, with eigenfunction $u_j(x)$ concentrated on points $x \in S^1 $ such that $f(x)$ is close to the corresponding eigenvalue $ s_j $. This is seen in the right-hand column of the figure for representative eigenfunctions $ u_j$. Thus, intuitively, as the number of qubits increases, the PVM associated with $\tilde S_n$ (which we approximate by the quantum computational PVM $\mathcal E_n$) provides a representation of the classical observable $ f $ of increasingly high resolution.   

\subsection{\label{secTorus}Quasiperiodic dynamics on the 2-torus}
%---------------------------------------------------------------
\begin{figure*}[t]
    \centering
    \includegraphics[width=.95\linewidth]{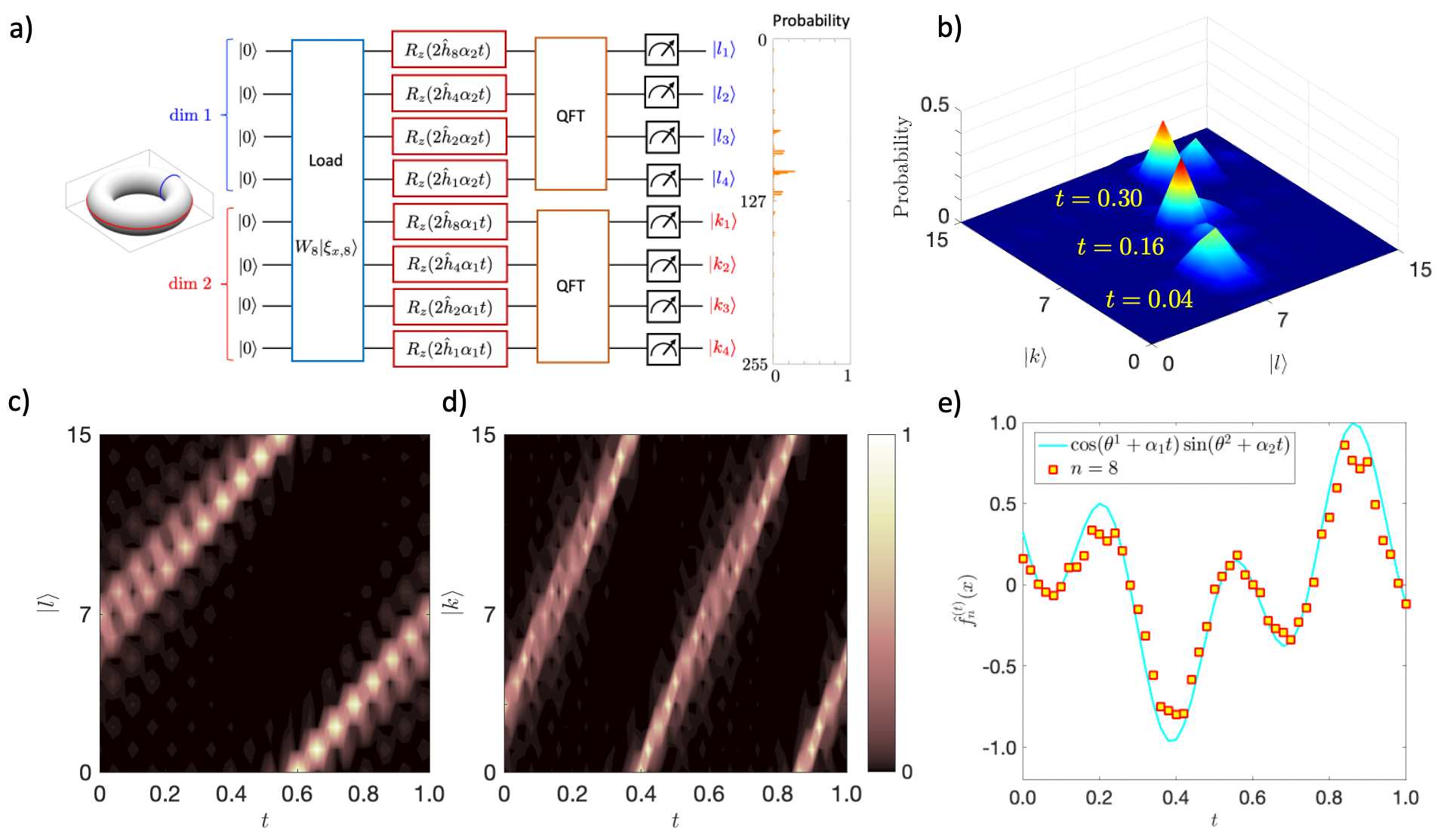}
    \caption{As in Fig.~\ref{dim1}, but for an 8-qubit approximation of a quasiperiodic rotation on the 2-torus with frequency parameters $\alpha_1 = 3 \sqrt 2 \pi $ and $\alpha_2 = 2 \pi$. (a) Quantum circuit for the quasiperiodic system, composed as two parallel copies of the circuit in Fig.~\ref{dim1}(a) for the one-dimensional case, with 4 qubits allocated to each dimension of the 2-torus. An empirical probability distribution obtained from $K=10^6$ shots is shown to the right of the circuit diagram, where the integers $b = 0, \ldots, 2^8-1 = 255 $ index the computational basis vectors $\ket b$ of the 256-dimensional Hilbert space $\mathbb B_n$ with $n=8$. The RKHA parameters are again  $p=\tau=1/4$. (b) Snapshots of the probability distribution at three representative evolution times, combined in a single surface plot. The horizontal axes labeled $\ket k$ and $\ket l $ correspond to the basis vector indices for each of the 4-qubit spaces associated with each torus dimension through the factorization $\mathbb B_8 = \mathbb B_4 \otimes \mathbb B_4$. Note that the indices $k$ and $l$ range from 0 to $2^4-1=15$. (c, d) Evolution of the marginal distributions obtained by measurement of the PVMs of each of the two 4-qubit spaces, i.e., one of the two torus dimensions only. The initial condition is $ x = (\theta^1,\theta^2) = (1.0,2.5) $, and measurements are performed at a fixed timestep $\Delta t=0.02$. The slopes of the probability contours in Panels~(c) and~(d) are proportional to the frequency parameters $\alpha_2$ and $\alpha_1$, respectively. Notice that the slopes in Panel~(c) are shallower than those in Panel~(d) since $\alpha_2 < \alpha_1$, and are equal to the corresponding slopes in Fig.~\ref{dim1}(c) since $\alpha_2$ is equal to the frequency parameter of the one-dimensional example. (e)~Reconstruction of the classical observable $f^{(t)}(x)=f^{(t)}(x_1,x_2)=\cos(\theta^2+\alpha_2 t)\sin(\theta^1+\alpha_1 t)$ from the ensemble means $\hat f^{(t)}_n(x)$ output from the quantum computer. The true evolution $f^{(t)}(x)$ is plotted as a cyan solid line.}
    \label{dim2}
\end{figure*}

%---------------------------------------------------------------

The two-dimensional case proceeds along similar lines as the one-dimensional example in Sec.~\ref{secCircle}, so we mainly focus on the points that are different from the one-dimensional example. The classical dynamical orbit on the 2-torus is now given by 
\begin{displaymath}
    x(t)=\Phi^t(x) = (\theta^1+\alpha_1 t, \theta^2+\alpha_2 t) \mod 2\pi,   
\end{displaymath}
where $\alpha_1$ and $\alpha_2$ are the frequency parameters and $x=(\theta^1, \theta^2)$ is the initial condition. We choose the (rationally independent) values $\alpha_1 = 3 \sqrt 2 \pi $ and $\alpha_2 = 2 \pi $, leading to an ergodic flow on $\mathbb T^2$. We again seek to approximate the evolution of a bandlimited classical observable $f$, in this case $f(x) = \sin(\theta^1) \cos(\theta^2)$. The evolution of this observable is given by
\begin{displaymath}
    f^{(t)}(x) = U^t f(x) = \sin(\theta^1+\alpha_1 t) \cos(\theta^2 + \alpha_2 t).
\end{displaymath}

To perform quantum simulation, we set the RKHA parameters $p = \tau = 1/4$ as in Sec.~\ref{secCircle}, and use a total of $n=8$ cubits, which corresponds to 4 qubits allocated to each torus dimension. The quantum computational Hilbert space, $\mathbb B_8$, is thus 256-dimensional, and admits the tensor product factorization 
\begin{equation}
    \label{eqB8}
    \mathbb B_8=\mathbb B_4\otimes \mathbb B_4. 
\end{equation}
For convenience in the notation, we will label the basis vectors for each of the $\mathbb B_4$ factors in~\eqref{eqB8} as $\ket{\bm k }$ and $\ket{\bm l }$, where $\bm k = (k_1,k_2,k_3,k_4) $ and $\bm l = (l_1, l_2, l_3, l_4) $ are 4-digit binary strings.  Note that the factorization in~\eqref{eqB8} is compatible with the tensor product structure of the infinite-dimensional RKHA $\rkha$ in~\eqref{eqRKHAProd}, in the sense that each $\mathbb B_4$ factor corresponds to the image space under a projection of the $\rkha^{(1)}$ spaces in~\eqref{eqRKHAProd}. See also Appendix~\ref{appWalsh}, and in particular~\eqref{eqHWalsh}. A similar tensor product structure applies for the quantum feature map, dynamical evolution, and QFT operators,
\begin{equation}
    \label{eqTorusMaps}
    \begin{aligned}
        \hat{\mathcal F}_n &= \hat{\mathcal F}_{n/2}^{(1)} \otimes \hat{\mathcal F}_{n/2}^{(1)}, \\
        \hat U^t_n &= (U^t_{n/2})^{(1)} \otimes (U^t_{n/2})^{(1)}, \\ 
        \bm{\mathfrak F}_{n,2} &= \bm{\mathfrak F}_{n/2,1} \otimes \bm{\mathfrak F}_{n/2,1},
    \end{aligned}
\end{equation}
so we can form the entire circuit by composing two 4-qubit circuits from the one-dimensional case; see Fig.~\ref{dim2}(a) for an illustration. In~\eqref{eqTorusMaps}, $(1)$-superscripts and $1$-subscripts denote maps inherited from the one-dimensional case.  

As in the one-dimensional example of Sec.~\ref{secCircle}, all quantum states occurring in our scheme are pure, so we implement the circuit in Fig.~\ref{dim2}(a) at the level of the vectors $\xi_{x,n} $ (normalized RKHS feature vectors),  $\ket{\hat \xi_{x,n}} = W_n \xi_{x,n} $ (initial state vectors), $\ket{\hat \xi_{x,n}^{(t)}} = \hat U^{t*}_n \ket{\hat \xi_{x,n}}$ (Koopman-evolved state vectors), and $\ket{\tilde \xi^{(t)}_{x,n}} = \mathfrak F_{n,2} \ket{\hat \xi^{(t)}_{x,n}}$ (state vectors after application of the QFT). Note that the normalized feature vector associated with classical state $x \in \mathbb T^2$ takes the form 
\begin{displaymath}
    \xi_{x,n}=\sum_{j\in J_{n,2}} \frac{\psi^*_j(x)}{\sqrt{\kappa_n}}\psi_j,
\end{displaymath}
with $n=8$ and
\begin{align*}
    \psi^*_j(x) = \exp\left[-\frac{\tau}{2}(|j_1|^p+|j_2|^p)\right] \exp[-i(j_1x_1+j_2x_2)]. 
\end{align*}
See again Table~\ref{tableWalsh} for an example of the ordering of the multi-index $j$ and its mapping to the computational basis in the case $n=4$ (the table would have 256 rows in the current example). We also note that our $n=8$ example has $2\times 4$ nonzero Walsh-Fourier expansion coefficients: $\tilde h_1=-9/2$, $\tilde h_2 =-2$, $\tilde h_4=-1$, and $\tilde h_8=-1/2$ for each torus dimension.

Figure~\ref{dim2}(b) displays snapshots of the empirical joint probability distribution of the $(\bm k, \bm l)$ indices at representative evolution times $t$, obtained from ensembles of $K=10^6$ measurements of the quantum computational PVM $\mathcal E_n$ on the state represented by $\ket{\tilde \xi^{(t)}_{x,n}} $ for the initial condition $x=(1.0,2.5)$. The locality of the distributions is indicative of the fact that the quantum computing model successfully tracks the orbit of the underlying classical dynamical system. Figs.~\ref{dim2}(d) and~\ref{dim2}(d) show marginals of these distributions over the $\bm k $ and $\bm l $ index spaces as a function time $t$, where periodic evolution at the generating frequencies $\alpha_1$ and $\alpha_2$, respectively, is apparent. 

In Fig.~\ref{dim2}(e) we compare the approximate evolution $\hat f^{(t)}_n(x)$ of the observable $f$ computed from the same ensembles of quantum measurements against the true evolution $f^{(t)}(x)$. Despite the modest number of qubits allocated to each torus dimension, $\hat f^{(t)}_n(x)$ reproduces the quasiperiodic behavior of $f^{(t)}(x)$ to an adequate degree of accuracy, with more pronounced errors occurring near the extrema of the true signal. As in the one-dimensional example of Fig.~\ref{dim1}(d), we expect such discrepancies to rapidly diminish as the number of qubits increases. Similarly, from this example it becomes clear  how one can generalize the dynamics to a torus of dimension $ d > 2$.

\section{\label{secIBM}Experiments on the IBM Quantum System One}

%---------------------------------------------------------------
\begin{figure*}[t]
    \centering
    \includegraphics[width=\linewidth]{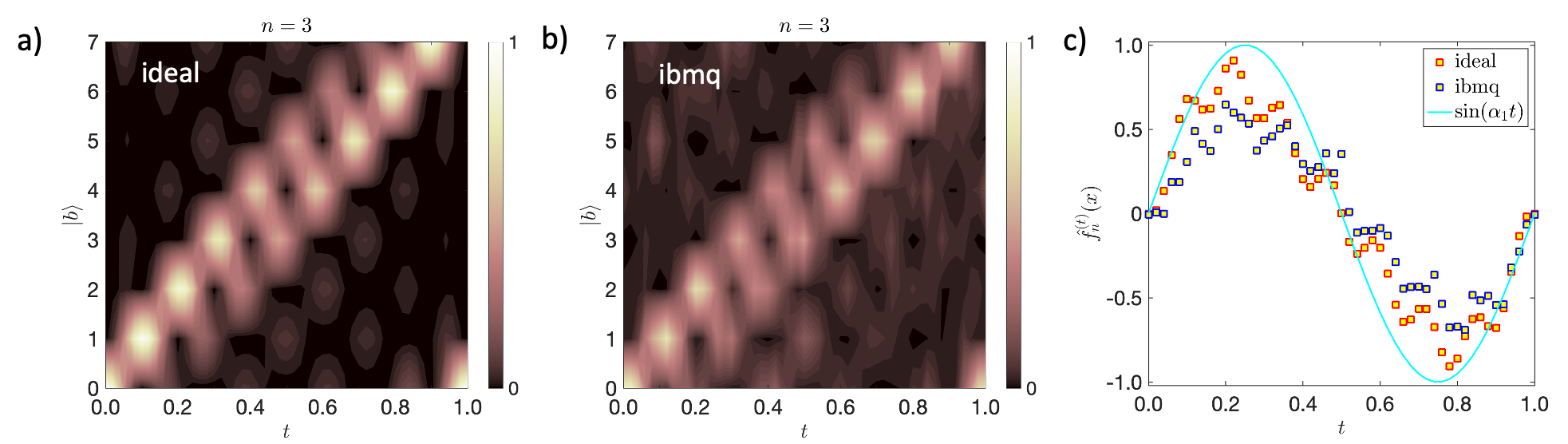}
    \caption{Comparison of 3-qubit approximations of a circle rotation with frequency $\alpha_1 = 2\pi$ from simulated circuit experiments in the ideal Qiskit Aer environment (ideal) and actual quantum computing experiments on the IBM Quantum System One (ibmq). (a)~Temporal evolution of the empirical probability similar to Fig.~\ref{dim1}(d) in the ideal Qiskit circuit simulation, using amplitude encoding with \texttt{QuantumCircuit.initialize} for the state preparation and the RKHA parameters $p=\tau=1/4$. (b)~Temporal evolution of the empirical probability distributions for $n=3$ on the quantum computer starting with a uniform superposition state $|\Omega\rangle$ at $t=0$. (d)~Reconstruction of the classical observable $f^{(t)}(x)=\sin(x(t))=\sin(\theta^1+\alpha_1 t)$ from the ensemble means, $\hat f^{(t)}_n(x)$. The analytical result $f^{(t)}(x)$ is plotted as a cyan solid line. In all panels, the initial condition is $x=\theta^1 = 0$. Measurements are performed at a fixed timestep $\Delta t = 0.02$. The number of shots is $K=2^{18}=\text{262,144}$ in both cases.}
    \label{ibmq}
\end{figure*}

%---------------------------------------------------------------
The circle rotation algorithm for $n=3$ qubits was also implemented on the IBM Quantum System One to demonstrate the readiness of QECD on a real NISQ device. This system has a quantum volume (an empirical metric that quantifies the capability and error rates of a quantum device) of 32. The corresponding program was again written in Qiskit (see Sec.~\ref{secExamples}), and then transpiled (translated) into a sequence of appropriate elementary gate operations acting on the physical superconducting qubits via microwave channels at the hardware level. No error correction was used in our simulation.  As mentioned in Sec.~\ref{secPreparation}, the encoding of $2^n$ (complex) amplitudes that represent the feature vector $\ket{\hat{\xi}_{x,n}}$ associated with classical state $x\in X$ in an $n$-qubit quantum register can lead to an exponential growth of gates. To give a concrete example for $n=3$: amplitude encoding using \texttt{QuantumCircuit.initialize} with no circuit optimization is transpiled into a sequence of 84 elementary quantum gates. This conversion results to 52 elementary gates for a higher transpiler optimization level of 2. 

To circumvent this expensive amplitude encoding of classical data, it was shown in Sec.~\ref{secPreparation} that the initial state vector $\ket{\hat{\xi}_{x,n}}$ can also be obtained to any degree of accuracy with a circuit of size $O(n)$ and depth $O(1)$. In the particular case $x = e$ (i.e., the point with canonical angle coordinates $\theta^1, \ldots,\theta^d = 0 $), the encoding reduces to a uniform superposition state for $n$-qubits, $|\Omega\rangle$, which is obtained via $n$ Hadamard gates $\mathsf H$ applied to the standard basis quantum state $|0\rangle^{\otimes n}$ (see~\eqref{Hadamard}). This step reduces the number of gates, and thus the circuit depth, significantly to 33 and 30 for the transpiler optimization levels 0 and 2, respectively. This depth is close to the quantum volume of the   computer. 

Figure \ref{ibmq} directly compares the results of an ideal Qiskit Aer simulator for $n=3$ and $\tau=p=1/4$ with an experiment on the IBM Quantum System One for the observable $f(x) = \sin x$ and an initial uniform superposition state $|\Omega\rangle$ (approximating $\ket{\hat\xi_{e,n}}$). Despite the noise caused by decoherence, the evolution of probability densities (Fig.~\ref{ibmq}(b)) and expectation values (Fig.~\ref{ibmq}(c)) obtained from the NISQ device remain consistent with the Qiskit simulation (Fig.~\ref{ibmq}(a,c)). The number of shots, which is limited to 8192 on the Quantum System One, was enhanced to $2^{18}$ by aggregating results from multiple jobs.     

Unfortunately, increasing the number of qubits beyond $n=3$ led to noticeable degradation of the results on the quantum computer relative to the Qiskit simulations, despite our best efforts to manage noise and decoherence with the tools available to us. Still, to our knowledge, the $n=3$ results reported in this section constitute the first successful simulation of an observable of a classical dynamical system on a manifold by an actual NISQ device. We expect that as the coherence characteristics, error mitigation and/or circuit optimization schemes for quantum computation improve, the QECD framework presented in this paper will successfully scale to higher qubit numbers.

%========================================================================================
\section{\label{secConclusions}Summary and outlook}

We have developed a framework for approximating the evolution of observables of a classical dynamical system by a finite-dimensional quantum system implementable on an actual quantum computer. The procedure, which we refer to as quantum embedding of classical dynamics (QECD), takes the classical system as an input, and passes through intermediate classical statistical, infinite-dimensional quantum mechanical, and finite-dimensional quantum mechanical (matrix-mechanical), representations, ultimately arriving at an $n$-qubit quantum computational representation of the system. We have thus addressed the full pipeline starting from the classical dynamical system all the way to its experimental verification on a real quantum computer, the IBM Quantum System One.

For the class of dynamical systems under study (i.e., measure-preserving, ergodic dynamical systems with pure point spectra), QECD provides an exponential quantum advantage over classical computation, in the sense of being able to simulate a $2^n$-dimensional Hilbert space of classical observables using circuits of size $O(n^2)$ and depth $O(n)$. In addition, the quantum state encoding the initial classical quantum state is efficiently prepared, and predictions from the quantum computational system are extracted through projective measurement in the standard computational basis without requiring postprocessing techniques such as quantum state tomography. 

One of the mathematical underpinnings of our approach is the theory of reproducing kernel Hilbert spaces (RKHSs). RKHS theory is widely used in kernel methods for machine learning, but was employed here to construct quantum mechanical analogs of feature maps that behave consistently under classical function evaluation and quantum mechanical expectation. A further foundational ingredient is the operator-theoretic description of dynamical systems, which utilizes linear Koopman operators to characterize the action of a (nonlinear) dynamical system on observables. 

We described how QECD proceeds along two composite mappings, one taking state variables $x \mapsto \hat\rho_{x,n}$ to density operators $\hat\rho_{x,n}$ on an $n$-qubit Hilbert space, $\mathbb B_n$, and another one taking classical observables $f \mapsto \hat S_n$ to self-adjoint operators $\hat S_n$ on $\mathbb B_n$. A key aspect of the resulting quantum system is a tensor product factorization of its Hamiltonian in terms of Walsh operators, yielding quantum circuits of low size and depth. In particular, it was shown that for an ergodic dynamical system with finitely-generated pure point spectrum, this results in a circuit of size $O(n)$ and no cross-channel communication, implementing unitary Koopman evolution. The QECD framework also includes a state preparation stage of size $O(n)$, as well as a quantum Fourier transform (QFT) stage of size $O(n^2)$ to enable information retrieval through measurement in the computational basis.

The scheme exhibits three types of approximation error, all of which can be controlled, as we have shown, in appropriate asymptotic limits:
\begin{enumerate}
    \item Finite-dimensional approximation errors due to projection of the infinite-dimensional quantum system on the RKHS $\mathcal H$ to the finite-dimensional quantum computational system on $\mathbb B_n$. These errors vanish as $n\to\infty$, and the convergence is unconditional on the defining parameters of $\mathcal H$ if idealized state preparation and measurement is employed (see Sec.~\ref{secMatrixMechanical}).
    \item Bias errors due to preparation of an approximate initial quantum state and measurement of an approximate observable using efficient circuits. These errors vanish in a joint limit of decreasing RKHS parameter $\tau $ and increasing $n$ (see Secs.~\ref{secMeas}, \ref{secPreparation}, and Appendix~\ref{appQFT}).
    \item Monte Carlo errors associated with approximation of quantum mechanical expectations with a finite number of measurement shots (see Sec.~\ref{secApproximateMeas}). These errors vanish as the number of shots, $K$, increases at fixed $n$ and $\tau$.    
\end{enumerate} 

We illustrated our approach with periodic and quasiperiodic dynamical systems on the circle and 2-torus, respectively, where many aspects of the quantum embedding of classical dynamics can be directly validated against closed-form solutions. Our numerical experiments were based on simulated quantum circuits of up to $n=8$ qubits, implemented using the Qiskit framework. In addition we demonstrated the ability of our framework to deal with a classical dynamical system on a real noisy quantum computer. The results demonstrated high-fidelity simulation of the evolution of classical observables through ensemble averages of independent quantum measurements. Our approach is straightforwardly generalizable to quasiperiodic dynamics of arbitrarily large intrinsic dimension through parallel composition of quantum circuits. 

The work presented in this paper should be considered a first step, particularly given its focus on systems with pure point spectra. Applications of the procedure to mixing (chaotic) dynamical systems will invariably have to deal with the continuous spectrum of the Koopman operator, potentially generating quantum circuits of higher connectivity than for quasiperiodic dynamics. Studies in this direction are currently underway using RKHS-based spectral discretization approaches for Koopman operators \cite{DasEtAl21} (see Sec.~II~C~2 in the SM), which are able to consistently approximate, in a spectral sense, measure-preserving, ergodic dynamical flows of arbitrary spectral character (pure point spectrum, mixed spectrum, and continuous spectrum) by unitary evolution groups with pure point spectra. A possible route to generalize QECD to this class of systems is to employ the scheme of Ref.~\cite{DasEtAl21} to first approximate the Koopman group on $L^2(\mu)$ by a unitary evolution group on an RKHS with a discrete spectrum, and then apply the quantum computational techniques developed in this paper to simulate the discrete-spectrum system.

Another avenue of future research is to develop data-driven formulations of the present quantum embedding framework, using kernel methods to build orthonormal bases from dynamical trajectory data, and employ these bases to represent quantum mechanical states and observables \cite{Giannakis19} (see Sec.~II in the SM \cite{Suppl}). This line of research should lead to a systematic development of quantum machine learning algorithms that can describe classical dynamical systems on NISQ devices. This comprises not only classification and regression tasks \cite{Schuld21}, but also the development of data-driven quantum algorithms for modeling nonlinear dynamics in high-dimensional phase spaces. A longer-term goal would be to explore applications of quantum mechanical methodologies to perform simulation and forecasting of real-world systems such as climate dynamics \cite{SlawinskaEtAl19} and turbulent fluid flows \cite{GiannakisEtAl18}. 

%========================================================================================
\acknowledgements
We wish to thank Sachin Bharadwaj for helpful discussions. 
We acknowledge the use of IBM Quantum services for this work. The views expressed are those of the authors, and do not reflect the official policy or position of IBM or the IBM Quantum team. In this paper we used {\em ibmq\_ehningen}, which is one of the IBM Quantum Falcon Processors. We thank the Fraunhofer Gesellschaft (Germany) for support. D.\ Giannakis acknowledges support from the US National Science Foundation under grants 1842538 and DMS-1854383, the US Office of Naval Research under MURI grant N00014-19-1-242, and the US Department of Defense under Vannevar Bush Faculty Fellowship grant N00014-21-1-2946. The work of A.\ Ourmazd  was supported by the US Department of Energy, Office of Science, Basic Energy Sciences under award DE-SC0002164 (underlying dynamical techniques), and by the US National Science Foundation under awards STC 1231306 (underlying data analytical techniques) and DBI-2029533 (underlying analytical models). P.\ Pfeffer is supported by the Deutsche Forschungsgemeinschaft with project SCHU 1410/30-1 and by the project ``DeepTurb -- Deep Learning in and of Turbulence'' of the Carl Zeiss Foundation (Germany). J.\ Slawinska acknowledges support from the core funding of the Helsinki Institute for Information Technology (HIIT) and the Institute for Basic Sciences (IBS), Republic of Korea, under IBS-R028-D1.

%========================================================================================
\appendix

\section{\label{appQuantumRep}Quantum mechanical representation of classical observables} 

In this appendix, we state various properties and results on the representation of classical observables by quantum mechanical operators employed in the main text.

\subsection{\label{BanachStar}Banach $^*$-algebra structure of $\rkha$}

The fact that the RKHA $\rkha$ from Sec.~\ref{secRKHA} is an abelian, unital, Banach $^*$-algebra under pointwise multiplication of functions means that it has the following defining properties: 
\begin{enumerate}
    \item $\rkha $ is closed under pointwise multiplication of functions, i.e., the function $ h : X \to \mathbb C $ with $ h( x ) = f(x)g(x) $ lies in $\rkha $ whenever $ f $ and $ g $ lie in $ \rkha$. Thus, $\rkha$ is an algebra, and is clearly abelian since $fg=gf$.
    \item $\rkha$ is equipped with an antilinear involution operation $^* : \rkha \to \rkha $ given by complex conjugation of functions, i.e. $ ( f^* )(x) = f(x)^* $. Thus, $\rkha $ is also a $^*$-algebra. 
    \item There exists a constant $ C > 0 $ such that for every $ f, g \in \rkha$ the relationships
        \begin{equation}
            \label{eqBanachAlg}
            \lVert f g \rVert_{\rkha} \leq C \lVert f \rVert_{\rkha} \lVert g \rVert_{\rkha}, \quad \lVert f^* \rVert_{\rkha} = \lVert f \rVert_{\rkha},
        \end{equation}
        hold. Thus, $ \rkha $ is a Banach $^*$-algebra.
    \item The function $1_X : X \to \mathbb C$ equal everywhere to 1 lies in $ \rkha$ and satisfies $ 1_X f = f $ for all $ f\in \rkha$. Thus, finally $ \rkha$ is also unital.
\end{enumerate}

More generally, the topic of Banach function algebras on locally compact abelian groups (with respect to either pointwise multiplication or convolution), has a long history of study; e.g., \cite{Wermer54,Brandenburg75,Feichtinger79,Grochenig07,KuznetsovaMolitorBraun12}.

\subsection{\label{appInjectivity}Injectivity of the map $\tilde T$}

We verify the assertion made in Section~\ref{secQuantumRep} that the map $ \tilde T : \rkha \to B(\rkha) $ is injective on $\rkha_\text{sa}$. For that, it is enough to show that if $ \tilde T f = 0 $ for $ f \in \rkha_\text{sa}$, then $ f= 0 $. By definition of $ \tilde T $, $\tilde Tf = 0 $ implies that $ \pi f = - ( \pi f )^*$, or, equivalently 
\begin{equation}
    \label{eqSkewF}
    \langle \psi_i, f \psi_j \rangle_{\rkha} = - \langle  f \psi_i, \psi_j \rangle_{\rkha}, \quad \forall i,j \in \mathbb Z^d.
\end{equation}
Expanding $ f = \sum_{l \in \mathbb Z^d} \tilde f_l \psi_l $, and setting $ i = 0 $ in~\eqref{eqSkewF}, we get 
\begin{displaymath}
    \tilde f_j^* = - c_{j,-j} \tilde f_{-j}. 
\end{displaymath}
However, because $f$ is real, we have $ \tilde f_j^* = \tilde f_{-j}$, and since $ c_{j,-j}$ is nonzero we conclude that $ \tilde f_j = 0 $, and thus $ f=0$. 

\subsection{\label{appConsistency}Consistency of representations based on the reproducing kernel Hilbert space $\mathcal H$}

Recall the construction of the RKHS $\mathcal H$ in Sec.~\ref{secRKHS}. Even though  $\mathcal H$ is a strict subspace of the RKHA $\rkha$, the quantum feature map $\mathcal F : X \to Q(\mathcal H) $ from~\eqref{eqQ} allows us to consistently recover all predictions made for classical observables obtained via the feature map $\tilde{\mathcal F} : X \to Q(\rkha) $ of $\rkha$ in~\eqref{eqQRKHA}, as we now describe. 

First, observe that by definition of $ \varrho_x = \tilde{\mathcal F}(x)$ and $ \rho_x = \mathcal F(x)$, we have
\begin{equation}
    \label{eqRhoVarrho}
    \rho_x = \frac{\tilde\kappa}{\kappa} \bm \Pi \rho_x,
\end{equation}
where $\bm \Pi $ is the projector onto $B(\mathcal H)$, defined in~\eqref{eqPiProj}. As a result, if $ A \in B(\rkha)$ is a quantum mechanical observable whose range is included in $ \mathcal H$ (so that $A $ is well-defined as an operator on $\mathcal H$), and whose nullspace includes the orthogonal complement $ \mathcal H^\perp$ in $\rkha$, we have
\begin{equation}
    \label{eqRhoEquiv}
    \langle A \rangle_{\varrho_x} = \frac{\kappa }{\tilde\kappa}\langle A \rangle_{\rho_{x}}.
\end{equation}
Indeed, since every observable $A$ in this class satisfies $\bm \Pi A = A$, using~\eqref{eqRhoVarrho} and the cyclic property of the trace, we get 
\begin{align*}
    \langle A \rangle_{\varrho_x} & = \tr(\varrho_x A) = \tr(\varrho_x (\bm \Pi A)) \\
    &= \tr(\varrho_x \Pi A \Pi ) = \tr( \Pi \varrho_x \Pi A ) \\
    &= \frac{\kappa}{\tilde \kappa} \tr(\rho_x A) = \frac{\kappa}{\tilde \kappa} \langle A \rangle_{\rho_x},
\end{align*}
which verifies~\eqref{eqRhoEquiv}.
Thus, for all observables $A \in B(\rkha) $ satisfying
\begin{equation}
    \label{eqRanA}
    \ran A \subseteq \mathcal H, \quad \ker A \supseteq \mathcal H^\perp, 
\end{equation}
expectation values with respect to $ \varrho_x$ can be recovered from expectation values with respect to $ \varrho_x$ up to a constant scaling factor. For our purposes, this means that the quantum mechanical observables $ \bm\Pi(\pi f)  $ and $ \bm\Pi(\tilde T f) $ obtained through the projections  $ \bm \Pi \pi $ and $  \bm \Pi \tilde T $ of $ \pi $ and $ \tilde T $ from~\eqref{eqPi} and~\eqref{eqT}, respectively, satisfy~\eqref{eqRhoEquiv}.

As noted in Sec.~\ref{secRKHSRep}, in order to reach consistency between classical function evaluation and quantum mechanical expectation, analogously to~\eqref{expec1}, we introduce the modified representation maps $ \varpi : \rkha \to B(\mathcal H)$ and $ T : \rkha \to B(\mathcal H)$ in~\eqref{eqVarpi} to account for scaling errors. We reproduce the definitions here for convenience:
\begin{displaymath}
    \varpi = \bm \Pi \pi  L^{-1}, \quad T = \bm \Pi \tilde T  L^{-1}.
\end{displaymath}
We define $ L : \rkha \to \rkha $ as the self-adjoint, diagonal operator satisfying the eigenvalue equation
\begin{equation}
    \label{eqLOp}
    L \psi_l = \frac{\eta_l}{\kappa} \psi_l \quad\mbox{with}\quad \eta_l = \sum_{j \in J'_l}e^{-\tau \lvert j \rvert^p},
\end{equation}
where $J'_l$ is the index set defined as 
\begin{displaymath}
    J'_l = \{ j \in J : j + l \in J \}.  
\end{displaymath}
Note that by construction of $J'_l$, the numbers $\eta_l$ are strictly positive, and have the maximum value $\eta_0 = \kappa$. Moreover, the $\eta_l$ attain their smallest value, $e^{-\tau}$, when $\lvert l \rvert = 1$, i.e., the multi-index $ l = (l_1,\ldots, l_d) \in \mathbb Z^d$ has exactly one entry $l_i$ equal to $\pm 1$ and all other entries equal to 0. As a result, $L$ is an invertible operator with bounded inverse, satisfying
\begin{displaymath}
    L^{-1} \psi_l = \frac{\kappa}{\eta_l} \psi_l.
\end{displaymath}
Since $\kappa/\eta_l \geq 1$, we deduce that $L^{-1}$ acts by inflating the expansion coefficients of elements of $\rkha$ in the $ \{ \psi_j \} $ basis. 

We then have: 

\begin{prop}
    \label{propExpec}    
    The following classical--quantum consistency relation holds for every $f \in \rkha $ and  $x \in X$:
    \begin{displaymath}
        f(x) = \langle \varpi f\rangle_{\rho_{x}}. 
    \end{displaymath}
    Moreover, if $f$ is a real-valued observable in $ \rkha_{\text{sa}}$, we have 
    \begin{displaymath}
        f(x) = \langle T f\rangle_{\rho_x}
    \end{displaymath}
\end{prop}

\begin{proof}    
Suppose that $g = \psi_l $ for some $ l \in \mathbb Z^d$, and let $ A = \bm \Pi A_g$, where $A_g = \pi \psi_l \in B(\rkha)$ is the multiplication operator by $\psi_l$. Then, $A$ satisfies~\eqref{eqRanA}, and using~\eqref{eqRhoEquiv}, we get
\begin{align}
    \nonumber \langle A \rangle_{\rho_x} &= \frac{\tilde \kappa }{\kappa} \langle A \rangle_{\varrho_x} = \frac{\tilde \kappa }{\kappa}\tr( \varrho_x \Pi A_g \Pi ) \\
    \nonumber &= \frac{\tilde \kappa }{\kappa}\sum_{j \in \mathbb Z^d} \langle \psi_j, \varrho_x \Pi A_g \Pi \psi_j \rangle_{\rkha} \\
    \nonumber &= \frac{\tilde \kappa }{\kappa}\sum_{j \in J} \langle \psi_j, \varrho_x \Pi A_g \psi_j \rangle_{\rkha} \\
    \nonumber &= \frac{\tilde \kappa }{\kappa}\sum_{j \in J} \langle \psi_j, \varrho_x \Pi (\psi_l \psi_j) \rangle_{\rkha}  \\
    \nonumber &= \frac{1}{\kappa} \sum_{j\in J} \langle \tilde k_x, \Pi(\psi_l\psi_j) \rangle_{\rkha} \langle \psi_j, \tilde k_x \rangle_{\rkha} \\
    \nonumber &= \frac{1}{\kappa} \sum_{j\in J'_l} \langle k_x, \psi_l \psi_j \rangle_{\rkha} \langle k_x, \psi_j \rangle_{\rkha}^*\\
    \nonumber &= \frac{1}{\kappa} \sum_{j \in J'_l} \psi_j^*(x) \psi_j(x) \psi_l(x) \\
    \nonumber &= \frac{1}{\kappa} \sum_{j\in J'_l} e^{-\tau \lvert j \rvert^p} \lvert \phi_j(x)\rvert^2 \psi_l(x) \\
    \nonumber &= \frac{1}{\kappa} \sum_{j\in J'_l} e^{-\tau \lvert j \rvert^p}  \psi_l(x) \\
    \label{eqRhoEquiv2}&= \frac{\eta_l}{\kappa} \psi_l(x) = L \psi_l(x).
\end{align}
Meanwhile, an application of~\eqref{expec1} for $ f = L \psi_l $ gives 
\begin{equation}
    \label{eqRhoEquiv3} L \psi_l(x) = \langle \pi(L \psi_l ) \rangle_{\varrho_x},
\end{equation}
and combining~\eqref{eqRhoEquiv2} and~\eqref{eqRhoEquiv3} we arrive at
\begin{equation}
    \label{eqRhoEquiv4}
    \langle \bm \Pi (\pi g) \rangle_{\rho_x} = \langle \pi( L g ) \rangle_{\varrho_x},
\end{equation}
where $ g= \psi_l $. Since the basis vector $\psi_l $ was arbitrary, it follows by linearity that~\eqref{eqRhoEquiv4} holds for every $ g \in \rkha$. Setting, in particular, $g = L^{-1} f $ yields  
\begin{displaymath}
    \langle \bm \Pi(\pi(L^{-1} f)) \rangle_{\rho_x} = \langle \pi f \rangle_{\varrho_x} \iff \langle \varpi f \rangle_{\rho_x} = f(x),     
\end{displaymath}
which confirms the first claim of the proposition. The second claim, $f(x) = \langle T f \rangle_{\rho_x}$, follows similarly under the additional assumption that $f^* = f$.
\end{proof}

\subsection{\label{appRKHSDyn}Dynamics on the reproducing kernel Hilbert space $\mathcal H$}

By construction, the RKHS $\mathcal H$ is a Koopman-invariant subspace of $\rkha$, i.e., $ U^t \mathcal H = \mathcal H$ for all $t \in \mathbb R$. As a result, we can define a generator $V : D(V) \to \mathcal H$ with $D(V) \subset \mathcal H$, a corresponding Koopman operator $U^t : \mathcal H \to \mathcal H$, and corresponding evolution maps on observables, $\mathcal U^t : B(\mathcal H) \to B(\mathcal H)$, and states, $\Psi^t : Q(\mathcal H) \to Q(\mathcal H)$ analogously to the corresponding operators associated with $\rkha$. These operators satisfy the compatibility relations (cf.~\eqref{eqUPi} and~\eqref{eqPsiF}) 
\begin{displaymath}
    \mathcal U^t(\varpi f) = \varpi(U^t f), \quad \Psi^t(\mathcal F(x)) = \mathcal F(\Phi^t(x)) 
\end{displaymath}
for every $ f \in \rkha$, $ x \in X$, and $ t \in \mathbb R$, where $\varpi : \rkha \to B(\mathcal H) $ is the map on observables in~\eqref{eqVarpi} and $\mathcal F: X \to Q(\mathcal H)$ the quantum feature map in~\eqref{eqQ}. In addition, using the consistency relations in Proposition~\ref{propExpec} and~\eqref{eqUtSF}, we get 
\begin{equation}
    \label{eqConsistency2}
    \begin{aligned}
        U^t f(x) &= \langle \mathcal U^t(\varpi f) \rangle_{\rho_x} = \langle \varpi f \rangle_{\Psi^t(\rho_x)}, \\
        U^t f(x) &= \langle \mathcal U^t(T f) \rangle_{\rho_x} = \langle T f \rangle_{\Psi^t(\rho_x)},
    \end{aligned}
\end{equation}
where $T: \mathcal H \to  B(\mathcal H)$ was defined in~\eqref{eqVarpi}, and the equalities in the second line hold for real-valued functions in $\mathcal H$. It follows from~\eqref{eqConsistency2} that we can consistently represent the evolution of classical observables in $\rkha $ (which is a dense subspace of $C(X)$) by quantum mechanical evolution of observables in $B(\mathcal H)$, even though $ \mathcal H$ is a non-dense subspace of $\rkha$.

\section{\label{appWalsh}Walsh operator representation of the Koopman generator}

Here, we lay out the calculation of the discrete Walsh transform of the spectral function $h \in L^2_n([0,1])$ of the Hamiltonian $H_n$ induced by the Koopman generator of a quasiperiodic dynamical system, defined in~\eqref{eqHFunc}. In particular, we show that $h$  is expressible as a linear combination of Rademacher functions $R_l=w_{2^l}$ (without contributions from more general Walsh functions), leading to the factorization of $H_n$ in~\eqref{eqHDecompQuasiperiodic}. 

First, by \eqref{eqOmega} and~\eqref{eqHFunc}, for any $ m \in \{ 0, \ldots, 2^n -1 \} $ we have    
\begin{equation}
    \label{eqFOmega}
    h\left(\frac{m}{2^n}\right) = \omega_j=\alpha_1 j_1 + \alpha_2 j_2 + \ldots + a_dj_d,
\end{equation}
where $j_1, \ldots, j_d $ are integers in the set $J_1$, defined uniquely by the property that the concatenated binary strings $\eta(j_1), \ldots, \eta(j_d)$ give the dyadic decomposition of $ m / 2^n$,
\begin{equation}
    \label{eqRad2a}
    \gamma\left( \frac{m}{2^n} \right) = (\eta(j_1),\ldots,\eta(j_d)).
\end{equation}
We can express the left-hand side of~\eqref{eqRad2a} in terms of Rademacher functions using~\eqref{eqRad2}, viz. 
\begin{align}
    \nonumber \gamma\left( \frac{m}{2^n} \right) &=  \left[ \gamma_1\left( \frac{m}{2^{n}}\right), \ldots,  \gamma_n\left(\frac{m}{2^{n}}\right)\right] \\
    \label{eqRad2b}&= \frac{1}{2} - \frac{1}{2} \left[R_0\left(\frac{m}{2^n}\right), \ldots, R_{n-1}\left(\frac{m}{2^n}\right)\right].
\end{align}
Meanwhile, setting $ m_i = o( j_i ) $ and using again~\eqref{eqRad2}, the right-hand side of~\eqref{eqRad2a} becomes
\begin{multline}
    \label{eqRad2c}
    [\eta(j_1), \eta(j_2), \ldots, \eta(j_d)] \\
    \begin{aligned}
        &= \left[ \gamma\left( \frac{m_1}{2^{n/d}}\right), \ldots,  \gamma\left(\frac{m_d}{2^{n/d}}\right)\right] \\
        &= \frac{1}{2} - \frac{1}{2}\left[     R_0\left(\frac{m_1}{2^{n/d}} \right), \ldots, R_{n/d-1}\left(\frac{m_1}{2^{n/d}}\right) \right. \\
        &\qquad\quad\;\;\;\;\;\; R_0\left(\frac{m_2}{2^{n/d}} \right), \ldots, R_{n/d-1}\left(\frac{m_2}{2^{n/d}}\right), \\
        &\qquad\qquad\;\;\;\;\;\;\quad \quad \quad\;\;\ldots\ldots\\
        &\qquad\quad\;\;\;\;\;\;  \left. R_0\left(\frac{m_d}{2^{n/d}} \right), \ldots, R_{n/d-1}\left(\frac{m_d}{2^{n/d}}\right) \right]. 
    \end{aligned}
\end{multline}
Substituting for $\gamma(m/2^n)$ and $(\eta(j_1),\ldots,\eta(j_d))$ in~\eqref{eqRad2a} using~\eqref{eqRad2b} and~\eqref{eqRad2c}, respectively, we deduce that for each $ i \in \{ 1, \ldots, d \} $ and $ l~\in~\{ 0, \ldots, n - 1 \} $ 
\begin{equation}
    \label{eqRad3}
    R_l\left(\frac{m}{2^n}\right) = R_{l-(i-1)n/d}\left(\frac{m_i}{2^{n/d}}\right), 
\end{equation}
for all $ m \in \{ 0, \ldots, 2^n - 1 \} $. 

Observe now that for $ j_i \in J_1$, 
\begin{align*}
   j_i &= 
   \begin{cases}
       m_i - 2^{n/d - 1} &:\quad 0 \leq m_i \le 2^{n/d - 1} -1 \\
       m_i - 2^{n/d - 1} + 1 &:\quad 2^{n/d -1} \leq m_i \leq 2^{n/d} - 1
   \end{cases}\\
   & = \;\;\;\sum_{l=0}^{n/d-1} \dfrac{1 - R_l\left(\dfrac{m_i}{2^{n/d}}\right)}{2^{l+2-n/d}}
       + \dfrac{ 1 - R_0\left(\dfrac{m_i}{2^{n/d }}\right) }{ 2 } - 2^{n/d-1} \\
   &= - \sum_{l=0}^{n/d-1} \dfrac{R_l\left(\dfrac{m_i}{2^{n/d}}\right)}{2^{l+2-n/d}}
       - \dfrac{R_0\left(\dfrac{m_i}{2^{n/d}}\right) }{2} \\
   &= \;\;\;-\sum_{l=0}^{n/d-1} \dfrac{R_{l+(i-1)n/d}\left(\dfrac{m}{2^n}\right)}{2^{l +2-n/d}}
     -  \dfrac{R_{(i-1)n/d}\left(\dfrac{m}{2^n}\right)}{2},          
\end{align*}
where we used~\eqref{eqRad3} to obtain the last line. Substituting the above in~\eqref{eqFOmega}, we obtain
\begin{align*}
    h\left(\frac{m}{2^n}\right) &= - \sum_{i=1}^d \alpha_i \left[\sum_{l=0}^{n/d-1} R_{l+(i-1)n/d}\left(\frac{m}{2^n}\right) 2^{-l -2+n/d}\right] \\
    &\quad - \sum_{i=1}^d \frac{\alpha_i}{2} R_{(i-1)n/d}\left(\frac{m}{2^n}\right) \\
    &= - \sum_{i=1}^d \frac{\alpha_i}{2} \sum_{l=0}^{n/d-1} \left( 2^{-l-1+n/d} + \delta_{l0} \right) \\
    & \quad \times R_{l+(i-1)n/d}\left(\frac{m}{2^n}\right). 
\end{align*}
We therefore conclude that for a quasiperiodic system, the spectral function of the generator $ h $ is expressible as a linear combination of Rademacher functions. Explicitly, we have 
\begin{equation}
    \label{eqHWalsh}
    h = \sum_{i=1}^d \sum_{l=0}^{n/d-1} \hat h_{2^{l+(i-1)n/d}} R_{l+(i-1)n/d},
\end{equation}
with
\begin{displaymath}
    \hat h_{2^{l+(i-1)n/d}} = - \alpha_i( 2^{-l-1+n/d} + \delta_{l0} ) / 2,
\end{displaymath}
which is consistent with the factorization of the Hamiltonian $H_n$ in~\eqref{eqHDecompQuasiperiodic}.

\section{\label{appQFT}Approximate diagonalization of observables using the quantum Fourier transform}

In this appendix, we perform an analysis of approximate diagonalization of quantum mechanical observables induced at the quantum computational level from classical observables through the use of the QFT. In Appendices~\ref{appDiagD1} and~\ref{appDiagDD}, we describe how such quantum mechanical observables become increasingly diagonal as the number of qubits $n$ increases, and provide explicit bounds verifying the approximate eigenvalue equation~\eqref{eqApproxDiag}. In Appendix~\ref{appExpec}, we show that quantum mechanical expectation values of the approximately diagonalized observables converge to the true expectation values in a limit of infinite qubit number $n$ and vanishing RKHA parameter $\tau$. Appendices~\ref{appProof} and~\ref{appMeasBound} contain proofs of two auxiliary lemmas, Lemma~\ref{lemDiag} and~\ref{lemMeasBound}, which are stated in Appendices~\ref{appDiagD1} and~\ref{appExpec}, respectively.  

\subsection{\label{appDiagD1}Approximate diagonalization in dimension $d=1$}

We begin with the one-dimensional case, $d=1$, where $X= S^1$. In this case, the index set $J_n$ in~\eqref{eqJ} becomes $J_n= J_{n,1} = \{ -N/2, \ldots, -1, 1, \ldots, N/2 \}$ with $N=2^n$, and the map $\mathfrak F_{n,d}$ in~\eqref{eqQFTProd} reduces to the standard $n$-qubit QFT, $\mathfrak F_{n,d} \equiv \mathfrak F_n$. We also recall that $p \in (0,1)$ and $\tau > 0 $ are the parameters associated with the RKHA $\rkha$. 

\subsubsection{\label{appDiagRegRep} Diagonalization using the regular representation $\pi$}
Fixing $m \in \mathbb Z $, consider the regular representer (multiplication map) $ \pi \psi_m \in B(\rkha)$ of basis vector $\psi_m$,  the associated quantum computational observable 
\begin{displaymath}
    \hat A_{m,n} := (\mathcal W_n \circ \bm \Pi_n \circ \bm \Pi \circ \pi) \psi_m \in B(\mathbb B_n), 
\end{displaymath}
and the Fourier-transformed observable
\begin{equation}
    \label{eqTildeA}
    \tilde A_{m,n} = \bm{\mathfrak F}_n \hat A_{m,n} \equiv  \mathfrak F_n^* \hat A_{m,n} \mathfrak F_n.
\end{equation}
First, note that by definition of the projection $\bm \Pi_n$, $\hat A_{m,n}$ is the zero operator (and thus trivially diagonal) whenever $\lvert m \rvert > N/2$. This is a manifestation of an effective ``Nyquist limit'' on the wavenumber $m$ of classical observables that can be resolved by the finite-dimensional system on $\mathbb B_n$. Here, we are interested in characterizing the behavior of $\hat A_{m,n}$ in the ``well-resolved'' regime, $ \lvert m \rvert \ll N/2$. The following lemma provides a bound showing that (a) such well-resolved observables  $\hat A_{m,n}$ are approximately diagonal in the quantum computational basis $\{ \ket 0, \ldots, \ket{N -1} \}$; and (b) the diagonal part approximately recovers the values of $\psi_m$ at particular points on the circle $S^1$.

\begin{lem}
    With the notation of~\eqref{eqTildeA}, the observable $\tilde A_{m,n} $ satisfies
    \begin{displaymath}
        (\tilde A_{m,n})_{kl} := \bra k \tilde A_{m,n} \ket l  = \psi_m(\theta_l) \delta_{kl} + \varepsilon_{mnkl},
    \end{displaymath}
    where $\theta_l = 2 \pi l / N $, and $\varepsilon_{mnkl}$ is a residual obeying the bound
    \begin{displaymath}
        \lvert \varepsilon_{mnkl} \rvert \leq \frac{C \tau \lvert m \rvert}{N^{1-p}}  + \frac{(2\lvert m \rvert + 1) e^{-\tau \lvert m \rvert^p}}{N},  
    \end{displaymath}
    for a constant $C$ independent of $k$, $l$, $m$, $n$, $p$, and $\tau$.
    \label{lemDiag}
\end{lem}

A proof of Lemma~\ref{lemDiag} will be given in Appendix~\ref{appProof}. Using this basic result, we can derive error estimates for more general quantum mechanical observables than those induced by the individual basis functions $\psi_m$.

First, note that the terms $(2 \lvert m \rvert + 1) e^{-\tau \lvert m \rvert^p}$, $m \in \mathbb Z$, are bounded by a constant that depends on $p$ and $\tau$ (and diverges as either of these parameters tends to 0). Moreover, since $p>0$, $1/N$ is bounded by a constant times $1/N^{1-p}$. Thus, for every  $ p \in (0,1)$ and $\tau >0 $ there exists a constant $C_{p,\tau}$  such that for all $ m \in \mathbb Z$,
\begin{displaymath}
    \frac{(2\lvert m \rvert + 1) e^{-\tau \lvert m \rvert^p}}{N} \leq  \frac{C_{p,\tau}}{N^{1-p}}.   
\end{displaymath}
This means that we can simplify the estimate for $\lvert \varepsilon_{mnkl} \rvert$ in Lemma~\ref{lemDiag} to (the less precise) bound
\begin{equation}
    \lvert \varepsilon_{mnkl}\rvert \leq  \frac{C_{p,\tau} + C \tau \lvert m \rvert }{N^{1-p}}.
    \label{eqEpsEst}
\end{equation}
Using~\eqref{eqEpsEst}, we estimate the square norm of the residual 
\begin{displaymath}
    \ket{r_{mnl}} := \tilde A_{m,n} \ket l - \psi_m(\theta_l) \ket l 
\end{displaymath}
as 
\begin{align*}
    \lVert r_{mnl} \rVert^2_{\mathbb B_n} &= \lVert (\tilde A_{m,n} - \psi_m(\theta_l) \ket l) \rVert^2_{\mathbb B_n} \\
    &=  \sum_{k=0}^{N-1} \lvert \bra k \tilde A_{m,n} - \psi_m(\theta_l) I \ket l \rvert^2 \\
    &= \sum_{k=0}^{N-1} \lvert \varepsilon_{mnkl} \rvert^2 \\
    & \leq \sum_{k=0}^{N-1}  \frac{(C_{p,\tau} + C \tau \lvert m \rvert)^2 }{N^{2(1-p)}}\\
    & =  \frac{(C_{p,\tau} + C \tau \lvert m \rvert)^2 }{N^{1-2p}},
\end{align*}
giving
\begin{equation}
    \label{eqResid}
    \lVert r_{mnl} \rVert_{\mathbb B_n} \leq \frac{C_{p,\tau} + C \tau \lvert m \rvert}{N^{1/2-p}}.
\end{equation}
Thus, so long as $ p < 1/2$, the norm of the residual converges to zero as $n\to \infty$, uniformly with respect to $l \in \mathbb N_0 $. 

We next generalize to \emph{bandlimited} observables, i.e., observables $ f^{(M)} : X \to \mathbb C  $ for which there exists $ M \in \mathbb N$ such that $ f^{(M)} = \sum_{m=-M}^M \hat f_m \phi_m$, where the $\phi_m$ are the Fourier functions on $X$, and the $ \hat f_m $ are complex expansion coefficients. We denote the vector space of such bandlimited observables on $X$ by $\mathfrak B$. Note that $\mathfrak B$ is a dense subalgebra of $C(X)$, and is also a dense subalgebra of $\rkha$ for any $ \tau > 0 $ and $ p \in (0,1)$ (in the respective norms). In particular, viewed as an element of $\rkha$, $f^{(M)} = \sum_{m=-M}^m \hat f_m \phi_m$ can be equivalently expressed as $ f = \sum_{m=-M}^M \tilde f_m \psi_m$, where $ \tilde f_m = e^{\tau \lvert m \rvert^p/2} \hat f_m$. 

By linearity, every such observable $f^{(M)} \in \mathfrak B$ is represented at the quantum computational level by
\begin{displaymath}
    \hat A^{(M)}_n := (\mathcal W_n \circ \bm \Pi_n \circ \bm\Pi \circ \pi) f^{(M)} = \sum_{m=-M}^M \tilde f_m \hat A_{m,n}, 
\end{displaymath}
and after application of the QFT by 
\begin{equation}
    \label{eqABN}
    \tilde A^{(M)}_n = \bm{\mathfrak F_n} \hat A^{(M)}_n = \sum_{m=-M}^M \tilde f_m \tilde A_{m,n}.
\end{equation}
Thus, using Lemma~\ref{lemDiag} and~\eqref{eqResid}, we obtain
\begin{align*}
    \bra k \tilde A^{(M)}_n \ket l &= \sum_{m=-M}^M \tilde f_m \bra k \tilde A_{m,n} \ket l \\
    &= \sum_{m=-M}^M \tilde f_m \psi_m(\theta_l) \delta_{kl} + \sum_{m=-M}^M \tilde f_m \varepsilon_{mnkl} \\
    &= f^{(M)}(\theta_l) \delta_{kl} + \varepsilon^{(M)}_{nkl}, 
\end{align*}
where the residual $\varepsilon^{(M)}_{nkl} := \sum_{m=-M}^M \tilde f_m \varepsilon_{mnkl} $ can be estimated as 
\begin{align*}
    \lvert \varepsilon^{(M)}_{nkl} \rvert &= \left \lvert \sum_{m=-M}^M \tilde f_m \varepsilon_{mnkl} \right \rvert \\
    &\leq \left( \sum_{m=-M}^M \lvert \tilde f_m \rvert^2 \right)^{1/2} \left( \sum_{m=-M}^M \lvert \varepsilon_{mnkl} \rvert^2 \right)^{1/2} \\
    & \leq \lVert f^{(M)} \rVert_{\rkha} \frac{1}{N^{1-p}} \left( \sum_{m=-M}^M (C_{p,\tau} + C\tau \lvert m \rvert)^2 \right)^{1/2}.
\end{align*}
We thus conclude that for bandlimited observables the off-diagonal residual $\varepsilon^{(M)}_{nkl}$ vanishes as $n \to \infty$ at fixed $p$ and $\tau$, uniformly with respect to $k,l \in \mathbb N_0$. For later convenience, we set 
\begin{displaymath}
    C_{p,\tau,M}^2 = \sum_{m=-M}^M (C_{p,\tau} + C \tau \lvert m \rvert)^2,
\end{displaymath}
so that
\begin{equation}
    \label{eqResidBandlimited}
    \lvert \varepsilon^{(M)}_{nkl} \rvert \leq \frac{ C_{p,\tau,M}}{N^{1-p}}\lVert f \rVert_{\rkha}.   
\end{equation}
Analogously to~\eqref{eqResid} we can bound the norm of the residual $\ket{r^{(M)}_{nl}} := \tilde A^{(M)}_n \ket l - f^{(M)}(\theta_l)\ket l$ as
\begin{equation}
    \label{eqResidBandlimitedR}
    \lVert r^{(M)}_{nl} \rVert_{\mathbb B_n} = \left( \sum_{k=0}^{N-1} \lvert \varepsilon^{(M)}_{nkl} \rvert^2 \right)^{1/2} \leq \lVert f^{(M)} \rVert_\rkha \frac{C_{p,\tau,M}}{N^{1/2-p}},
\end{equation}
and we deduce that the residual vanishes as $n\to \infty$ if $ p < 1/2$.

Suppose now that $f= \sum_{m=-\infty}^\infty \tilde f_m \psi_m \in \rkha $ is not bandlimited. Then, for any  $\epsilon > 0 $ there exists $M \in \mathbb N_0 $ such that the bandlimited observable $f^{(M)} := \sum_{m=-M}^M \tilde f_m \psi_m \in \mathfrak B$ satisfies 
\begin{equation}
    \label{eqFEps}
    \lVert f - f^{(M)} \rVert_\rkha < \epsilon. 
\end{equation}
Defining
\begin{equation}
    \label{eqAB}
    \tilde A_n := (\bm{\mathfrak F}_n \circ \mathcal W_n \circ \bm \Pi_n \circ \bm\Pi \circ \pi) f 
\end{equation}
and $\tilde A^{(M)}_n$ by~\eqref{eqABN}, we get
\begin{align*}
    \lvert \bra k \tilde A_n \ket l - f(\theta_l) \delta_{kl} \rvert &= \lvert \bra k ( \tilde A_n - \tilde A^{(M)}_n ) \ket l \\
    & \qquad - ( f(\theta_l) - f^{(M)}(\theta_l) ) \delta_{kl} \\
    & \qquad + \bra k \tilde A^{(M)}_n \ket l - f^{(M)}(\theta_l) \delta_{kl} \rvert \\
    & \leq \lvert \bra k ( \tilde A_n - \tilde A^{(M)}_n ) \ket l \rvert \\
    &  \qquad + \lvert f(\theta_l) - f^{(M)}(\theta_l) \rvert \\
    &  \qquad + \lvert \bra k \tilde A^{(M)}_n \ket l - f^{(M)}(\theta_l) \delta_{kl} \rvert.
\end{align*}
To bound the terms in the right-hand side of the last inequality, note first that the operators $\pi: \rkha \to B(\rkha)$, $ \bm \Pi : B(\rkha) \to B(\mathcal H)$, $ \bm \Pi_n : B(\mathcal H ) \to B(\mathcal H_n)$, $\mathcal W_n : B(\mathcal H_n) \to B(\mathbb B_n)$, and $\bm{\mathfrak F}_n : B(\mathbb B_n) \to B(\mathbb B_n) $ all have unit norm. Using this fact, it follows that
\begin{multline*}
    \lvert\bra k \tilde A_n - \tilde A^{(M)}_n \ket l \rvert\\
    \begin{aligned}
        &= \lvert\bra k (\mathcal W_n \circ \bm \Pi_n \circ \bm\Pi \circ \pi)(f - f^{(M)})\ket l \rvert \\
        &\leq \lVert (\mathcal W_n \circ \bm \Pi_n \circ \bm\Pi \circ \pi)(f - f^{(M)}) \rVert_{\mathbb B_n} \\
        & \leq \lVert \mathcal W_n \rVert \lVert \bm \Pi_n \rVert \lVert \bm\Pi \rVert \lVert \pi \rVert \lVert f - f^{(M)} \rVert_\rkha \\
        & < \epsilon \lVert f - f^{(M)} \rVert_\rkha. 
    \end{aligned}
\end{multline*}
Moreover, it follows from the reproducing property of $\rkha$ that
\begin{align*}
    \lvert f(\theta_l) - f^{(M)}(\theta_l) \rvert &= \lvert \langle k_{\theta_l}, f - f^{(M)} \rangle_\rkha \rvert \\
    & \leq \lVert k_{\theta_l} \rVert_\rkha \lVert f - f^{(M)} \rVert_\rkha\\
    & < \kappa \epsilon. 
\end{align*}
Using these bounds and~\eqref{eqResidBandlimited}, we obtain
\begin{align*}
    \lvert \bra k \tilde A_n \ket l - f(\theta_l) \delta_{kl} \rvert & \leq \epsilon (1 + \kappa) \lVert f \rVert_\rkha \\
    & \quad + \frac{C_{p,\tau,M}}{N^{1-p}}  \lVert f^{(M)} \rVert_\rkha \\
    &\leq \left( (1 + \kappa) \epsilon + \frac{C_{p,\tau,M}}{N^{1-p}} \right) \lVert f \rVert_\rkha.
\end{align*}
In particular, for large-enough $N$ we have 
\begin{displaymath}
    \frac{C_{p,\tau,M}}{N^{1-p}} < \epsilon,
\end{displaymath}
and thus
\begin{equation}
    \label{eqResidGeneral}
    \lvert \bra k \tilde A_n \ket l - f(\theta_l) \delta_{kl} \rvert \leq \epsilon (2 + \kappa ) \lVert f \rVert_\rkha. 
\end{equation}
Since $\epsilon$ was arbitrary, we conclude that as $n \to \infty$, $ \lvert \bra k \tilde A_n \ket l - f(\theta_l) \delta_{kl} \rvert $ converges to 0, i.e., the matrix elements of the quantum mechanical observable $ \tilde A_n$ are consistently approximated by the matrix elements of the diagonal observable associated with the values $f(\theta_l)$. Note that unlike the bandlimited case we do not have an explicit rate for this convergence. 

Consider now the residual 
\begin{equation}
    \label{eqResidGeneral2}\ket{r_{nl}} = \tilde A_n \ket l - f(\theta_l) \ket l.
\end{equation}
In order to examine the asymptotic behavior of $\ket{r_{nl}} $ as $n\to\infty$, it is useful to view the spaces $\mathbb B_n$ as a nested family of subspaces of the sequence space $\ell^2$, i.e., $\mathbb B_1 \subset \mathbb B_2 \subset \cdots \subset \ell^2$. With this identification, $\{ \ket 0, \ket 1, \ldots \} $ is an orthonormal basis of $\ell^2$, and $\ket{r_{1l}}, \ket{r_{2l}}, \ldots$ is a bounded sequence in $\ell^2$. According to~\eqref{eqResidGeneral}, for any $k \in \mathbb N_0$, this sequence satisfies 
\begin{displaymath}
    \lim_{n\to\infty} \langle k | r_{nl} \rangle_{\mathbb B_n} = 0.
\end{displaymath}
It then follows from standard Hilbert space results that as $n\to\infty$, $\ket{ r_{nl}} $ converges to zero in the weak topology of $\ell^2$. That is, for any $ u \in \ell^2$, we have 
\begin{equation}
    \label{eqWeakConv}
    \lim_{n\to\infty}\langle u_n | r_{nl} \rangle_{\mathbb B_n} \to 0, 
\end{equation}
where $u_n $ is the orthogonal projection of $u$ onto $\mathbb B_n$.

In summary, in dimension $d=1$, the residual $\ket{r_{nl}}$ from~\eqref{eqResidGeneral2} converges weakly to zero as $n\to\infty$ for any $ f \in \rkha$. Moreover, if $f$ is bandlimited, the convergence is strong (i.e., the residual norm vanishes) with a rate of convergence estimated by~\eqref{eqResidBandlimitedR}.

\subsubsection{Diagonalization using the self-adjoint representation $\tilde T$}

Using the estimates obtained in Appendix~\ref{appDiagRegRep}, we now derive approximate diagonalization results for the self-adjoint observables induced by the map $\tilde T : \rkha \to B(\rkha)$ in~\eqref{eqT}. For any $ f \in \rkha$, consider the self-adjoint observable $\tilde S_n \in B(\mathbb B_n)$ with
\begin{equation}
    \label{eqSnDiag}
    \tilde S_n = (\bm{\mathfrak F}_n \circ \mathcal W_n \circ \bm \Pi_n \circ \bm\Pi \circ \tilde T ) f \equiv \frac{ \tilde A_n + \tilde A_n^*}{2},
\end{equation}
where $\tilde A_n$ is defined in~\eqref{eqAB}. Then, we have 
\begin{multline*}
    \lvert \bra k \tilde S_n \ket l - \Real f(\theta_l) \rvert \\
    \begin{aligned}
        &= \frac{1}{2}\lvert \bra k \tilde A_n \ket l - f(\theta_l) \delta_{kl} + \bra k \tilde A_n^* \ket l - f^*(\theta_l) \delta_{kl} \rvert \\
        &\leq \frac{1}{2}\lvert \bra k \tilde A_n \ket l - f(\theta_l) \delta_{kl} \rvert + \lvert \bra k \tilde A_n^* \ket l - f^*(\theta_l) \delta_{kl} \rvert \\
        &= \frac{1}{2}\lvert \bra k \tilde A_n \ket l - f(\theta_l) \delta_{kl} \rvert + \lvert ( \bra l \tilde A_n \ket k - f(\theta_l) \delta_{kl} )^* \rvert \\ 
        &= \frac{1}{2}\lvert \bra k \tilde A_n \ket l - f(\theta_l) \delta_{kl} \rvert + \lvert \bra l \tilde A_n \ket k - f(\theta_k) \delta_{lk}  \rvert, \\
    \end{aligned}
\end{multline*}
and we can use the results of Appendix~\ref{appDiagRegRep} to bound the two terms in the last line. In particular, if $ f = \sum_{m=-M}^M \tilde f_m \psi_m$ is bandlimited, then it follows from~\eqref{eqResidBandlimited} that
\begin{displaymath}
    \lvert \bra k \tilde S_n \ket l - \Real f(\theta_l) \rvert \leq \frac{\lvert \varepsilon^{(M)}_{nkl} \rvert + \lvert \varepsilon^{(M)}_{nlk} \rvert}{2} \leq \frac{C_{p,\tau,M}}{N^{1-p}} \lVert f \rVert_{\rkha},
\end{displaymath}
and for general $f \in \rkha$,  
\begin{displaymath}
    \lvert \bra k \tilde A_n \ket l - \Real f(\theta_l) \delta_{kl} \rvert \leq \epsilon (2 + \kappa ) \lVert f \rVert_\rkha, 
\end{displaymath}
with the same notation as~\eqref{eqResidGeneral}. Moreover, convergence results for the residual $\tilde S_n \ket l - \Real f(\theta_l) \ket l $ can be derived analogously to those for $\ket{r_{nl}}$ in Appendix~\ref{appDiagRegRep}.

\subsection{\label{appDiagDD}Approximate diagonalization in dimension $d>1$}

We can extend the results in Appendix~\ref{appDiagD1} to dimension $ d > 1$ by taking advantage of the tensor product structure of the RKHA $\rkha$ on $\mathbb T^d$ and the maps effecting the transformations from the classical to the quantum computational level. Following the notation of Sec.~\ref{secRKHA}, we will use $(1)$-superscripts to distinguish vector spaces, vectors, and linear maps associated with the circle $S^1$; see, e.g.,~\eqref{eqRKHAProd}. With this notation, the representation map  $\pi : \rkha \to B(\rkha)$ for dimension $d$ decomposes as $ \pi = \bigotimes_{i=1}^d \pi^{(1)}$, and similarly we have $\bm \Pi : B(\rkha) \to B(\mathcal H)$, $\bm \Pi_n : B(\mathcal H) \to B(\mathcal H_n)$, and $ \mathcal W_n : B( \mathcal H_n) \to B(\mathbb B_n)$ with $ \bm \Pi = \bigotimes_{i=1}^d \bm \Pi^{(1)}$, $\bm \Pi_n = \bigotimes_{i=1}^d \bm \Pi^{(1)}_{n/d}$, and $\mathcal W_n = \bigotimes_{i=1}^d \mathcal W^{(1)}_{n/d}$, where we have assumed that the number of qubits $n$ is an integer multiple of $d$. We also recall the definition of the tensor product QFT operator $\bm{\mathfrak F}_{n,d} : B(\mathbb B_n) \to B(\mathbb B_n)$ in~\eqref{eqQFTInd}, i.e.,
\begin{displaymath}
    \bm{\mathfrak F}_{n,d} A := \mathfrak F_{n,d} A \mathfrak F_{n,d}^* \equiv \left( \bigotimes_{i=1}^d \bm{\mathfrak F}_{n/d} \right) A.
\end{displaymath}

Given any tensor product element $ f = \bigotimes_{i=1}^d f^{(i)} \in \rkha $, we have 
\begin{displaymath}
    \tilde A_n := (\bm{\mathfrak F}_{n,d} \circ \mathcal W_n \circ \bm \Pi_n \circ \bm \Pi \circ \pi) f = \bigotimes_{i=1}^d \tilde A_n^{(i)}, 
\end{displaymath}
where
\begin{displaymath}
    \tilde A_n^{(i)} = (\bm{\mathfrak F}_{n/d} \circ \mathcal W^{(1)}_{n/d} \circ \bm \Pi^{(1)}_{n/d} \circ \bm \Pi^{(1)} \circ \pi^{(1)} )f^{(i)}
\end{displaymath}
Meanwhile, for any binary string $ \bm b = (\bm b^{(1)}, \ldots, \bm b^{(d)} ) \in \{ 0, 1 \}^n$ with associated evaluation point $ x_{\bm b} \in \mathbb T^d$ from~\eqref{eqGrid} we have 
\begin{displaymath}
    f(x_{\bm b}) = \prod_{i=1}^d f(\theta_{\bm b^{(i)}}).
\end{displaymath}
Thus, for any two computational basis vectors $\ket{\bm a}$ and $\ket{\bm b}$ of $\mathbb B_n$ with $ \bm a = (\bm a^{(1)}, \ldots, \bm a^{(d)} ) $ and $ \bm b = (\bm b^{(1)}, \ldots, \bm b^{(d)} ) $ we have
\begin{multline*}
    \lvert \bra{\bm a} \tilde A_n \ket{\bm b} - f(x_{\bm b}) \delta_{\bm a \bm b} \rvert \\
    = \prod_{i=1}^d \left\lvert \bra{\bm a^{(i)}} \tilde A_n \ket{\bm b^{(i)}} - f^{(i)}(\theta_{\bm b^{(i)}}) \delta_{\bm a^{(i)}\bm b^{(i)}} \right\rvert,    
\end{multline*}
and we can use the results of Appendix~\ref{appDiagD1} to bound the right-hand side. In particular, it follows from~\eqref{eqResidGeneral} that $ \lvert \bra{\bm a} \tilde A_n \ket{\bm b} - f(x_{\bm b}) \delta_{\bm a \bm b} \rvert $ converges to 0 as $n\to \infty$, so that $ \tilde A_n $ is consistently approximated by a diagonal observable with eigenvalues equal to the values of $f$ at the points $\bm x_{\bm b}$. Moreover, the residual is $O(N^{1-p})$ analogously to~\eqref{eqResidBandlimitedR} if $f$ is bandlimited, and converges weakly to zero as $n$ increases in the sense of~\eqref{eqWeakConv}. 

The extension to elements of $\rkha$ which are not of tensor product form follows by linearity. We omit the details of these calculations in the interest of brevity. 

Note now that for every $f \in \rkha$, the spectrum of the corresponding multiplication operator $\pi f$ consists precisely of the range of values of $f$, i.e., $\sigma(\pi f) = \ran f$ \cite{DasGiannakis20b}. In particular since the elements of $\rkha$ are all continuous functions, $\pi f$ has nonempty continuous spectrum, unless $f$ is constant.  
Define $D_n : \mathbb B_n \to \mathbb B_n $ and $E_n : \mathbb B_n \to \mathbb B_n$ as the diagonal operators satisfying  
\begin{equation}
    \label{eqDOp}
    D_n \ket{\bm b}  = f(x_{\bm b}) \ket{\bm b},  \quad E_n \ket{\bm b} = \Real f(x_{\bm b}) \ket{\bm b},
\end{equation}
where $E_n$ is self-adjoint. The following theorem summarizes the properties of the quantum computational observables approximating $\pi f $ and $ \tilde T f $ obtained in Appendices~\ref{appDiagD1} and~\ref{appDiagDD}.

\begin{thm}
    \label{thmSpec} Let $f \in \rkha$ be arbitrary, and consider the operators $\tilde A_n$ and $\tilde S_n$ defined as in~\eqref{eqAB} and~\eqref{eqSnDiag} for dimension $ d \geq 1$. Consider also the diagonal operators in~\eqref{eqDOp}. Then, the following hold as $n \to \infty $.  
    \begin{enumerate}[a.]
        \item The matrix elements $\bra k \tilde A_n \ket l $ of $\tilde A_n $ converge to the matrix elements $\bra k D_n \ket l = f(x_{l} ) \delta_{kl} $ of $D_n$. 
        \item The matrix elements $\bra k \tilde S_n \ket l $ of $\tilde S_n$ converge to the matrix elements $\bra k E_n \ket l = \Real f(x_{l} ) \delta_{kl} $ of $E_n$.   
        \item For each basis vector $\ket l $, the residuals $ ( \tilde A_n - \tilde E_n ) \ket l $ and $ ( \tilde S_n - \tilde E_n ) \ket l $ converge to zero weakly. Moreover, if $f$ is bandlimited, the convergence is strong and the norms of the residuals are $O(N^{1-p})$.   

        \item For every element $z \in \ran f $ there exists a sequence of eigenvalues $z_n$ of $D_n$ and a sequence of eigenvalues $u_n$ of $E_n$ such that $ z = \lim_{n\to \infty} z_n $ and $\Real z = \lim_{n\to\infty} u_n$.  
    \end{enumerate}
\end{thm}

The approximate diagonalization result in~\eqref{eqApproxDiag} is a consequence of Theorem~\ref{thmSpec}. 

\subsection{\label{appExpec}Convergence of quantum mechanical expectations}

Thus far, we have established that every element $f$ of $\rkha $ can be consistently approximated in a spectral sense by operators  $D_n \in B( \mathbb B_n )$ which are diagonal in the computational basis. By construction (see Theorem~\ref{thmSpec}) the spectra of $D_n$ are subsets of the range of values of $f$.  As a result, quantum measurement of $D_n$ (which can be equivalently realized by measurement of the PVM associated with the computational basis as described in Sec.~\ref{secApproximateMeas}) yields outcomes consistent with values that $f$ takes on classical states in $X$. While this is a desirable property to have, it does not in itself guarantee that the quantum mechanical measurements are consistent with the value of $f$ on the particular classical state that the system happen to have. Establishing this type of consistency is the goal of this appendix. 

The convergence results that we derive will turn out to hold for a \emph{decreasing sequence} of RKHA parameters $\tau$, as opposed to \emph{fixed} $\tau$ values in Appendices~\ref{appDiagD1} and~\ref{appDiagDD}.  Thus, in what follows, we will use the notation $\rkha_\tau \equiv \rkha $ to make the dependence of the RKHAs on $ \tau >0 $ explicit. By construction, the spaces $\rkha_\tau$ form an increasing nested family as $\tau$ decreases to 0; that is, for every $ 0 < \tau < \tau'$ and $ f \in \rkha $ we have $\rkha_\tau \subset \rkha_{\tau'}$ and $ \lVert f \rVert_{\rkha_\tau} \geq \lVert f \rVert_{\rkha_{\tau'}} $. We also introduce explicit $\tau$ subscripts in our notation for the RKHSs $\mathcal H_\tau \subset \rkha_\tau$ and $\mathcal H_{\tau,n} \subset \mathcal H_\tau$ and the operators $ L_\tau : \rkha_\tau \to \rkha_\tau$,  $\pi_\tau : \rkha_\tau \to B(\rkha_\tau)$, $ \bm \Pi_\tau : B(\rkha_\tau) \to B(\mathcal H_\tau)$, and $\bm \Pi_{\tau,n} : B(\mathcal H_\tau) \to B(\mathcal H_{\tau,n})$. $\tau$ subscripts will also be introduced in our notation for elements of $\rkha_\tau$, $\mathcal H_\tau$, and the associated operator spaces as appropriate.    

As in Appendices~\ref{appDiagD1} and~\ref{appDiagDD}, we consider first the one-dimensional case, $d=1$, and an observable $ f = \psi_{m,\tau}$ equal to a basis vector of $\rkha_\tau$. We define the diagonal operator $D_{m,\tau,n}: \mathbb B_n \to \mathbb B_n$ with
\begin{displaymath}
    D_{m,\tau,n} \ket l = \psi_{m,\tau}(\theta_l) \ket l
\end{displaymath}
analogously to~\eqref{eqDOp}, and also set $\tilde D_{m,\tau,n} \in B(\mathcal H_{\tau,n})$ with
\begin{displaymath}
    \tilde D_{m,\tau,n} = (\mathcal W_n^* \circ \bm{\mathfrak F}_n^*) D_{m,\tau,n} = W_n^* \mathfrak F_n D_{m,\tau,n} \mathfrak F_n^* W_n.
\end{displaymath}
We also define 
\begin{equation}
    \label{eqAMN}
    \tilde A_{m,\tau,n} = (\bm{\mathfrak F}_n \circ \mathcal W_n \circ \bm \Pi_{\tau,n} \circ \bm \Pi_\tau \circ \pi_\tau) \psi_{m,\tau} \in B(\mathbb B_n)
\end{equation}
as in~\eqref{eqTildeA}. For any $ x \in X = S^1$, we consider the quantum computational state $\hat \rho_{x,\tau,n} = \hat{\mathcal F}_{\tau,n}(x) \in Q(\mathbb B_n)$ and the state $\tilde \rho_{x,\tau,n} \in Q(\mathbb B_n)$ after application of the QFT,
\begin{equation}
    \label{eqTildeRho}
    \tilde \rho_{x,\tau,n} = \bm{\mathfrak F}_n \hat \rho_{x,\tau,n} = (\bm{\mathfrak F}_n \circ \mathcal W_n) \rho_{x,\tau,n}.
\end{equation}
We then have:

\begin{lem}
    \label{lemMeasBound}
    With notation as above, the $n\to\infty$ limit of the expected difference $\langle \tilde A_{m,\tau,n} - D_{m,\tau,n} \rangle_{\tilde\rho_{x,\tau,n}}$ between measurements of $\tilde A_{m,\tau,n}$ and $D_{m,\tau,n}$ on the state $\tilde \rho_{x,\tau,n}$ exists, and satisfies
    \begin{displaymath}
        \lim_{n\to\infty} \lvert \langle \tilde A_{m,\tau,n} - D_{m,\tau,n} \rangle_{\tilde\rho_{x,\tau,n}} \rvert \leq  1 - e^{-\tau \lvert m \rvert^p/2}.
    \end{displaymath}
\end{lem}

A proof of Lemma~\ref{lemMeasBound} can be found in Appendix~\ref{appMeasBound}. For our purposes, a key implication of the result is that while the bias in measuring $D_{m,\tau,n}$ (instead of $\tilde A_{m,\tau,n}$) need not vanish as $n \to\infty$, it can be made arbitrarily small for a suitable choice of $\tau$. In particular, for any $\epsilon > 0 $ there exists $ \tau_m > 0 $ such that for all $\tau \in (0,\tau_m)$ we have $ 1 - e^{-\tau \lvert m \rvert^p / 2} < \epsilon$, and thus  
\begin{equation}
    \label{eqMeasBound}
    \lim_{n\to\infty} \lvert \langle \tilde A_{m,\tau,n} - D_{m,\tau,n} \rangle_{\tilde\rho_{x,\tau,n}} \rvert < \epsilon.
\end{equation}
Since $ 1 - e^{-\tau\lvert m \rvert^p/2} \leq \tau \lvert m \rvert^p / 2 $, the choice $ \tau_m = 2 \epsilon \lvert m \rvert^{-p}$ will suffice for~\eqref{eqMeasBound} to hold. 

Next, we consider bandlimited observables $ f^{(M)} \in \mathfrak B $ of the form  $ f^{(M)} = \sum_{m=-M}^M \tilde f_{m,\tau} \psi_{m,\tau}$. Let   
\begin{align}
    \nonumber\tilde A^{(M)}_{\tau,n} &= (\bm{\mathfrak F_n} \circ \mathcal W_n \circ \bm \Pi_{\tau,n} \circ \bm \Pi_\tau \circ \pi_\tau) f \\
    \label{eqAN}&= \sum_{m=-M}^M \tilde f_{m,\tau} \tilde A_{m,\tau,n}  \in B(\mathbb B_n)
\end{align}
be the corresponding quantum computational observable, and let $D_{\tau,n} \in B(\mathbb B_n) $ be the diagonal observable approximating $\tilde A_{\tau,n}$,
\begin{equation}
    \label{eqDN}
    D_{\tau,n} \ket l = f(\theta_l) \ket l, \quad D_{\tau,n} = \sum_{m=-M}^M \tilde f_{m,\tau} D_{m,\tau,n}. 
\end{equation}
Using Lemma~\ref{lemMeasBound} and following a similar approach as in Appendix~\ref{appDiagD1}, we find
\begin{equation}
    \label{eqMeasBound1}
    \lim_{n\to\infty} \lvert \langle \tilde A_{\tau,n} - D_{\tau,n} \rangle_{\tilde\rho_{x,\tau,n}} \rvert \leq C_{p,\tau,M} \lVert f \rVert_{\rkha_\tau},  
\end{equation}
where
\begin{displaymath}
    C_{p,\tau,M}^2 = \sum_{m=-M}^M \left( 1-e^{-\tau \lvert m \rvert^p / 2} \right)^2. 
\end{displaymath}
Again, for any $\epsilon > 0 $, there exists $\tau_M > 0 $ such that 
\begin{equation}
    \label{eqMeasBound2}
    \lim_{n\to\infty} \lvert \langle \tilde A_{\tau,n} - D_{\tau,n} \rangle_{\tilde\rho_{x,\tau,n}} \rvert < \epsilon \lVert f \rVert_{\rkha_\tau}, \quad \forall \tau \in (0, \tau_M ).   
\end{equation}
In this case, the choice $\tau_M = 2 \epsilon \lvert M \rvert^{-(p+\frac{1}{2})}$ is sufficient for the bound to hold.

To generalize to non-bandlimited observables, we must take into account the fact that the error bounds in~\eqref{eqMeasBound} and~\eqref{eqMeasBound2} imply convergence on a decreasing sequence of RKHA parameters $\tau$, as opposed to the diagonalization results in Appendix~\ref{appDiagD1} which hold for fixed $\tau$. With that in mind, we consider a space of classical observables that contains the RKHAs $\rkha_\tau$ for all admissible values of the parameters $\tau$ and $p$. In particular, we consider observables in the \emph{Wiener algebra} of $X$, i.e., the space of functions $ f : X \to \mathbb C$ with absolutely convergent Fourier series, which we denote here by $\mathfrak W$. The Wiener algebra $\mathfrak W$ is a dense subalgebra of $C(X)$. Moreover, the RKHAs $\rkha_\tau$ employed in this work are all dense subalgebras of $\mathfrak W$. Thus, we have the following relationships between algebras of classical observables (which also hold in dimension $d>1$):
\begin{displaymath}
    \mathfrak B \subset \rkha_\tau \subset \mathfrak W \subset C(X).
\end{displaymath}

Suppose then that $ f = \sum_{m=-\infty}^\infty \hat f_m \phi_m $ is an arbitrary element of $\mathfrak W$, where the sum over $m$ converges uniformly on $X$. Then, for any $\epsilon > 0 $ there exists $M_* \in \mathbb N$ such that for every $M > M_*$ the bandlimited observable $ f^{(M)} = \sum_{m=-M}^M \hat f_m \phi_m \in \mathfrak B$ satisfies
\begin{equation}
    \label{eqMeasBound3}
    \lVert f - f^{(M)} \rVert_{C(X)} < \epsilon/3.
\end{equation}
The bandlimited observable $f^{(M)}$ is an element of $\rkha_\tau$ for any $ \tau > 0 $, with RKHA norm satisfying
\begin{align}
    \nonumber \lVert f^{(M)}  \rVert_{\rkha_\tau} &= \left( \sum_{m=-M}^M e^{\tau \lvert m \rvert^p} \lvert \hat f_m \rvert^2 \right)^{1/2} \\
    \nonumber &\leq e^{\tau M^p / 2} \left( \sum_{m=-M}^M  \lvert \hat f_m \rvert^2 \right)^{1/2} \\
    \nonumber &\leq e^{\tau M^p / 2} \sum_{m=-M}^M  \lvert \hat f_m \rvert \\
    \nonumber &= e^{\tau M^p / 2} \sum_{m=-M}^M \left \lvert \frac{1}{2\pi} \int_0^{2\pi} e^{-im \theta} f(\theta) \, d\theta \ \right \rvert \\ 
    \nonumber &\leq e^{\tau M^p / 2} \sum_{m=-M}^M \lVert f \rVert_{C(X)} \\
    \label{eqMeasBound4}&= (2M + 1) e^{\tau M^p/2} \lVert f \rVert_{C(X)}.
\end{align}
We will also need the observable
\begin{displaymath}
    f^{(M)}_\tau = L_\tau f^{(M)} = \kappa_\tau \sum_{m=-M}^M \frac{\hat f_{m}}{\eta_{m,\tau}} \hat \phi_m
\end{displaymath}
as an intermediate approximation associated with the bias correction introduced in Sec.~\ref{secConsistency} and Appendix~\ref{appConsistency} to take into account the projection from $\rkha_\tau$ to $\mathcal H_\tau$. Here, $L_\tau$ is operator introduced in~\eqref{eqLOp} and $\eta_{m,\tau}$ are its eigenvalues, where we have again used $ \tau $ subscripts to make dependencies on that parameter explicit. We have
\begin{align*}
    \lVert f^{(M)} - f^{(M)}_\tau \rVert_{C(X)} &= \left \lVert \sum_{m=-M}^M \left( \frac{\eta_{\tau,m}}{\kappa_\tau} - 1 \right) \hat f_m \phi_m \right\rVert_{C(X)} \\
    & \leq C_\tau \sum_{m=-M}^M \lvert \hat f_m \rvert,
\end{align*}
where
\begin{displaymath}
    C_\tau = \max_{m\in[-M,M]} \left\lvert \frac{\eta_{\tau,m}}{\kappa_\tau} -1    \right\rvert = \frac{e^{-\tau}}{\kappa_\tau}.
\end{displaymath}
Note that to obtain the last result we used the fact that $\eta_{\tau,m}$ lies in the interval $[e^{-\tau},\kappa_\tau]$; see Appendix~\ref{appConsistency}. In particular, as $\tau \to 0 $, $C_\tau$ converges to 0 since $e^{-\tau}$ converges to 1 and $\kappa_\tau$ tends to infinity. Proceeding as in the derivation of~\eqref{eqMeasBound4} to bound the sum $\sum_{m=-M}^M \lvert \hat f_m \rvert$, we arrive at   
\begin{equation}
    \label{eqMeasBound4b}
    \lVert f^{(M)} - f^{(M)}_\tau \rVert_{C(X)} \leq C_\tau (2M+1) \rVert f \rVert_{C(X)}.
\end{equation}

Next, define the quantum computational observable $\tilde A_{\tau,n}^{(M)}  \in B(\mathbb B_n)$ as
\begin{align}
    \nonumber\tilde A_{\tau,n}^{(M)} &= (\bm{\mathfrak F_n} \circ \mathcal W_{\tau,n} \circ \bm \Pi_{\tau,n} \circ \bm \Pi_\tau \circ \pi_\tau) f^{(M)} \\
    \label{eqATN}&= \kappa_\tau\sum_{m=-M}^M \frac{\tilde f_{m,\tau}}{\eta_{m,\tau}} \tilde A_{m,\tau,n} 
\end{align}
where $\tilde f_{m,\tau} = e^{\tau\lvert m \rvert^p/2} \hat f_m$, and $\tilde A_{m,\tau,n} $ are operators defined as in~\eqref{eqAMN}. Define also the diagonal observable  
\begin{align}
    \label{eqDTN}
    D_{\tau,n}^{(M)} \ket l &= f^{(M)}(\theta_l) \ket l 
     = \sum_{m=-M}^M \tilde f_{m,\tau} \psi_{m,\tau}(\theta_l) \ket l.
\end{align}
Letting $x$ be an arbitrary point in $X$, defining the quantum state $\tilde \rho_{x,\tau,n} \in Q(\mathbb B_n)$ as in~\eqref{eqTildeRho}, and using~\eqref{eqMeasBound3} and~\eqref{eqMeasBound4b} we get 
\begin{multline*}
    \lvert f(x) - \langle D_{\tau,n}^{(M)} \rangle_{\tilde \rho_{x,\tau,n}} \rvert \\
    \begin{aligned}
        &= \lvert f(x) - f^{(M)}(x) + f^{(M)}(x) - f^{(M)}_\tau(x) + f^{(M)}_\tau(x) \\
        & \qquad - \langle \tilde A_{\tau,n}^{(M)} \rangle_{\tilde \rho_{x,\tau,n}}+ \langle \tilde A_{\tau,n}^{(M)} \rangle_{\tilde \rho_{x,\tau,n}} - \langle D_{\tau,n}^{(M)} \rangle_{\tilde \rho_{x,\tau,n}} \rvert  \\
        &\leq \lvert f(x) - f^{(M)}(x) \rvert + \lvert f^{(M)}(x) - f_\tau^{(M)}(x) \rvert \\
        & \qquad + \lvert f^{(M)}_\tau(x) -  \langle A_{\tau,n}^{(M)} \rangle_{\tilde \rho_{x,\tau,n}} \rvert + \lvert  \langle  \tilde A_{\tau,n}^{(M)} - D_{\tau,n}^{(M)} \rangle_{\tilde \rho_{x,\tau,n}} \rvert \\
        & \leq \lVert f - f^{(M)}  \rVert_{C(X)} + \leq \lVert f^{(M)} - f_\tau^{(M)}  \rVert_{C(X)}  \\
        & \qquad + \lvert f^{(M)}(x) - \langle \tilde A_{\tau,n}^{(M)} \rangle_{\tilde \rho_{x,\tau,n}} \rvert + \lvert \langle  \tilde A_{\tau,n}^{(M)} - D_{\tau,n}^{(M)} \rangle_{\tilde \rho_{x,\tau,n}} \rvert \\
        &<  \frac{\epsilon}{3} +  C_\tau (2M+1) \lVert f \rVert_{C(X)} + \lvert f^{(M)}(x)  - \langle \tilde A_{\tau,n}^{(M)} \rangle_{\tilde \rho_{x,\tau,n}} \rvert \\
        &\qquad + \lvert  \langle  \tilde A_{\tau,n}^{(M)} - D_{\tau,n}^{(M)} \rangle_{\tilde \rho_{x,\tau,n}} \rvert.
    \end{aligned}
\end{multline*}

We can now bound the second, third, and fourth terms in the right-hand side of the last inequality. In particular, it follows by applying Proposition~\ref{propExpec} to the observable $f_\tau^{(M)}$ that
\begin{displaymath}
    \lim_{n\to\infty}\lvert f^{(M)}_\tau(x)  - \langle \tilde A_{\tau,n}^{(M)} \rangle_{\tilde \rho_{x,\tau,n}} \rvert = 0,
\end{displaymath}
and from \eqref{eqMeasBound1} and~\eqref{eqMeasBound4} that
\begin{multline*}
    \lim_{n\to\infty} \lvert \langle \tilde A_{\tau,n}^{(M)} - D_{\tau,n}^{(M)} \rangle_{\tilde \rho_{x,\tau,n}} \rvert \\
    \begin{aligned}
        & \leq C_{p,\tau,M} \lVert f^{(M)} \rVert_{\rkha_\tau} \\
        &\leq C_{p,\tau,M} (2M+1) e^{\tau M^p/2} \lVert f \rVert_{C(X)}.
    \end{aligned}
\end{multline*}
Then, using the above in conjunction with the fact that $\lim_{\tau \to 0} C_\tau = 0$, it follows that  for any $M \in \mathbb N$ there exists $ \tau_M > 0 $ such that for all $\tau \in (0, \tau_M)$ we have, simultaneously,
\begin{equation}
    \label{eqMeasBound5} 
    \begin{cases}
        C_\tau(2M+1) \lVert f \rVert_{C(X)} < \epsilon / 3, \\ 
        C_{p,\tau,M} (2M+1) e^{\tau M^p/2} \lVert f \rVert_{C(X)} < \epsilon / 3, 
    \end{cases}
\end{equation}
and thus 
\begin{displaymath}
    \lim_{n\to\infty} \lvert f(x) - \langle D_{\tau,n}^{(M)} \rangle_{\tilde \rho_{x,\tau,n}} \rvert < \frac{\epsilon}{3} + \frac{\epsilon}{3} + 0 + \frac{\epsilon}{3} = \epsilon.  
\end{displaymath}
Since $\epsilon$ was arbitrary, we conclude that there exists a decreasing sequence of RKHA parameters $\tau_M$ such that the quantum mechanical expectation $\langle D_{\tau_M,n}^{(M)} \rangle_{\tilde \rho_{x,\tau_M,n}}$ converges to the classical value $f(x)$ in the iterated limit of $M\to \infty$ (infinite bandwidth) after $n\to\infty$ (infinite qubits), and the convergence is uniform with respect to $x \in X$. 

Having established this convergence result in dimension $d=1$, we can extend it to higher dimensions using tensor product arguments analogous to those in Appendix~\ref{appDiagDD}. It is also straightforward to derive analogous results using the symmetrized map $\tilde T_\tau : \rkha_\tau \to B(\mathcal H_\tau)$, inducing the self-adjoint quantum computational observable (cf.~\eqref{eqATN}) 
\begin{displaymath}
    \tilde S_{\tau,n}^{(M)} = (\bm{\mathfrak F_n} \circ \mathcal W_{\tau,n} \circ \bm \Pi_{\tau,n} \circ \bm \Pi_\tau \circ \tilde T_\tau) f^{(M)} 
\end{displaymath}
and the diagonal observable
\begin{equation}
    \label{eqETN}
    E_{\tau,n}^{(M)} \ket{\bm x_l} = \Real f^{(M)}(\bm x_l) \ket{\bm x_l}.
\end{equation}
We do not reproduce the details of these analyses in the interest of brevity. The following theorem summarizes the asymptotic convergence of our approach in these settings.

\begin{thm}
    \label{thmConv}
    Let $f = \sum_{m\in\mathbb Z^d} \hat f_m \phi_m $ be a classical observable in the Wiener algebra $\mathfrak W$ of $X = \mathbb T^d$. For $M \in \mathbb N $, $\tau>0$, and $n\in \mathbb N $, define the bandlimited observable $f^{(M)} = \sum_{\lvert m \rvert \leq M } \hat f_m \phi_m$ and the corresponding  diagonal quantum mechanical observables $D_{\tau,n}^{(M)}$ and $E_{\tau,n}^{(M)}$ from~\eqref{eqDTN} and~\eqref{eqETN}, respectively. Then, there exists a sequence $\tau_1, \tau_2, \ldots$, decreasing to 0, such that for any $x \in X$,
    \begin{align*}
        \lim_{M\to\infty}\lim_{n\to\infty} \langle D_{\tau_M,n}^{(M)} \rangle_{\tilde \rho_{x,\tau,n}} &= f(x), \\
        \lim_{M\to\infty}\lim_{n\to\infty} \langle E_{\tau_M,n}^{(M)} \rangle_{\tilde \rho_{x,\tau,n}} &= \Real f(x),
    \end{align*}
    uniformly with respect to $x\in X$. 
\end{thm}

The fact that quantum states at the quantum computational level evolve compatibly with the underlying classical dynamics, i.e., $\hat \Psi^t_n(\hat \rho_{x,\tau,n}) = \hat\rho_{\Phi^t(x),n}$, leads to the following corollary of Theorem~\ref{thmConv}, which establishes the asymptotic consistency of QECD in simulating the evolution of classical observables.
\begin{cor}
    \label{corConv}With the notation of Theorem~\ref{thmConv} and for any $t \geq 0$, let $\tilde f^{(t)}_{M,n} \in C(X)$ with
    \begin{displaymath}
        \tilde f^{(t)}_{M,n}(x) = \langle D_{\tau_M,n}^{(M)} \rangle_{\tilde \rho_{\Phi^t(x),\tau,n}},
    \end{displaymath}
    be the function representing the expected value of the time-$t$ simulation of $f$ by the quantum computer, given initial conditions $x$. Then, 
    \begin{displaymath}
        \lim_{M\to\infty}\lim_{n\to\infty} \tilde f^{(t)}_{M,n}(x)= U^t f(x). 
    \end{displaymath}
    where the convergence is uniform with respect to $x \in X $ and $t$ in compact sets. Moreover, if $f$ is real-valued, the analogous result holds for 
    \begin{displaymath}
        \tilde f^{(t)}_{M,n}(x) = \langle E_{\tau_M,n}^{(M)} \rangle_{\tilde \rho_{\Phi^t(x),\tau,n}}.
    \end{displaymath}
\end{cor}

Before closing this section, we note that while the convergence results in Theorem~\ref{thmConv} and Corollary~\ref{corConv} hold for observables in the Wiener algebra $\mathfrak W$ with absolutely convergent Fourier series, the fact that $\mathfrak W$ is a dense subspace of $C(X)$ means that any observable $f \in C(X) $ can be approximated to arbitrarily high precision in uniform norm by an observable $g \in \mathfrak W$, whose dynamical evolution can in turn be simulated to arbitrarily high precision using the quantum compiler as established in Corollary~\ref{corConv}. The function $g$ may be constructed by several means available from signal processing, e.g., by convolution of $f$ by an appropriate smoothing kernel. A detailed study of this topic is beyond the scope of the present work. 

\subsection{\label{appProof}Proof of Lemma~\ref{lemDiag}}

Using the definition of the map $W_n$ in~\eqref{eqWn} and the QFT in~\eqref{eqQFT}, we get 
\begin{align*}
    W_n^* \mathfrak F_n \ket l &= W_n^* \left( \frac{1}{\sqrt N} \sum_{q=0}^{N-1} e^{-2\pi i lq/N} \ket q \right) \\
    &=  \frac{1}{\sqrt N} \sum_{q=0}^{N-1} e^{-2\pi i lq/N} \psi_{o^{-1}(q)}  \\
    &= \frac{1}{\sqrt N} \sum_{j\in J_n} e^{-2\pi i lo(j)/N} \psi_j,
\end{align*}
leading to
\begin{align*}
    (\pi \psi_m) R^*_n W_n^* \ket l &=  \frac{1}{\sqrt N} \sum_{j\in J_n} e^{-2\pi i lo(j)/N} (\pi \psi_m )\psi_j\\
    &=  \frac{1}{\sqrt N} \sum_{j\in J_n} e^{-2\pi i lo(j)/N} \psi_m \psi_j \\
    &=  \frac{1}{\sqrt N} \sum_{j\in J_n} e^{-2\pi i lo(j)/N} c_{mj} \psi_{m+j}.
\end{align*}
Therefore, the operator $\tilde A_{m,n}$ has the matrix elements
\begin{align*}
    (\tilde A_{m,n})_{kl} &= \bra k \tilde A_{m,n} \ket l \\
    &= \bra k \mathfrak F_n^* W_n \Pi_n (\pi \psi_m) \Pi_n^* W_n^* \mathfrak F_n \ket l \\
    &= \left\langle \Pi_n^* W_n^* \mathfrak F_n  k, (\pi \psi_m)\Pi_n^* W_n^* \mathfrak F_n l \right\rangle_{\rkha} \\
    &= \left \langle \frac{1}{\sqrt N} \sum_{j' \in J_n} e^{- 2 \pi i k o(j')/ N} \psi_{j'}, \right. \\
    & \qquad \left. \frac{1}{\sqrt N} \sum_{j\in J_n} e^{-2\pi i lo(j)/N} c_{mj} \psi_{m+j} \right\rangle_{\rkha} \\
    &= \frac{1}{N} \sum_{j',j \in J_n} e^{2 \pi i (ko(j')-lo(j))/N} c_{mj} \delta_{j',m+j}\\ 
    &= \frac{1}{N} \sum_{j\in J_n} e^{2\pi i (ko(m+j)-lo(j))} c^{(n)}_{mj},
\end{align*}
where
\begin{equation*}
    c^{(n)}_{mj} = 
    \begin{cases}
        c_{mj}, & m+j \in J_n, \\
        0, & \text{otherwise}.
    \end{cases}
\end{equation*}

Observe now that if $m+j \in J_n$, then $o(m+j) = m + o(j)$. Therefore, since $c^{(n)}_{mj} = 0 $ whenever $ m+j \notin J_n $, we get
\begin{displaymath}
    (\tilde A_{m,n})_{kl} = \frac{1}{N} \sum_{j\in J_n} e^{2\pi i((k-l)o(j)+km)/N} c_{mj}^{(n)}.
\end{displaymath}
Thus, defining 
\begin{align*}
    \tilde c_{mj}^{(n)} &= c_{mj}^{(n)} - e^{-\tau \lvert m \rvert^p/2} \\
    &= 
    \begin{cases}
        e^{\tau(\lvert m \rvert^p + \lvert j \rvert^p - \lvert m+j \rvert^p)/2} - e^{-\tau \lvert m \rvert^p/2}, & m+j \in J_n, \\
        - e^{-\tau \lvert m \rvert^p / 2 }, & \text{otherwise}
    \end{cases}
\end{align*}
and 
\begin{displaymath}
    \varepsilon_{mnkl} = \frac{1}{N} \sum_{j\in J_n} e^{2\pi i((k-l)o(j)+km)/N} \tilde c_{mj}^{(n)},
\end{displaymath}
we get
\begin{align*}
    (\tilde A_{m,n})_{kl} &= \frac{1}{N} \sum_{j\in J_n} e^{2\pi i((k-l)o(j)+km)/N} e^{-\tau \lvert m \rvert^p / 2} + \varepsilon_{kl} \\
    &= \frac{1}{N} \sum_{q=0}^{N-1} e^{2\pi i((k-l)q+km)/N} e^{-\tau \lvert m \rvert^p / 2} + \varepsilon_{mnkl} \\
    &= e^{2\pi i km / N} e^{-\tau \lvert m \rvert^p / 2} \delta_{kl} + \varepsilon_{mnkl}. 
\end{align*}
Note that we used standard properties of discrete Fourier transforms to arrive at the last line. It then follows by definition of the $\psi_m$ basis vectors and $\theta_l$ gridpoints that 
\begin{displaymath}
    (\tilde A_{m,n})_{kl} = \psi_m(\theta_l)\delta_{kl} + \varepsilon_{mnkl},
\end{displaymath}
as claimed in the statement of the lemma.

We now proceed to bound the remainder $\varepsilon_{mnkl}$, assuming, for now, that $m \geq 0 $. Letting $\tilde N = N / 2$, we have 
\begin{align}
    \nonumber \lvert \varepsilon_{mnkl} \rvert &= \left\lvert \frac{1}{N} \sum_{j\in J_n} e^{2\pi i ( (k-l) o(j) + km)/N} \tilde c_{mj}^{(n)} \right\rvert \\
    \nonumber & \leq \frac{1}{N} \sum_{j\in J_n} \tilde c_{mj}^{(n)} \\
    \nonumber & = \frac{1}{N} \sum_{j=-\tilde N}^{-m} c_{mj}^{(n)} + \frac{1}{N} \sum_{j=-m+1}^{-1} e^{-\tau\lvert m \rvert^p} \\
    \nonumber & \quad + \frac{1}{N} \sum_{j=1}^{\tilde N - m } c_{mj}^{(n)} + \frac{1}{N} \sum_{\tilde N - m + 1}^{\tilde N} e^{-\tau \lvert m \rvert^p} \\
    \label{eqVarepsilon}&= \frac{(2 \lvert m \rvert + 1)e^{-\tau\lvert m \rvert^p}}{N} + \varepsilon_- + \varepsilon_+,  
\end{align}
where
\begin{displaymath}
    \varepsilon_- = \frac{1}{N} \sum_{j=-\tilde N}^{-m} c_{mj}^{(n)}, \quad  \varepsilon_+ = \frac{1}{N} \sum_{j=1}^{\tilde N - m } c_{mj}^{(n)}.
\end{displaymath}
Next, to bound the $\varepsilon_+$ term, consider the function $f(u) = u^p$. Since $p\in (0,1)$, $f $ is strictly concave on the positive real line. Thus, for $ m \geq 0 $ and $j \geq 1$, we have
\begin{align}
    \nonumber \lvert m+j \rvert^p - \lvert j \rvert^p &= \lvert f(m+j) - f(j) \rvert \\
    \nonumber & \leq \lvert f'(j) \rvert \lvert m \rvert \\
    \label{eqCBound1}&= p j^{p-1} \lvert m \rvert. 
\end{align}
Consider also the function $g(u) = e^{\tau u / 2} - 1 $ on the interval $ u \in [ 0, u_\text{max} ]$ with $ u_\text{max} = pm $. The function $g$ is strictly convex, so 
\begin{displaymath}
    g(u) \leq g'(u_\text{max}) u = \frac{\tau}{2} e^{\tau u_\text{max}/2} u = \frac{\tau}{2}e^{\tau pm/2} u.
\end{displaymath}
Therefore, for $m \geq 0 $ and $j \geq 1$, we obtain  
\begin{equation}
    \label{eqCBound2}
    \tilde c_{mj}^{(n)} = e^{-\tau \lvert m \rvert^p/2} g(f(m+j)-f(j)) \leq \tau p \lvert m \rvert j^{p-1}/2.
\end{equation}
Note that we have used~\eqref{eqCBound1} and the fact that $f(m+j) -f(j) \leq pm $ (which follows from the same equation). 

Next, let $a_{\tilde N}$ be the series
\begin{displaymath}
    a_{\tilde N} = \sum_{j=1}^{\tilde N} \left( \frac{j}{\tilde N} \right)^{p-1} \frac{1}{\tilde N}.
\end{displaymath}
As $\tilde N \to \infty$,  $a_{\tilde N}$ converges to the integral $ \int_0^1 u^{p-1} \, du = 1 / p$. Therefore, $a_{\tilde N}$ is bounded by a constant, $\tilde C$, leading to the bound
\begin{equation}
    \frac{1}{\tilde N} \sum_{j=1}^{\tilde N} j^{p-1} = \tilde N^{p-1} \tilde a_N \leq \tilde C \tilde N^{p-1}.
    \label{eqCBound3}
\end{equation}
Using~\eqref{eqCBound2} and~\eqref{eqCBound3}, we thus obtain
\begin{align}
    \nonumber \varepsilon_+ &:= \left\lvert \frac{1}{2 \tilde N}  \sum_{j=1}^{\tilde N} e^{2\pi i((k-l)o(j)+km)/N} \tilde c_{mj}^{(n)} \right\rvert \\
    \nonumber & \leq \frac{1}{2\tilde N} \sum_{j=1}^{\tilde N} \tilde c_{mj}^{(n)} \\
    \label{eqCBound4} & \leq \tilde C \tau p \lvert m \rvert \tilde N^{p-1} / 2.
\end{align}
Moreover, analogous arguments for $ j \leq - 1$ lead to the estimate    
\begin{align}
    \nonumber \varepsilon_- &:= \left\lvert \frac{1}{2 \tilde N}  \sum_{j=-\tilde N}^{-1} e^{2\pi i((k-l)o(j)+km)/N} \tilde c_{mj}^{(n)} \right\rvert \\
    \label{eqCBound5} &\leq \hat C \tau p \lvert m \rvert \tilde N^{p-1}/2
\end{align}
for a constant $\hat C$. 

Substituting~\eqref{eqCBound4} and~\eqref{eqCBound5} into~\eqref{eqVarepsilon}, it follows that 
\begin{align*}
    \lvert \varepsilon_{mnkl} \rvert &\leq \frac{(2 \lvert m \rvert + 1)e^{-\tau\lvert m \rvert^p}}{N} + \varepsilon_+ + \varepsilon_- \\
    & \leq \frac{(2 \lvert m \rvert + 1)e^{-\tau\lvert m \rvert^p}}{N} + \frac{C \tau p \lvert m \rvert}{ N^{1-p}} 
\end{align*}
with $ C = \min\{ \tilde C, \hat C \} $, which verifies the claim of the lemma for $ m \geq 0$. However, since $\psi_{-m} = \psi_m^*$, repeating the calculation described above for $m < 0 $ leads to the same bound, so we conclude that the claim holds for any $ m \in \mathbb Z$. \qed  

\subsection{\label{appMeasBound}Proof of Lemma~\ref{lemMeasBound}}

We have
\begin{multline*}
    \langle \tilde A_{m,\tau,n} - D_{m,\tau,n} \rangle_{\tilde\rho_{x,\tau,n}} \\
    \begin{aligned}
        &= \tr(\tilde \rho_{x,\tau,n} (\tilde A_{m,\tau,n} - D_{m,\tau,n})) \\
        &= \tr(\rho_{x,\tau,n}( \bm \Pi_{\tau,n}(\bm \Pi(\pi \psi_{m,\tau})) - \tilde D_{m,\tau,n})).
    \end{aligned}
\end{multline*}
By the results in Sec.~\ref{secMatrixMechanical} and Appendix~\ref{appQuantumRep}, it follows that
\begin{multline}
    \label{eqALimit}
    \lim_{n\to\infty} \tr(\rho_{x,\tau,n} \bm \Pi_{\tau,n}(\bm \Pi_\tau(\pi_\tau \psi_{m,\tau}))) \\ 
    \begin{aligned}
        &=\tr(\rho_{x,\tau} \bm \Pi_\tau(\pi_\tau \psi_{m,\tau})) \\
        &= \frac{\eta_{m,\tau}}{\kappa_\tau} \psi_{m,\tau}(x) \\
        &= \frac{\sum_{j\in J'_m}e^{-\tau\lvert j\rvert^p}}{\kappa_\tau}\psi_{m,\tau}(x), 
    \end{aligned}
\end{multline}
where we recall the definition of the index set $J'_m$, 
\begin{displaymath}
    J'_m = \{ j \in J : j + m \in J \}.
\end{displaymath}
Moreover, we have
\begin{align*}
    \tr(\rho_{x,\tau,n} \tilde D_{m,\tau,n}) &= \langle \xi_{x,\tau,n}, \tilde D_{m,\tau,n} \xi_{x,\tau,n} \rangle_{\mathcal H_\tau} \\
    &= \frac{\langle k_{x,\tau,n}, \tilde D_{m,\tau,n}) k_{x,\tau,n} \rangle_{\mathcal H_\tau}}{\kappa_{\tau,n}} \\
    &= \frac{(\tilde D_{m,\tau,n} k_{x,\tau,n})(x)}{\kappa_{\tau,n}}.
\end{align*}
In the above, the function $\tilde D_{m,\tau,n} k_{x,\tau,n} \in \mathcal H_{\tau,n}$ can be expressed as
\begin{multline*}
    D_{m,\tau,n} k_{x,\tau,n} \\
    \begin{aligned}
        &= W_n^* \mathfrak F_n D_{m,\tau,n} \mathfrak F_n^* W_n \left( \sum_{j \in J_n} \psi_{j,\tau}^*(x) \psi_{j,\tau} \right) \\ 
        &= W_n^* \mathfrak F_n D_{m,\tau,n} \mathfrak F_n^* \left( \sum_{j \in J_n} \psi_{j,\tau}^*(x) \ket{o(j)} \right) \\ 
        &= W_n^* \mathfrak F_n D_{m,\tau,n} \left( \frac{1}{\sqrt N} \sum_{l=0}^{N-1}  \sum_{j \in J_n} \psi_{j,\tau}^*(x) e^{2\pi i o(j) l / N} \ket l \right) \\ 
        &= W_n^* \mathfrak F_n \left( \frac{1}{\sqrt N} \sum_{l=0}^{N-1}  \sum_{j \in J_n} \psi_{j,\tau}^*(x) e^{2\pi i o(j) l / N} \psi_{m,\tau}(\theta_l) \ket l \right) \\ 
        &= W_n^* \left( \frac{1}{N} \sum_{k,l=0}^{N-1}  \sum_{j \in J_n} \psi_{j,\tau}^*(x) e^{2\pi i (o(j) - k) l / N} \psi_{m,\tau}(\theta_l) \ket k \right) \\ 
        &= \frac{1}{N} \sum_{k,l=0}^{N-1}  \sum_{j \in J_n} \psi_{j,\tau}^*(x) e^{2\pi i (o(j) - k) l / N} \psi_{m,\tau}(\theta_l) \psi_{o^{-1}(k)} \\   
        &= \sum_{j,j' \in J_n} \psi^*_{j,\tau}(x)\psi_{j',\tau} \\
        & \qquad \times \left( \frac{1}{N} \sum_{l=0}^{N-1} e^{i(o(j)-o(j'))(2\pi l / N)} \psi_{m,\tau}(2\pi l / N) \right).
    \end{aligned}
\end{multline*}
As $n\to\infty$, the summation in the parentheses in the last line converges to a continuous Fourier transform,
\begin{multline*}
    \lim_{n\to\infty} \frac{1}{N}\sum_{l=0}^{N-1}  e^{i(o(j)-o(j'))(2\pi l / N)} \psi_{m,\tau}(2\pi l / N) \\ 
    \begin{aligned}
        &= \int_{S^1} e^{- i (j-j') \theta} \psi_{m,\tau}(\theta) \, d\theta \\
        &= e^{-\tau \lvert m \rvert^p/2}  \int_{S^1} e^{ i (j-j'+ m) \theta}  \, d\theta \\
        &= e^{-\tau\lvert m \rvert^p/2} \delta_{j', j+m}.
    \end{aligned}
\end{multline*}
As a result, we have 
\begin{multline*}
    \lim_{n\to\infty} \frac{\tilde D_{m,\tau,n} k_{x,\tau,n}}{\kappa_{\tau,n}} \\  
    \begin{aligned}
        &= \frac{1}{\kappa_\tau} \sum_{j,j' \in J} \psi_{j,\tau}^*(x) \psi_{j'} e^{-\tau\lvert m \rvert^p/2} \delta_{j',j+m} \\
        &= \frac{1}{\kappa_\tau} \sum_{j\in J'_m} \psi^*_{j,\tau}(x)\psi_{j+m} e^{-\tau\lvert m \rvert^p/2}\\
        &= \frac{1}{\kappa_\tau} \sum_{j\in J'_m} \psi^*_{j,\tau}(x) e^{-\tau\lvert j + m \rvert^p/2}\phi_{j +m} e^{-\tau\lvert m \rvert^p/2}\\
        &= \frac{1}{\kappa_\tau} \sum_{j\in J'_m} \psi^*_{j,\tau}(x) \psi_{j,\tau} e^{-\tau(\lvert j + m \rvert^p - \lvert j \rvert^p )/2} \psi_{m,\tau},
    \end{aligned}
\end{multline*} 
and upon evaluation at $x$,
\begin{multline}
    \lim_{n\to\infty} \frac{(\tilde D_{m,\tau,n} k_{x,\tau,n})(x)}{\kappa_{\tau,n}} \\
    \label{eqDLimit}= \frac{1}{\kappa_\tau} \sum_{j\in J'_m} e^{-\tau \lvert j \rvert^p} e^{-\tau(\lvert j + m \rvert^p - \lvert j \rvert^p )/2} \psi_{m,\tau}(x).
\end{multline}
Therefore, combining~\eqref{eqALimit} and~\eqref{eqDLimit}, we obtain
\begin{multline*}
    \lim_{n\to\infty} \langle \tilde A_{m,\tau,n} - D_{m,\tau,n} \rangle_{\tilde\rho_{x,\tau,n}} \\
    = \frac{1}{\kappa_\tau} \sum_{j\in J'_m} e^{-\tau \lvert j \rvert^p} \left( 1-  e^{-\tau(\lvert j + m \rvert^p - \lvert j \rvert^p )/2}  \right) \psi_{m,\tau}(x),
\end{multline*}
and thus
\begin{multline*}
    \lim_{n\to\infty} \lvert \langle \tilde A_{m,\tau,n} - D_{m,\tau,n} \rangle_{\tilde\rho_{x,\tau,n}} \rvert \\
    \begin{aligned}
        &= \frac{e^{-\tau \lvert m \rvert^p/2}}{\kappa_\tau} \left \lvert \sum_{j\in J'_m} e^{-\tau \lvert j \rvert^p} \left( 1-  e^{-\tau(\lvert j + m \rvert^p - \lvert j \rvert^p )/2}  \right) \right\rvert \\
        &\leq \frac{1}{\kappa_\tau} \sum_{j\in J'_m} e^{-\tau \lvert j \rvert^p} \left \lvert 1-  e^{-\tau(\lvert j + m \rvert^p - \lvert j \rvert^p )/2} \right\rvert. 
    \end{aligned}
\end{multline*}

Note now that for fixed $m \in \mathbb Z$, the largest value of $ e^{-\tau \lvert m \rvert^p/2}\lvert 1-  e^{-\tau(\lvert j + m \rvert^p - \lvert j \rvert^p )/2} \rvert$ over $ j \in \mathbb Z$ occurs for $\lvert j \rvert = \lvert m \rvert $. That is, we have
\begin{multline*}
    e^{-\tau \lvert m \rvert^p/2}\left \lvert 1-  e^{-\tau(\lvert j + m \rvert^p - \lvert j \rvert^p )/2} \right\rvert \\
    \begin{aligned}
        &\leq e^{-\tau\lvert m \rvert^p/2}\max \left\{  \left\lvert 1-  e^{\tau \lvert m \rvert^p/2} \right\rvert,  \left\lvert 1-  e^{-\tau \lvert m \rvert^p/2} \right\rvert \right\} \\   
        &= \max \left\{ \left\lvert e^{-\tau\lvert m \rvert^p/2} - 1 \right\rvert,  e^{-\tau\lvert m \rvert^p/2} \left\lvert 1-  e^{-\tau \lvert m \rvert^p/2} \right\rvert \right\} \\   
        &\leq \max \left\{ \left\lvert e^{-\tau\lvert m \rvert^p/2} - 1 \right\rvert,   \left\lvert 1-  e^{-\tau \lvert m \rvert^p/2} \right\rvert \right\} \\   
        &= 1-e^{-\tau\lvert m \rvert^p/2},
    \end{aligned}
\end{multline*}
so that 
\begin{multline*}
    \lim_{n\to\infty} \lvert \langle \tilde A_{m,\tau,n} - D_{m,\tau,n} \rangle_{\tilde\rho_{x,\tau,n}} \rvert\\
    \begin{aligned}
        &\leq \frac{1}{\kappa_\tau} \sum_{j\in J'_m} e^{-\tau \lvert j \rvert^p} \left( 1- e^{-\tau\lvert m \rvert^p/2} \right) \\
        &= \frac{\eta_{m,\tau}}{\kappa_\tau} \left( 1- e^{-\tau\lvert m \rvert^p/2} \right) \\
        &\leq 1- e^{-\tau\lvert m \rvert^p/2},
    \end{aligned}
\end{multline*}
proving the lemma. \qed

%\bibliography{bibliography_dg,bibliography_js}

\end{document}